\documentclass{article}
\usepackage{amssymb,amsmath,amsthm,enumerate} 
\usepackage{braket}

\renewcommand\bra[1]{\mathinner{\langle{#1}\vert}}
\renewcommand\ket[1]{\mathinner{\vert{#1}\rangle}}

\usepackage{mathtools} 
\usepackage{xfrac} 
\usepackage{a4wide}
\usepackage{pgf,tikz}
\usepackage{epsfig}
\usepackage{graphicx}

\usetikzlibrary{matrix,graphs,arrows,positioning,calc,decorations.markings,shapes.symbols,shapes.geometric,shapes.misc}
\usepackage{calc, ifthen, ae, xstring, etoolbox, xkeyval,pgfplots}
\usepackage[pdftex,bookmarks,colorlinks,breaklinks]{hyperref} 
\definecolor{dullmagenta}{rgb}{0.4,0,0.4} 
\definecolor{darkblue}{rgb}{0,0,0.4}
\hypersetup{linkcolor=red,citecolor=blue,filecolor=dullmagenta,urlcolor=darkblue}
\usepackage{eucal}
\usepackage{upgreek} 
\newcommand{\Pain}[1]{\text{P}_{\mathrm{#1}}}

\newcommand{\Pro}{\mathbb{P}}

\newcommand{\N}{\mathbb{N}}

\newcommand{\Z}{\mathbb{Z}}
\newcommand{\C}{\mathbb{C}}
\newcommand{\R}{\mathbb{R}}

\newcommand{\rb}{\mathrm{b}}
\newcommand{\bb}{\mathbf{b}}

\newcommand{\ru}{\mathrm{u}}
\newcommand{\rv}{\mathrm{v}}
\newcommand{\rp}{\mathrm{p}}
\newcommand{\roq}{\mathrm{q}}

\newcommand{\rd}{\mathrm{d}}
\newcommand{\re}{\mathrm{e}}
\newcommand{\ri}{\mathrm{i}}

\newcommand\la{{\lambda}}

\newcommand\ka{{\kappa}}
\newcommand\al{{\alpha}}
\newcommand\be{{\beta}}
\newcommand\ze{{\zeta}}

\newcommand\gam{{\gamma}}
\newcommand\de{{\delta}}
\newcommand\om{{\omega}}

\newcommand\tal{\tilde{\al}}
\newcommand\tbe{\tilde{\beta}}
\newcommand\tgam{\tilde{\gamma}}

\newcommand\bv{\mathbf{v}}
\newcommand\bx{\mathbf{x}}

\newcommand\bu{\mathbf{u}}

\newcommand\eps{{\epsilon}} 

\usepackage{epstopdf}
\usepackage{newcent}
\usepackage{color}

\def\ds{\displaystyle}

\def\PV{\mbox{\rm P$_{\rm V}$}}

\def\ep{\varepsilon}
\def\ph{\varphi}

\def\p{Painlev\'e}

\def\peq{\p\ equation}
\def\peqs{\p\ equations}
\def\bk{B\"ack\-lund}
\def\bt{B\"ack\-lund transformation}
\def\bts{B\"ack\-lund transformations}
\def\cdot{{\,\scriptstyle\bullet\,}}
\def\hide#1{}

\newcommand{\deriv}[3][]{\frac{\rd^{#1}{#2}}{{\rd{#3}}^{#1}}}

\newcommand{\comment}[1]{}
\def\^#1{\widehat{#1}}
\def\~#1{\widetilde{#1}}
\def\pms#1#2#3{(\al,\be,\gam)=\big(#1,#2,#3\big)}
\usepackage{youngtab}

\usepackage{hyperref}

\usepackage{authblk}

\usepackage{subcaption}

\newcommand{\beq}{\begin{equation}} 
\newcommand{\eeq}{\end{equation}} 
\newtheorem{thm}{Theorem}[section] 

\newtheorem{propn}[thm]{Proposition} 

\newtheorem{conje}[thm]{Conjecture}

\newtheorem{lem}[thm]{Lemma}

\newtheorem{cor}[thm]{Corollary} 

\newenvironment{prf}{\trivlist \item [\hskip 
\labelsep \textbf{Proof:}]\ignorespaces}{\qed \endtrivlist}

\theoremstyle{definition}
\newtheorem{definition}[thm]{Definition}
\newtheorem{remark}[thm]{Remark}

\newcommand{\poch}[2]{\left(#1 \right)_{#2}}
\newcommand{\Sum}[2]{\sum\limits_{#1}^{#2}}
\newcommand{\ddiff}[2]{\frac{\mathrm{d}{#1}}{\mathrm{d}{#2}}}

\newcommand{\pdiff}[2]{\frac{\partial #1}{\partial #2}}
\newcommand{\sqbr}[1]{\left[#1\right]}
\newcommand{\vbr}[1]{\left|#1\right|}
\newcommand{\bbr}[1]{\Big|#1\Big|}
\newcommand{\sbr}[1]{\left(#1\right)}
\newcommand{\cbr}[1]{\left\{#1\right\}}
\newcommand{\nBesselI}[2]{I_{#1}\!\left(#2\right)} 
\newcommand{\nBesselK}[2]{K_{#1}\!\left(#2\right)} 
\newcommand{\BesselZ}[2]{\CMcal{Z}_{#1}\!\left(#2\right)}
\newcommand{\BesselDet}[4]{\mathcal{B}_{#1,#2,#3}\!\left(#4\right)}
\newcommand{\nKummerM}[3]{M\! \left(#1 , #2 , #3\right)} 

\newcommand{\Prod}[2]{\prod\limits_{#1}^{#2}}
\newcommand{\EXP}[1]{\exp\!\left(#1\right)}
\newcommand{\SIN}[1]{\sin\!\left(#1\right)}

\newcommand{\bn}[2]{\Big[\substack{#1\\[2pt]#2}\Big]}

\newcommand{\PB}{\left\{\cdot\,,\cdot\right\}}

\newcommand{\dep}[1]{\deriv{#1}{\eps}}

\title{Special solutions of a discrete Painlev\'e equation for quantum minimal surfaces}

\author[1]{Peter A. Clarkson} 
\author[2]{Anton Dzhamay}
\author[1]{Andrew N.W. Hone\footnote{Corresponding author e-mail: A.N.W.Hone@kent.ac.uk}} 
\author[1]{Ben Mitchell}

\affil[1]{School of Engineering, Mathematics and Physics \protect\\ 
University of Kent,
Canterbury, CT2 7NF, UK
}

\affil[2]{Beijing Institute of Mathematical Sciences and Applications \protect\\ 
No. 544, Hefangkou Village, Huaibei Town, Huairou District \protect\\
Beijing 101408, 
China
}

\date{}

\begin{document} 

\maketitle
 
\begin{abstract} 
We consider solutions of a discrete \p\ equation arising 
from a construction of quantum minimal surfaces by Arnlind, Hoppe and Kontsevich, and 
in earlier work of Cornalba and Taylor on static membranes. 
While the discrete 
equation admits a continuum limit to the continuous \p\ I equation, we find that it has the same space of initial values as the \p\ V equation with certain specific parameter values. We further explicitly show how 
each iteration of this discrete \p\ I equation corresponds to a certain composition of B\"acklund transformations for \p\ V, as was first remarked in work by Tokihiro, 
Grammaticos and Ramani. 
In addition, we show that some explicit special function solutions of \p\ V, written in terms of modified Bessel functions, yield the unique positive solution of the initial value problem 
required for quantum minimal surfaces. 
\end{abstract} 

\section{Introduction} 

Minimal surfaces can be characterised as maps $\bx:\,\Sigma\rightarrow \R^d$ that 
extremise the Schild functional 
\beq\label{schild}
S[\bx] = 
\int_\Sigma \sum_{j<k} \{x_j,x_k\}^2 \, \om, 
\eeq 
where $\Sigma$ is a surface with symplectic form $\om$ and associated 
Poisson bracket $\PB$, 
and $(x_j)_{j=1,\ldots,d}$ are coordinates 
on $\R^d$. The Euler-Lagrange equations obtained from the action $S$ are 
\beq\label{el} 
\sum_{j=1}^d \{x_j,\{x_j,x_k\}\} =0, \qquad k=1,\ldots, d. 
\eeq 
In this context, quantization is achieved by replacing the 
classical observables $x_j$ with 
self-adjoint operators $X_j$ acting on a Hilbert space $\mathcal{H}$, 
and taking the commutator in place of the Poisson bracket. Hence, following \cite{ach}, 
one can say that a quantum minimal surface is a collection of such operators 
satisfying the relations 
\beq\label{qms} 
\sum_{j=1}^d [X_j,[X_j,X_k]]=0, \qquad k=1,\ldots, d. 
\eeq 
As a set of matrix equations, the system \eqref{qms} has previously appeared 
in string theory, 
as a large-$N$ matrix model \cite{ikkt}, or static membrane equation \cite{ct}. 

For the case of minimal surfaces in $\R^4\cong\C^2$, it is a classical result \cite{eisenhart} 
that an arbitrary analytic function $f$, and the plane curve associated with its graph, defines a solution of \eqref{el}, by 
setting 
\beq\label{curve} 
z_2 = f(z_1), \qquad z_1=x_1+\ri x_2, \quad z_2 = x_3+\ri x_4; 
\eeq
and more generally one can consider a Riemann surface defined by an arbitrary analytic relation $F(z_1,z_2)=0$. 
The latter relation between the complex coordinates $z_1,z_2$ implies that 
\beq\label{zpb}
\{z_1,z_2\}=0, 
\eeq
while imposing the requirement of constant curvature gives the equation 
\beq\label{curva}
\{\bar{z}_1,z_1\}+ \{\bar{z}_2,z_2\}=\ri \ka, 
\eeq 
where, up to rescaling, 
$\ka\in\R$ is the curvature. 
The real and imaginary parts of \eqref{zpb}, together with the equation \eqref{curva}, provide three linear relations between the 
brackets $\{x_j,x_k\}$ for $1\leq j<k\leq 4$; these equations constitute a first order system, which have the second order 
Euler-Lagrange equations \eqref{el} as a consequence 
(so they are analogous to first order Bogomol'nyi equations in a field theory). 
The corresponding solution of the equation \eqref{qms} has also been 
considered 
by Cornalba and Taylor in the context of matrix models \cite{ct}, taking 
\beq\label{qcurve} 
Z_2=f(Z_1) 
\eeq 
so that $[Z_1,Z_2]=0$, 
with 
\beq\label{zzdag}
Z_1=X_1+\ri X_2, \quad Z_2=X_3+\ri X_4, \quad 
[Z_1^\dagger,Z_1]+[Z_2^\dagger,Z_2]=\eps \, \mathbf{1}, 
\eeq 
where $\eps\in\R$ is a parameter. 

In $d=4$, the case where \eqref{qcurve} is the hyperbola $Z_1Z_2=c\, \mathbf{1}$ is the simplest example treated in \cite{ahk}, which admits an elegant operator-valued solution. 
The next interesting case considered in \cite{ct}, and by Arnlind and company, is the parabola, which (after making the explicit parametrization of the curve 
as $Z_1=W$, $Z_2=W^2$), leads to an operator $W$ 
satisfying 
\beq\label{weq}
[W^\dagger, W]+[(W^\dagger)^2,W^2]=\eps\mathbf{1}, 
\eeq 
acting on the Hilbert space
$\mathcal{H} = \Set{ \, \ket{n} | n=0,1,2,\ldots }$ 
according to 
$W \ket{n} = w_n \ket{n+1}$. 
In terms of the squared amplitude $v_n=|w_n|^2$, applying the expectation $\bra{n} \cdot \ket{n}$ to both sides of the commutator equation 
\eqref{weq} leads to the third order difference equation 
$$ 
v_n -v_{n-1}+ v_{n+1}v_n - v_{n-1}v_{n-2} = \eps, 
$$ 
which has the form of a total difference. Hence, upon integration (summation) of this discrete equation, we obtain the 
second order non-autonomous equation 
\beq \label{fulldpI}
v_n (v_{n+1}+v_{n-1}+1)=\eps \, n +\zeta, 
\eeq
for some constant $\zeta$.

Identification of the particular solution of \eqref{fulldpI} required for the quantum minimal surface 
involves consideration of the semiclassical limit. The classical version of the complex parabola 
$z_2=z_1^2$ is parametrised in polar coordinates 
by $z_1=r\re^{\ri \ph}$, $z_2=r^2\re^{2\ri\ph}$, 
so the Poisson bracket equation \eqref{curva} implies that 
$\tilde{r},\ph$ are a pair of canonically conjugate (flat) coordinates, 
where 
\beq\label{treq} 
\tilde{r} = r^4+\tfrac{1}{2}r^2-c. 
\eeq 
Canonical quantization means replacing 
$\tilde{r}\to -\ri\hbar\frac{\partial}{\partial \ph}$, where the latter is the momentum operator conjugate to $\hat{U}$, with $\re^{\ri\hat{U}}\ket{n}=\ket{n+1}$, and
we identify the states $\ket{n}$ for $n\geq 0$ with the non-negative modes $\re^{\ri n\ph}$ on the circle. 
Comparing with \eqref{treq} gives the requirement that 
$v_n^2 +\tfrac{1}{2}v_n\sim n\hbar+c$, 
leading to the approximate solution 
\beq\label{vapprox} 
v_n \approx \tfrac{1}{4}\left(\sqrt{1+8(n+1)\eps}-1\right) 
\eeq 
which agrees with the asymptotic behaviour of positive solutions of \eqref{fulldpI}, 
both in the limit $\hbar\to 0$ with $n$ fixed, and for $n\to\infty$ with $\hbar$ fixed, 
provided that the conditions 
$\zeta=\eps=2\hbar$, 
$c=\hbar$ 
are imposed. 
Hence the second order difference equation is taken as 
\beq\label{dpi}
v_{n+1}+v_{n-1} +1= \frac{\eps(n+1)}{v_n}, 
\eeq 
and one should seek a solution with the initial conditions
\beq\label{inits}
v_{-1}=0, \qquad 
v_0
>0, 
\eeq 
 with the further requirement that $v_n\geq 0$ for all $n>0$, since $v_n$ is a squared amplitude. (Note that the approximate form \eqref{vapprox} also 
satisfies $v_{-1}=0$ and $v_n>0$ for all $n\geq 0$.) 

The equation \eqref{fulldpI} is an example of a discrete \p\ equation. It is commonly referred to in the literature as 
a discrete \p\ I (dP$_{\rm I}$) equation \cite{rg}, because it has a continuum limit to the continuous \p\ I equation 
$\frac{\rd^2w}{\rd t^2}=6w^2+t$.
This dP$_{\rm I}$ equation has been obtained as a reduction of a chain of discrete dressing transformations \cite{grp}, while 
it is also one among 
a number of discrete \p\ equations that were identified by Tokihiro, Grammaticos and Ramani \cite{tgr} 
as arising from compositions of 
\bk\ 
transformations for the \p\ V equation, that is 
\beq\label{pv} 
\frac{\rd^2w}{\rd t^2} =\left(\frac{1}{2w}+\frac{1}{w-1}\right)\, \left(\frac{\rd w}{\rd t}\right)^2 -\frac{1}{t} \frac{\rd w}{\rd t} +\frac{(w-1)^2(\al w^2 + \be)}{t^2 w} 
+ \frac{ \gam w}{t} + \frac{\delta w(w+1)}{(w-1)}
.
\eeq 
(For what follows, only the generic case $\delta\neq 0$ will be relevant, so in that case 
we can set $\delta=-\tfrac{1}{2}$.) 

For the sake of completeness, and to avoid confusion, 
we should remark that equation \eqref{fulldpI} is not the only discrete equation to be called dP$_{\rm I}$. The ``standard" version of dP$_{\rm I}$ is the equation 
\beq\label{dpi_stan} u_{n+1}+u_n+u_{n-1} =1 + \frac{\la n +\mu}{u_n} , \eeq 
with $\la$ and $\mu$ constants;
see, for example equation (3.2) in \cite{rg}. 
It is known that equation \eqref{dpi_stan} is associated with \p\ IV, rather that \p\ V, cf.~\cite{refFIK,magnus}, and 
it is also shown in \cite{refFGR,refCMW} that \eqref{dpi_stan} can be derived from \bk\ transformations of \p\ IV. 
(For another approach, via reductions of the Volterra lattice, see 
\cite{refAS}.)
Note that there are various other inequivalent discrete Painlev\'e equations referred to as dP$_{\rm I}$, or sometimes as alt.~dP$_{\rm I}$ equations. This is best understood using the Sakai classification scheme for Painlev\'e equations suggested in the paper \cite{Sak:2001:RSAWARSGPE}, which provided a complete classification of possible configuration spaces on which discrete Painlev\'e dynamics can occur. Such spaces are families of rational algebraic surfaces known as generalized Halphen surfaces (see 
Figure~\ref{fig:Surface-Type-Classification}). For the differential Painlev\'e equations, these spaces
were introduced earlier by Okamoto \cite{Oka:1979:FAESOPCFPP} as the so-called \emph{spaces of initial conditions}, in which case the parameters of the family are essentially
the parameters of the differential Painlev\'e equation,  and discrete Painlev\'e equations are certain compositions of their B\"acklund transformations, as indicated in 
Figure~\ref{fig:Surface-Type-Classification}. The arrows here can be understood on the one hand, as some parameter degenerations of surface families, and on the other hand, 
as a result of taking the continuum limit of some particular  discrete Painlev\'e dynamics. 

As we show in section \ref{complex} below, equation \eqref{dpi} describes a very special dynamics on the $D_{5}^{(1)}$-surface family. The symmetry group of this family 
is a fully extended affine Weyl group $\widehat{W}\left(A_{3}^{(1)}\right) = W\left(A_{3}^{(1)}\right)\ltimes \operatorname{Aut}\left(A_{3}^{(1)}\right)$,
where $\operatorname{Aut}\left(A_{3}^{(1)}\right)\simeq \mathbb{D}_{4}$ is the dihedral group of symmetries of the affine $A_{3}^{(1)}$ Dynkin diagram, i.e.,
the group of symmetries of a square. This symmetry group describes B\"acklund transformations of the Painlev\'e V differential  equation.  
Standard examples of discrete Painlev\'e equations on this surface family correspond to translations in the weight 
lattice of the usual extended affine Weyl group $\widetilde{W}\left(A_{3}^{1}\right)$, and are commonly known as dP$_{\rm IV}$ and dP$_{\rm III}$ equations. Equation 
\eqref{dpi} is only a quasi-translation, which becomes a translation on a certain sub-locus of the full family with a smaller symmetry group, via  
so-called \emph{projective reduction} \cite{KajNakTsu:2011:PRDPSTA} 
(but further discussion of this is outside the scope of the present paper). In contrast, the ``standard''  dP$_{\rm I}$ equation \eqref{dpi_stan} describes dynamics on the $E_{6}^{(1)}$ surface family, in alignment with the fact that it is associated with 
B\"acklund transformations for 
the Painlev\'e IV equation.  
\begin{center}
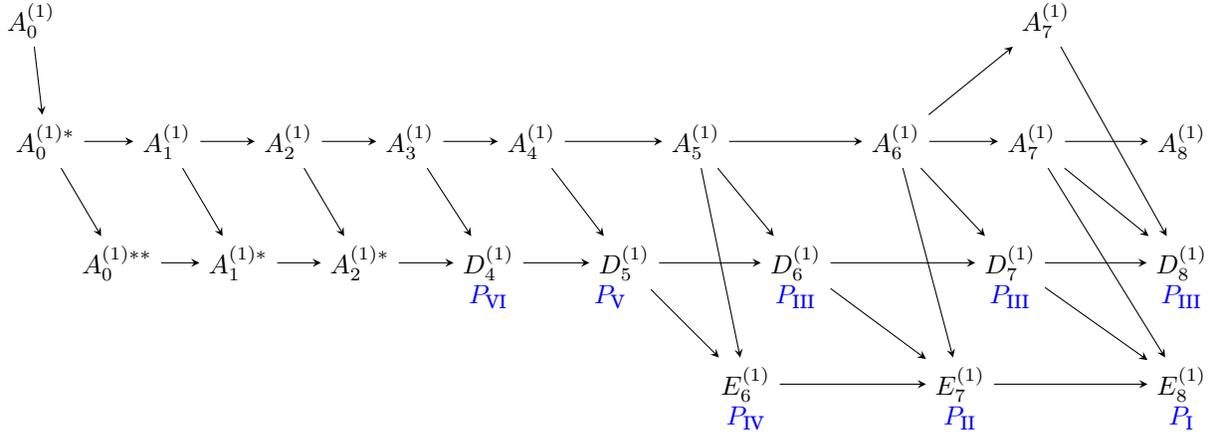
\begin{figure}[h]
{
			\begin{tikzpicture}[>=stealth,scale=0.95]
			\node (e8e) at (0.8,3.4) {$A_{0}^{(1)}$};
			\node (a1qa) at (15,3.4) {$A_{7}^{(1)}$};
			\node (e8q) at (1,1.7) {$A_{0}^{(1)*}$};
			\node (e7q) at (2.7,1.7) {$A_{1}^{(1)}$};	
			\node (e6q) at (4.4,1.7) {$A_{2}^{(1)}$};	
			\node (d5q) at (6.1,1.7) {$A_{3}^{(1)}$};	
			\node (a4q) at (7.8,1.7) {$A_{4}^{(1)}$};	
			\node (a21q) at (10.1,1.7) {$A_{5}^{(1)}$};	
			\node (a11q) at (12.9,1.7) {$A_{6}^{(1)}$};	
			\node (a1q) at (14.8,1.7) {$A_{7}^{(1)}$};	
			\node (a0q) at (16.9,1.7) {$A_{8}^{(1)}$};						
			\node (e8d) at (2,0) {$A_{0}^{(1)**}$};	
			\node (e7d) at (3.7,0) {$A_{1}^{(1)*}$};	
			\node (e6d) at (5.4,0) {$A_{2}^{(1)*}$};	
			\node (d4d) at (7.2,0) {$D_{4}^{(1)}$};	
			\node (a3d) at (9.1,0) {$D_{5}^{(1)}$};	
			\node (a11d) at (11.5,0) {$D_{6}^{(1)}$};	
			\node (a1d) at (14.5,0) {$D_{7}^{(1)}$};	
			\node (a0d) at (16.9,0) {$D_{8}^{(1)}$};	
			\node (a2d) at (10.8,-1.7) {$E_{6}^{(1)}$};	
			\node (a1da) at (13.8,-1.7) {$E_{7}^{(1)}$};	
			\node (a0da) at (16.9,-1.7) {$E_{8}^{(1)}$};	
			\draw[->] (e8e) -> (e8q);	\draw[->] (a1qa) -> (a0d); 	
			\draw[->] (e8q) -> (e7q); 	\draw[->] (e8q) -> (e8d); 	
			\draw[->] (e7q) -> (e6q); 	\draw[->] (e7q) -> (e7d); 	
			\draw[->] (e6q) -> (d5q); 	\draw[->] (e6q) -> (e6d); 	
			\draw[->] (d5q) -> (a4q); 	\draw[->] (d5q) -> (d4d); 	
			\draw[->] (a4q) -> (a21q); 	\draw[->] (a4q) -> (a3d); 	
			\draw[->] (a21q) -> (a11q); \draw[->] (a21q) -> (a11d); \draw[->] (a21q) -> (a2d);	
			\draw[->] (a11q) -> (a1q); 	\draw[->] (a11q) -> (a1d); 	\draw[->] (a11q) -> (a1qa); 	\draw[->] (a11q) -> (a1da);
			\draw[->] (a1q) -> (a0q); 	\draw[->] (a1q) -> (a0d);	\draw[->] (a1q) -> (a0da);
			\draw[->] (e8d) -> (e7d);
			\draw[->] (e7d) -> (e6d);
			\draw[->] (e6d) -> (d4d);
			\draw[->] (d4d) -> (a3d);	
			\draw[->] (a3d) -> (a11d);	\draw[->] (a3d) -> (a2d);	
			\draw[->] (a11d) -> (a1d);	\draw[->] (a11d) -> (a1da);	
			\draw[->] (a1d) -> (a0d);	\draw[->] (a1d) -> (a0da);	
			\draw[->] (a2d) -> (a1da);	\draw[->] (a1da) -> (a0da);
			\node [blue] at ($(d4d.south) + (0,-0.1)$) {$P_{\text{VI}}$};
			\node [blue] at ($(a3d.south) + (-0.2,-0.1)$) {$P_{\text{V}}$};
			\node [blue] at ($(a11d.south) + (0,-0.1)$) {$P_{\text{III}}$};
			\node [blue] at ($(a1d.south) + (0,-0.1)$) {$P_{\text{III}}$};
			\node [blue] at ($(a0d.south) + (0,-0.1)$) {$P_{\text{III}}$};
			\node [blue] at ($(a2d.south) + (0,-0.1)$) {$P_{\text{IV}}$};
			\node [blue] at ($(a1da.south) + (0,-0.1)$) {$P_{\text{II}}$};
			\node [blue] at ($(a0da.south) + (0,-0.1)$) {$P_{\text{I}}$};
		\end{tikzpicture}
} 
\caption{Surface-type classification scheme for Painlev\'e equations.}
\label{fig:Surface-Type-Classification}
\end{figure}
\end{center}

The purpose of this article is to determine an explicit analytic solution for the initial value problem 
\eqref{inits} associated with a quantum minimal surface. First of all, we consider the existence and uniqueness of a   
positive solution to the initial value problem \eqref{inits} for the dP$_{\rm I}$ equation (so $v_n>0$ for all $n\geq 0$). 
Next, we use the complex geometry of 
the 
equation \eqref{dpi}, obtained by blowing up $\Pro^1\times\Pro^1$, to show that 
it corresponds to the same space of initial conditions as the \p\ V equation \eqref{pv}, with the specific parameter values 
$$ 
(\al, \be, \gam, \delta) = (\tfrac{1}{18} (n+1)^2, -\tfrac{1}{18}, -\tfrac{1}{3}(n+1), -\tfrac{1}{2}) 
.
$$
We will then proceed to employ some recent results by two of us (Clarkson and Mitchell, obtained in collaboration with Dunning), 
giving explicit modified Bessel function formulae for families of classical solutions of \p\ V that were previously considered in the literature \cite{fw, masuda}, 
and use these to determine an exact analytic expression for the unique solution of the initial value problem 
\eqref{inits} so that $v_n$ remains positive for all $n>0$. 
Our main result is as follows. 
\begin{thm}\label{mainthm}
For each $\eps>0$, the unique positive solution of the dP$_{\rm I}$ equation \eqref{dpi} 
subject to the initial conditions \eqref{inits} is determined by the value of $v_0=v_0(\eps)$, which is given by a ratio of modified Bessel functions, that is
\beq\label{vexpl}
v_0 = \tfrac{1}{2}\left(\frac{K_{5/6}(\tfrac{1}{2}t)}{K_{1/6}(\tfrac{1}{2}t)}-1\right), 
\qquad \mathrm{where} \quad t=\frac{1}{3\eps}. 
\eeq 
For each $n\geq 0$, the corresponding quantities $v_n>0$ 
are written explicitly as ratios of Wronskian determinants 
whose entries are specified in terms of modified Bessel functions. 
\end{thm}

\section{Unique positive solution: cold open} 
\setcounter{equation}{0} 

In this section, we will present the preliminary steps of the proof that there is a unique solution of the dP$_{\rm I}$ equation \eqref{dpi}, subject to the initial conditions \eqref{inits}, that is non-negative (in fact, positive) 
for all $n\geq 0$. The precise statement is as follows. 

\begin{thm}\label{positive} 
For each value of $\eps>0$ 
there is a unique value of $v_0>0$ such that the solution 
of the second order difference equation \eqref{dpi} with the initial data 
 \eqref{inits} satisfies $v_n> 0$ for all $n\geq 0$.
\end{thm}

In our initial approach to proving the above result, 
we start by considering the set of real sequences $\bu = (u_n)_{n\geq 0}$, which contains the Banach space 
$$ 
\ell_\eps^{\infty} = \big\{ \, \bu\,\big\vert \, \Vert \bu \Vert <\infty\, \big\}, 
$$ 
where $\Vert \cdot \Vert$ denotes the weighted $\ell^\infty$ norm 
$$ 
\Vert \bu\Vert = \underset{n\geq 0}{\sup} \,\frac{| u_n |}{\eps(n+1)}. 
$$ 
Then we can define 
a transformation $T$, which acts on real non-negative sequences 
$\bu\geq 0$ (that is, $u_n\geq 0$ for all $n\geq 0$), according to 
\beq\label{tdef} 
T(u_n) = 
{\everymath={\displaystyle}
\begin{cases*}
 \frac{\eps}{(u_{n+1}+1)}, & if\quad $n = 0$ \\
 \frac{\eps(n+1)}{(u_{n+1}+u_{n-1}+1)},\quad & if\quad $n> 0$. 
 \end{cases*}
}
\eeq 
By a convenient abuse of notation, we write $\bu\mapsto T\bu$ for the action on sequences, while brackets are used to denote the individual terms $T(u_n)$ 
of a sequence $T\bu$. 
Under the action of $T$, any non-negative sequence
is mapped to a 
subset of the unit ball in $\ell^\infty_\eps$, namely 
$$ 
\mathcal{A}^{(0)} = \big\{ \,\bu\geq 0\, \,\big\vert\,\, \Vert \bu\Vert\leq 1\,\big\}, 
$$ 
which is a complete set with respect to this norm. 

Numerically, for any fixed $\eps>0$, the repeated application of the mapping $T$ to a (truncated) positive sequence 
provides a rapid numerical method to obtain the positive solution of the 
 dP$_{\rm I}$ equation to any desired precision. (See Figs.\ref{vseqplot} and \ref{v0plots}, obtained from 100 iterations of $T$ applied to a truncated sequence with $0\leq n \leq 20$, where 
the approximation \eqref{vapprox} was used to specify the initial conditions and 
fix the boundary values at $n=-1$ and 21.) 
The set $\mathcal{A}^{(0)}$ is mapped to a subset of itself, so 
ideally we would want to show that $T$ is a contraction mapping on this set, and 
hence, by the Banach fixed point theorem, 
it would follow that it has a unique fixed point $\bv$ with $T(\bv)=\bv$. From \eqref{tdef}, such a fixed point $\bv=(v_n)_{n\geq 0}$ is a positive sequence that satisfies the dP$_{\rm I}$ 
equation \eqref{dpi} with initial condition $v_{-1}=0$. 
However, basic estimates and numerical calculations show that $T$ is not a contraction mapping on the whole set $\mathcal{A}^{(0)}$, and in fact the squared mapping $T^2$ behaves better than $T$; so we need to use some 
more refined bounds to prove the uniqueness of the 
positive solution $\bv$. 
In particular, we will adapt some ideas from \cite{ct} and \cite{ahk}, where it was observed that, for each $n\neq 0$, the value of the positive solution 
$v_n$ should be obtained as the intersection of a sequence of intervals of successively shrinking diameter. Furthermore, at the end of section  \ref{classical}, we will proceed to show that there is only one  solution 
of (\ref{dpi}) satisfying the required bounds. 

\begin{figure}[ht!]
 \includegraphics[width=4in, height=4in] {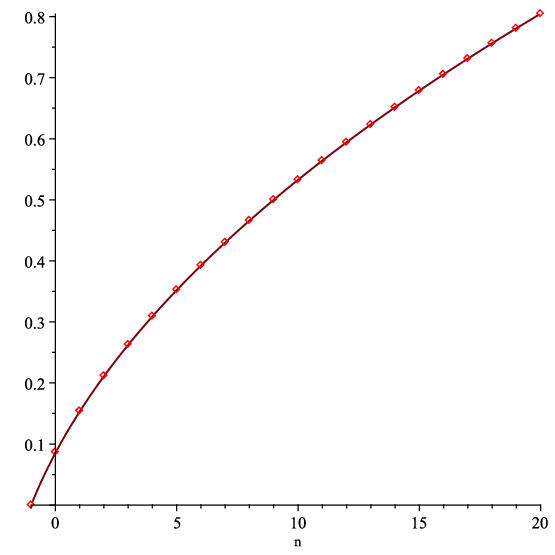}
\centering
 \caption{Numerical computation of $v_n$ (dots) with $\eps=0.1$, for $-1\leq n \leq 20$, 
compared with the graph of the 
approximation \eqref{vapprox}.
}\label{vseqplot} 
\end{figure}

We define a set of non-negative sequences $\{\mathbf{b}^{(k)}\}_{k\geq 1}$ by successively applying $T$ to the zero sequence $\mathbf{0}$, so that 
\beq\label{bseq}
\bb^{(-1)} =\mathbf{0}, \qquad\bb^{(k)} =T\bb^{(k-1)} \quad \mathrm{for}\quad k\geq 0. 
\eeq 
The first few steps in the $T$-orbit of $\mathbf{0}$ are specified by a formula for their terms, valid for all $n\geq 0$: 
\beq\label{firstbs}
 \rb_n^{(0)}=\eps (n+1), \qquad \rb_n^{(1)} = \frac{\eps(n+1)}{1+2\eps(n+1)}, 
\qquad \rb_n^{(2)} = \frac{\eps(n+1)}{1+\frac{\eps n}{1+2\eps n} +\frac{\eps (n+2)}{1+2\eps (n+2)} }. 
\eeq 
Thereafter, for $k\geq 3$, there is no longer a uniform expression for the 
iterate $\bb^{(k)}$ as a ratio of polynomials 
$$ 
\rb^{(k)}_n = \frac{\rp_n^{(k)}(\eps)}{\roq_n^{(k)}(\eps)} \qquad  \rp_n^{(k)},\roq_n^{(k)} \in \Z[\eps], 
$$ 
valid for all $n$: due to the fact that 
the equations \eqref{tdef} defining $T$  for $n=0$ and $n>0$ are different,  
the coprime polynomials $\rp_n^{(k)}(\eps), \roq_n^{(k)}(\eps)$ 
have distinct forms 
for $n=0,\ldots, k-3$, while there is another formula for them that is uniformly valid only for $n\geq k-2$. 
For instance, when $k=3$ we have 
$$ 
\rb_0^{(3)} = \frac{\eps (1+12\eps+24\eps^2)}{1+14\eps+40\eps^2+24\eps^3}, 
\qquad \mathrm{and} \quad 
\rb_n^{(3)} = \frac{\rp_n^{(3)}(\eps)}{\roq_n^{(3)}(\eps)} \quad \mathrm{for}\quad n\geq 1, 
$$ 
where 
\begin{align*}
\rp_n^{(3)}(\eps) & =  (n+1)\eps \Big(1+6n\eps + 8(n^2-1)\eps^2 \Big) 
\Big(1+6(n+2)\eps + 8(n+1)(n+3)\eps^2\Big), \\ 
%
 \roq_n^{(3)}(\eps) 
& =  1+14(n+1)\eps+8(9n^2+18n+4)\eps^2+8(n+1)(21n^2+42n-11)\eps^3 \\ 
&\qquad +16(n+1)^2(11n^2+22n-20)\eps^4+64(n+1)^3(n-1)(n+3)\eps^5. 
\end{align*} 
Nevertheless, for all $n$ there are expressions for 
$\rb^{(k)}_n$ as rational functions of $\epsilon$ and the variable $z=\eps(n+1)$, which are described in Lemma \ref{jlim}  below.

If we start with a sequence $\bu\in \mathcal{A}^{(0)}$ and apply $T$ once, then we find 
$$ 
\frac{\eps (n+1)}{1+\eps n+\eps (n+2)}
\leq \frac{\eps (n+1)}{1+u_{n-1}+u_{n+1} } \leq \eps (n+1), 
$$ 
or in other words 
$ \rb_n^{(1)}\leq
T(u_n) \leq \rb_n^{(0)}$, 
while another application of $T$ gives 
$$
\frac{\eps (n+1)}{1+\eps n+\eps (n+2)}
\leq \frac{\eps (n+1)}{1+T(u_{n-1})+T(u_{n+1}) } \leq 
 \frac{\eps(n+1)}{\ds 1+\frac{\eps n}{1+2\eps n} +\frac{\eps (n+2)}{1+2\eps (n+2)} }, 
$$ 
so that
$ \rb_n^{(1)}\leq
T^2(u_n) \leq \rb_n^{(2)}
$. 
Continuing in this way, by induction we obtain the following result. 

\begin{lem}\label{abounds}
For each non-negative integer $k$, the 
iterates $T^k\mathbf{u}$ of 
$\mathbf{u}\in \mathcal{A}^{(0)}$ 
satisfy 
the inequalities 
\beq\label{evenk}
 \rb_n^{(2j-1)}\leq
T^{2j}(u_n) \leq \rb_n^{(2j)} \qquad \mathrm{for}\,\mathrm{all} \quad n\geq 0, 
\eeq 
when $k=2j$ is even, and 
 \beq\label{oddk}
 \rb_n^{(2j+1)}\leq
T^{2j+1}(u_n) \leq \rb_n^{(2j)} \qquad \mathrm{for}\,\mathrm{all} \quad n\geq 0, 
\eeq 
when $k=2j+1$ is odd. 
The sequences of lower/upper bounds in 
\eqref{bseq} satisfy 
\beq\label{bkbds}
0\leq \rb_n^{(2j-1)}< \rb_n^{(2j+1)}<\rb_n^{(2j+2)}<\rb_n^{(2j)} 
\qquad \mathrm{for}\,\mathrm{all} \quad n\geq 0,
\eeq 
for each $j\in \N$. 
\end{lem} 

For each $k\geq 0$ we have 
the set $\mathcal{A}^{(k)}=T^k\mathcal{A}^{(0)}$, 
and the preceding result implies that the 
next set in the sequence, 
$\mathcal{A}^{(k+1)}\subset {A}^{(k)}$, is a proper subset of the previous one. 
Furthermore, the inequalities \eqref{bkbds} 
immediately imply the existence of the limits of upper and lower bounds, that is 
\beq\label{uplo}
\lim_{j\to \infty} \rb_n^{(2j-1)}
=\limsup_{j\geq 0} \rb_n^{(2j-1)} \leq 
\liminf_{j\geq 0} \rb_n^{(2j)}=
\lim_{j\to \infty} \rb_n^{(2j)}
\eeq 
for each $n\geq 0$.
The problem is then how to show the equality of the upper and lower limits above for each $n$, since in that case a squeezing argument immediately implies that the iterates 
$T^k\bu$ converge to the unique positive fixed point of $T$. 

\begin{propn}\label{existence} 
For all $\eps>0$ there exists (at least one) $v_0=v_0(\eps)$ such that 
the solution of \eqref{dpi} with the initial data 
 \eqref{inits} is positive, and satisfies $v_n> 0$ for all $n\geq 0$, as well as 
 $$ 
 \bv = (v_n)_{n\geq 0}\in\underset{k\geq 0}{\cap}{\cal A}^{(k)},  
 $$
 so that,  for all $n\geq 0$, 
 \beq\label{vbbds}
 \lim_{j\to\infty} \rb_n^{(2j-1)} \leq v_n \leq 
\lim_{j\to\infty} \rb_n^{(2j)} .
 \eeq 
\end{propn}
\begin{prf} 
The existence of a positive solution $\bv$ is proved in \cite{ahk}, 
where it shown that for each $\eps>0$ there is an infinite sequence of 
open intervals $I_k = I_k(\eps)\subset \R$, with $I_1=(0,\eps)$ and $I_k\subset I_{k-1}$, such 
that $v_0\in I_k\implies v_1,\ldots,v_k>0$, and 
${\cap}_{k\geq 0}\, {I_k}\neq \emptyset$. Hence 
if $v_0\in {\cap}_{k\geq 0}\, {I_k}$ then 
the corresponding sequence $\bv$ is a positive solution. 
Then, because $T\bv =\bv$, it follows from Lemma \ref{abounds} that $\bv\in{\cal A}^{(k)}$
for each $k\geq 0$, and hence (\ref{vbbds}) holds for each $n\geq 0$.
\end{prf}

It is instructive to compare the upper and lower bounds for different $n$, as well as introducing  
the rescaled bounds $\rho_n^{(k)}$, 
which specify the norms: 
$$
\rho_n^{(k)} = \frac{\rb_n^{(k)}}{\eps(n+1)}, \qquad \Vert\mathbf{b}^{(k)}\Vert = \underset{n\geq 0}{\sup} \, \rho_n^{(k)}.
$$ 
Clearly we have $\rho_n^{(-1)} 
=0$, 
while 
$\rb_{n+1}^{(0)}>\rb_n^{(0)}
$ 
and 
$\rho_n^{(0)} =1
$ 
for all $n$, 
and from \eqref{bkbds} we also see that 
\beq \label{rhorange}
0<\rho_n^{(k)}<1 \qquad \mathrm{for}\,\mathrm{all} \quad n\geq 0, k\geq 1
.
\eeq
If we now assume for some $k$ that 
$ \rb_{n+1}^{(k)}>\rb_n^{(k)}$ holds for all $n\geq 0$, then from the definition of the map $T$ we may write 
\begin{align*}
\frac{1}{\rho_{n+1}^{(k+1)}} - 
\frac{1}{\rho_{n}^{(k+1)}} & =
\big(1+\rb_n^{(k)}+\rb_{n+2}^{(k)}\big)
- \big(1+\rb_{n-1}^{(k)}+\rb_{n+1}^{(k)}\big) \\ 
 & = 
\big(\rb_n^{(k)}-\rb_{n-1}^{(k)}\big)+ 
\big(\rb_{n+2}^{(k)}-\rb_{n+1}^{(k)}\big)>0 
\end{align*}
where we set $\rb_{-1}^{(k)}=0$ so that this makes sense
when $n=0$, 
which implies that $\rho_{n+1}^{(k+1)}<\rho_{n}^{(k+1)}$. 
On the other hand, 
$$\rb_{n+1}^{(k)}>\rb_n^{(k)}\iff
(n+2)\rho_{n+1}^{(k)}> (n+1)\rho_{n}^{(k)},$$
and we can calculate
\begin{align*}
\frac{\rb_{n+1}^{(k+1)}-\rb_n^{(k+1)}}{\rho_{n+1}^{(k+1)}\rho_{n}^{(k+1)}} & = 
\eps(n+2) \Big(1+\rb_{n-1}^{(k)}+\rb_{n+1}^{(k)}\Big)
 - \eps(n+1) \Big(1+\rb_n^{(k)}+\rb_{n+2}^{(k)}\Big) \\ 
& = \eps\Big(1+
(n+2)\big(n\rho_{n-1}^{(k)}+(n+2)\rho_{n+1}^{(k)}\big)
-(n+1)\big((n+1)\rho_{n}^{(k)}+(n+3)\rho_{n+2}^{(k)}\big)
\Big) .
\end{align*}
If we now assume that 
$\rho_{n+1}^{(k)}<\rho_{n}^{(k)}$ holds for all $n\geq 0$,
then we can replace the term with index $n+2$ above, and also (for $n>0$) 
replace the term with index $n-1$, and so 
rearrange to find the lower bound 
$$ \frac{\rb_{n+1}^{(k+1)}-\rb_n^{(k+1)}}{\rho_{n+1}^{(k+1)}\rho_{n}^{(k+1)}} >
\eps\big(1+ \rho_{n+1}^{(k)}-\rho_{n}^{(k)}\big) >0, 
$$
using \eqref{rhorange} to obtain the final inequality. 
Thus, by induction on $k$, 
we find the following. 

\begin{lem}\label{abounds2}
For $k\geq 0$, the sequences of lower/upper bounds satisfy 
\beq\label{nbk}
\rb_{n+1}^{(k)}>\rb_n^{(k)}
\qquad \mathrm{for}\,\mathrm{all} \quad n\geq 0,
\eeq 
while for $k\geq 1$ the rescaled 
bounds satisfy 
\beq\label{nrhok} 
\left(\frac{n+1}{n+2}\right)\, 
\rho_n^{(k)}< \rho_{n+1}^{(k)}<\rho_n^{(k)}
\qquad \mathrm{for}\,\mathrm{all} \quad n\geq 0,
\eeq 
and hence for each $k$ the norm of $\mathbf{b}^{(k)}$ is 
$\Vert\mathbf{b}^{(k)}\Vert = \rho_0^{(k)}$.
\end{lem}

\begin{figure}[ht!]
\centering
\begin{minipage}[b]{3in}
 \includegraphics[width=\linewidth, height=3in] {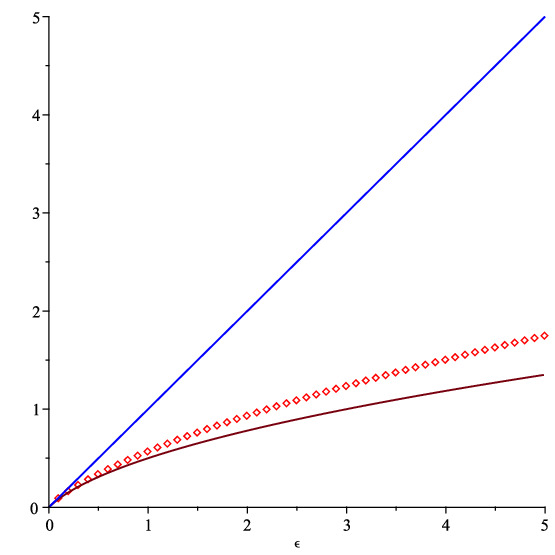}
\end{minipage}
\hfill
\begin{minipage}[b]{3in}
 \includegraphics[width=\linewidth, height=3in] {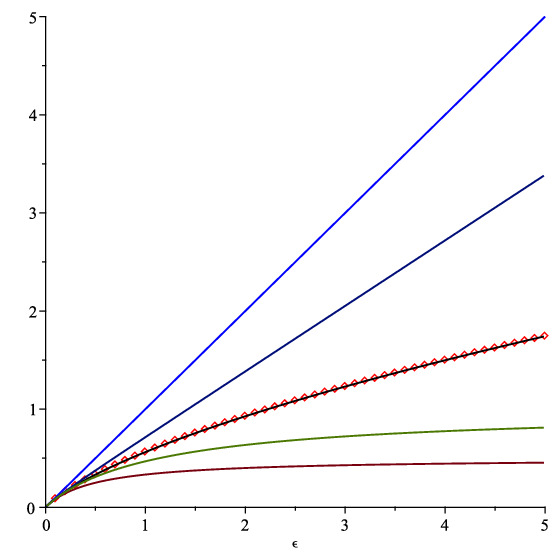}
\end{minipage}
\caption{Left: numerical computation of $v_0(\eps)$ (dots) in the range $0<\eps \leq 5$, 
plotted against linear bound $b_0^{(0)}=\eps$ 
and 
approximation $\hat{v}_0(\eps)$ as in \eqref{vbar0}. 
Right: same computation, but compared with upper bounds 
$b_0^{(0)},b_0^{(2)}$, lower bounds $b_0^{(1)},b_0^{(3)}$, 
and exact formula \eqref{vexpl}.} \label{v0plots} 
\end{figure}

\begin{lem}\label{jlim} 
For each $j\geq 0$, the rescaled bounds have the asymptotic behaviour 
\beq\label{revenoddas}
\rho_n^{(2j)}\sim \frac{1}{j+1}, \qquad 
\rho_n^{(2j+1)}\sim \frac{j+1}{2\eps (n+1)}, \qquad 
\mathrm{as}\quad n\to\infty. 
\eeq
Moreover, the leading part of the Taylor expansion of each of  the 
rescaled bounds  
at $\eps=0$ is 
\beq\label{rtaylor}
\rho_n^{(k)}= 1-2\eps (n+1) +O(\eps^2) \qquad 
\mathrm{for}\,\mathrm{for}\,\mathrm{all} \quad  k\geq 1, 
\eeq 
while for all $\eps>0$ the first derivatives with respect to $\eps$ satisfy
\beq\label{ederivs}
\frac{\rd \rho_n^{(k)}}{\rd \eps}<0 , \qquad 
\frac{\rd\rb_n^{(k)}}{\rd \eps}>0 
\eeq 
for $k\geq 1$, $n\geq 0$, so that each bound 
$\rho_n^{(k)}$, $\rb_n^{(k)}$  is monotone 
decreasing/increasing in $\eps$, respectively. 
\end{lem}
\begin{prf}
As noted above, the action of the map $T$ is such that  the components $\rb_n^{(k)}$ of the sequences 
$\bb^{(k)}$ are rational functions of $\eps$ which have an   expression as ratios of coprime polynomials in $\eps$ 
that is uniformly valid for all $n$ 
when $k<3$, while for $k\geq 3$ these polynomials have a uniform structure  for $n\geq k-2$ only. 
However, if we set $z=\eps (n+1)$,
then for $k=0,1,2$  
we can write 
\beq\label{1strhos} 
\rho_n^{(0)}=R^{(0)}(z,\eps) := 1, \quad 
\rho_n^{(1)}=R^{(1)}(z,\eps) :=\frac{1}{1+2z},\quad 
\rho_n^{(2)}=R^{(2)}(z,\eps) :=\frac{1}{1+\frac{z-\eps}{1+2(z-\eps)}+\frac{z+\eps}{1+2(z+\eps)}}, 
\eeq 
and for all $n$
we can similarly 
express $\rho_n^{(k)}$ 
as $R^{(k)}(z,\eps)$,  
a rational function of $z$ and $\eps$  
that is determined recursively via the finite difference equation 
\beq\label{RRec}
R^{(k+1)}(z,\eps) =\frac{1}{1+(z-\eps)R^{(k)}(z-\eps,\eps) +(z+\eps)R^{(k)}(z+\eps,\eps)}. 
\eeq 
So from the definition of the map $T$,   the identity 
$\rho_n^{(k)} = R^{(k)}(\eps(n+1),\eps)$ 
holds for all $n$. 
The result \eqref{revenoddas} thus corresponds to the asymptotics of the rational functions $R^{(k)}(z,\eps)$ 
as $z\to\infty$, which varies according to the parity of $k$, and we can proceed by induction on $j$. For the base case  $j=0$ we have 
$R^{(0)}(z,\eps)=1$, 
$R^{(1)}(z,\eps)\sim\frac{1}{2z}$, 
so the claim for $\rho_n^{(0)}$ is trivially true, while 
for $\rho_n^{(1)}$
it gives the correct result by substituting $z=\eps(n+1)$. 
So for the induction, for some fixed $j$  we can assume 
that 
$R^{(2j)}(z,\eps)\sim\frac{1}{j+1}$
as $z\to\infty$, 
and then by applying \eqref{RRec} 
we immediately obtain
$$
R^{(2j+1)}(z,\eps)\sim \left(
1+ \frac{z-\eps}{j+1}+\frac{z+\eps}{j+1}
\right)^{-1}\sim \frac{j+1}{2z}, \qquad z \to\infty , 
$$ 
which, upon setting $z=\eps(n+1)$ gives the correct leading order behaviour  
for $\rho_n^{(2j+1)}$ as $n\to\infty$. 
Applying  \eqref{RRec} once again gives 
$$
R^{(2j+2)}(z,\eps)\sim \left(
1+ \frac{(z-\eps)(j+1)}{2z}+ \frac{(z+\eps)(j+1)}{2z}
\right)^{-1}\sim \frac{1}{j+2}, \qquad z \to\infty , 
$$ 
and this completes the inductive step. 

For the leading order behaviour of the scaled bounds at $\eps=0$, 
it is convenient to write an equivalent version of \eqref{RRec} in terms of $\rho_n^{(k)}$,  namely 
\beq\label{rhonrec}
\rho_n^{(k+1)}=\left(
1+\eps n \rho_{n-1}^{(k)}+\eps (n+2) \rho_{n+1}^{(k)}
\right)^{-1}, 
\eeq 
which  
is valid for all $n$.
When $k=1$, the leading order expansion \eqref{rtaylor} is immediately obtained from the geometric series for 
\beq\label{rho1}
\rho_n^{(1)} = \frac{1}{1+2\eps(n+1)}, 
\eeq 
and the general case easily follows via induction on $k$, 
by applying \eqref{rhonrec} at each step. 

To obtain the monotonicity in $\eps$ of 
$\rho_n^{(k)}$ and 
$\rb_n^{(k)}$, it is clear from (\ref{rho1}) and from  
$\rb_n^{(1)}=\tfrac{1}{2}\big(1-\rho_n^{(1)}\big)$, 
that the inequalities \eqref{ederivs} hold for $k=1$, and we proceed by induction on $k$. 
Then, assuming that (for all $n\geq 0$) both 
$\ds\dep{\rho_n^{(k)}}<0$ and $\ds\dep{\rb_n^{(k)}}>0$ hold for some $k$, 
differentiating \eqref{rhonrec} yields 
\beq\label{rhodoteq}
\frac{\rd \rho_n^{(k+1)}}{\rd \eps}=-\big(\rho_n^{(k+1)}\big)^2\left(
\dep{\rb_{n-1}^{(k)}}+\dep{\rb_{n+1}^{(k)}}\right)<0, 
\eeq 
implying that $\rho_n^{(k+1)}$ is monotone decreasing in $\eps$. 
Now differentiating $\rb_n^{(k+1)}=\eps (n+1) \rho_n^{(k+1)} $ yields 
\begin{align*}
\ds\dep{\rb_n^{(k+1)}}& = 
(n+1) \left( \rho_n^{(k+1)} +\eps \ds\dep{\rho_n^{(k+1)}}\right) \\ 
& = (n+1) \rho_n^{(k+1)} \left( 1-\eps \rho_n^{(k+1)} 
\left(\dep{\rb_{n-1}^{(k)}}+\dep{\rb_{n+1}^{(k)}}\right) \right) \\
& = \rb_n^{(k+1)} \left( \eps^{-1}- \rho_n^{(k+1)} 
\left(\dep{\rb_{n-1}^{(k)}}+\dep{\rb_{n+1}^{(k)}}\right) \right), 
\end{align*}
where we used \eqref{rhodoteq}. Then we calculate 
\begin{align*}
\dep{\rb_{n-1}^{(k)}}+\dep{\rb_{n+1}^{(k)}} &= 
n \left( \rho_{n-1}^{(k)} +\eps \dep{\rho_{n-1}^{(k)}}\right) 
+(n+2) \left( \rho_{n+1}^{(k)} +\eps \dep{\rho_{n+1}^{(k)}}\right)\\
&< n \rho_{n-1}^{(k)} + (n+2) \rho_{n+1}^{(k)} 
=\eps^{-1}\big(\rb_{n-1}^{(k)}+\rb_{n+1}^{(k)} \big) , 
\end{align*}
using the inductive hypothesis on $\rd\rho_n^{(k)}/\rd \eps$, 
and together with \eqref{rhonrec} this  implies that 
$$
    \dep{\rb_n^{(k+1)}} > \frac{\rb_n^{(k+1)}}{\eps}\left(1- 
    \rho_n^{(k+1)} \big(\rb_{n-1}^{(k)}+\rb_{n+1}^{(k)}\big) 
    \right) 
    = \frac{\rb_n^{(k+1)}\rho_n^{(k+1)} }{\eps} >0, 
$$
as required. 
\end{prf}

\begin{remark}\label{updown}
Since the sequence $(\rho_n^{(2j)})_{n\geq 0}$ decreases monotonically with $n$, as in Lemma  \ref{abounds2}, and tends to $1/(j+1)$, it follows that 
$\rho_n^{(2j)}>\frac{1}{j+1}$ 
for all $n\geq 0$, while the recursion \eqref{rhonrec} with $k=2j$ shows that 
$\rho_n^{(2j+1)}<\frac{j+1}{2\eps (n+1)}$, 
which is equivalent to 
$ \rb_n^{(2j+1)}<\frac{1}{2}(j+1)$.  
\end{remark}

To further address our main assertion about the uniqueness of the positive solution of \eqref{dpi}, 
we introduce the 
differences 
\beq\label{deldef} 
\Delta^{(k)}_n = (-1)^k\big(\rho^{(k)}_n - \rho^{(k-1)}_n\big), 
\eeq 
where the alternating sign is chosen so that 
$\Delta^{(k)}_n > 0$ for all $k\geq 0$, $n\geq 0$, as is seen 
directly by dividing the inequalities \eqref{bkbds} by $\eps(n+1)$ for each $n$. Then the coincidence of the lower and upper limits in (\ref{uplo}), which yields the desired squeezing argument, is equivalent to the statement that 
\beq\label{dellim}
\lim_{k\to\infty}\Delta^{(k)}_n = 0 \qquad \mathrm{for}\,\mathrm{each}\quad n\geq 0.  
\eeq 
To see why 
the latter result is plausible, we 
consider the 
behaviour of these differences for small $\eps$, 
which will be needed later. 

\begin{lem}\label{delezero} The leading part of the Taylor 
expansion of each of the differences \eqref{deldef} at $\eps=0$ is 
$$
\Delta^{(k)}_n =  c_n^{(k)}\,\eps^k \Big(1+O(\eps) \Big), 
$$
where $c_n^{(0)} =1$ for all $n$, and 
the leading coefficient is given recursively by 
\beq\label{cnkrec}
c^{(k+1)}_n = n c^{(k)}_{n-1} +  (n+2) c^{(k)}_{n+1} 
\qquad \mathrm{for}\,\, \mathrm{all} \,\,k,n\geq 0. 
\eeq 
\end{lem}
\begin{prf} 
The result is by induction on $k$. For the base case $k=0$ we have 
$\Delta^{(0)}_n = 1$, and hence $c_n^{(0)} =1$,  for all $n$. 
For the inductive step, we write 
$$ \rho^{(k)}_n - \rho^{(k+1)}_n
= \rho^{(k)}_n  \rho^{(k+1)}_n
\Big( (\rho^{(k+1)}_n)^{-1}-(\rho^{(k)}_n)^{-1} \Big) , 
$$
and so, by using \eqref{rhonrec} and collecting terms inside the round brackets, we obtain the identity 
\beq\label{delid}
\Delta^{(k+1)}_n =\rho^{(k)}_n  \rho^{(k+1)}_n\big(\eps n \Delta^{(k)}_{n-1} + \eps (n+2) \Delta^{(k)}_{n+1}\big). 
\eeq 
Upon using the inductive hypothesis and substituting in the 
leading order expansion \eqref{rtaylor} for the two prefactors 
on the right-hand side of \eqref{delid}, we immediately obtain 
$
\Delta^{(k+1)}_n =  c_n^{(k+1)}\,\eps^{k+1} \big(1+O(\eps) \big)  
$, 
where (for each $n$) the leading coefficient $c_n^{(k+1)}$  
is given in terms 
of the coefficients with superscript $k$ by the recursion 
\eqref{cnkrec}, as required. 
\end{prf}

Note that we have 
$c_n^{(1)} = 2(n+1)$, 
$c_n^{(2)} = 4(n^2+2n+2)$,
and it is apparent from \eqref{cnkrec} that these leading 
coefficients are monotone increasing with $n$, that is 
$c_{n+1}^{(k)}>c_n^{(k)}$, 
for $k\geq 1$. 
This monotonicity is desirable, since it suggests that, for small 
enough $\eps$, we should have 
$\Delta^{(k)}_{n+1}>\Delta^{(k)}_n$, while from Lemma \ref{jlim}
we see that 
\beq\label{deljlim}
\lim_{n\to\infty}\Delta^{(2j)}_n 
=\frac{1}{j+1}=\lim_{n\to\infty}\Delta^{(2j+1)}_n; 
\eeq 
so if the sequence 
$(\Delta^{(k)}_n)_{n\geq 0}$ is increasing with $n$ for each $k\geq 1$, 
then the sought after result \eqref{dellim} follows 
immediately from taking the limit $j\to\infty$ in \eqref{deljlim}. 
However, the monotonicity of $c_n^{(k)}$ in $n$ is not enough, because the convergence of the Taylor series 
\eqref{rtaylor}, and hence the result of Lemma \ref{delezero}, 
does not hold uniformly in $n$. For instance, 
the geometric series for $\rho_n^{(1)}$ has radius of 
convergence $\tfrac{1}{2(n+1)}$. 
Moreover, if we introduce the functions 
\beq\label{Delfns}
\Delta^{(k)}(z,\eps) = (-1)^k \Big(R^{(k)}(z,\eps) - R^{(k-1)}(z,\eps)\Big) 
\eeq 
in terms of the $R^{(k)}(z,\eps)$ satisfying \eqref{RRec}, then we 
might hope to use their behaviour in the range $z\geq \eps$ to determine suitable bounds on the discrete set of values 
$ \Delta^{(k)}_n = \Delta^{(k)}\big((n+1)\,\eps,\eps\big)$. 
However, 
this turns out to be tricky for two reasons: first of all, we can show that $ \Delta^{(2)}(z,\eps)$ and the other functions  \eqref{Delfns} 
are not monotone in $z$ except for  small $\eps\lesssim 0.3$; 
and secondly, for $k\geq 3$ these rational functions have poles at certain points in the range $z\geq \eps$, lying in between the discrete values of interest, so they are unbounded on this range. Figure \ref{deltaplots} illustrates these features for $k=2$ and $4$.

\begin{figure}[ht!]
\centering
\begin{minipage}[b]{3in}
 \includegraphics[width=\linewidth, height=3in] {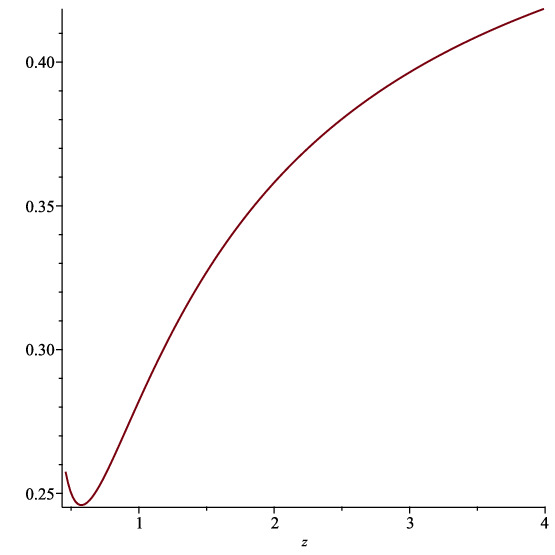}
\end{minipage}
\hfill
\begin{minipage}[b]{3in}
 \includegraphics[width=\linewidth, height=3in] {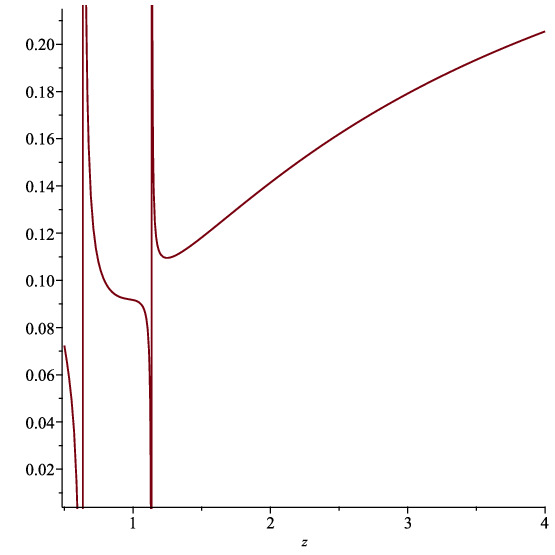}
\end{minipage}
\caption{Left: graph of $\Delta^{(2)}(z,\eps)$ against $z$ for 
$\eps=0.5$ showing the range $0.5\leq z \leq 4$, 
illustrating the local minimum in the interval $\eps<z<2\eps$. 
Right: graph of $\Delta^{(4)}(z,\eps)$ against $z$ for 
$\eps=0.5$, showing vertical asymptotes at two poles,  
lying  in the intervals $\eps<z<2\eps$ and 
$2\eps<z<3\eps$, respectively.} \label{deltaplots} 
\end{figure}

Further numerical investigations suggest that if $\eps$ is not two large, then the products $\rho_n^{(2j+1)}\rho_n^{(2j)}$ and 
$\rho_n^{(2j)}\rho_n^{(2j-1)}$ appear to be bounded above by their  values as $n\to\infty$, as in \eqref{revenoddas}. This is in spite of  Remark \ref{updown}, which shows that each product consists of an upper bound for one of the factor and a lower bound for the other. It is easy to see that $\rho_n^{(1)}\rho_n^{(0)}<\tfrac{1}{2\eps(n+1)}$, so the first 
case of interest is the product $\rho_n^{(2)}\rho_n^{(1)}$, for which 
the required bound can be expressed as 
$ R^{(2)}(z,\eps)R^{(1)}(z,\eps) < \frac{1}{4z}$ 
for $ z\geq \eps$, 
and this can be shown to be equivalent to  the condition 
$$ 
\rho^{(2)}_0 \rho^{(1)}_0 = \frac{(1+4\eps)}{(1+2\eps)(1+6\eps)}
< \frac{1}{4\eps}
$$
which is satisfied whenever $0<\eps<\frac{1}{2}(\sqrt{2}+1)$. 
This leads to the following 

\begin{conje}\label{pos} 
For each value of $\eps$ in the range $0<\eps<\eps^\star$, 
where 
\beq\label{estar}
\eps^* = \tfrac{1}{2}(\sqrt{2}+1)\approx 1.2071, 
\eeq 
the rescaled bounds satisfy 
\beq\label{rhobs} 
\rho_n^{(2j-1)}\rho_n^{(2j-2)}<\frac{1}{2\eps (n+1)}  \qquad \mathrm{and} \qquad 
\rho_n^{(2j)}\rho_n^{(2j-1)}<\frac{j}{2\eps (n+1)(j+1)},
\eeq 
for all $n\geq 0$ and all $j\geq 1$.
\end{conje}

The desirability of the bounds in \eqref{rhobs}, and especially the second one, comes from the fact that it immediately yields an inductive proof that 
\beq\label{Delboys}
\Delta^{(2j+1)}_n < \Delta^{(2j)}_n<\frac{1}{j+1} \qquad \mathrm{for} \,\,n\geq 0
\eeq 
is valid for all $j\geq 1$. 
Indeed, the first inequality above is always true (from 
Lemma \ref{abounds}), while for the second one the 
inductive step is to use \eqref{delid} to obtain 
$$
\begin{array}{rcl}
\Delta^{(2j+2)}_n &= &\rho^{(2j+2)}_n  \rho^{(2j+1)}_n\big(\eps n \Delta^{(2j+1)}_{n-1} + \eps (n+2) \Delta^{(2j+1)}_{n+1}\big) \\ 
& < & \frac{j+1}{2\eps (n+1)(j+2)}\left(\eps n\big( \frac{1}{j+1} \big)+ \eps (n+2) \big(\frac{1}{j+1}\big)\right)=\frac{1}{j+2},
\end{array} 
$$
as required (where we used \eqref{rhobs} and the inductive 
hypothesis to get the inequality). Taking $j\to\infty$ in  \eqref{Delboys}, 
it then follows that 
$\lim_{j\to\infty} \Delta^{(2j)}_n = 0 = \lim_{j\to\infty} \Delta^{(2j+1)}_n$   
for all $n\geq 0$, yielding the desired squeezing argument, and thus we have 
\begin{cor}
\label{conjco}
If Conjecture \ref{pos} is valid, then there is a unique positive solution 
of \eqref{dpi} with initial data 
\eqref{inits} whenever $\eps$ lies in the range \eqref{estar}. 
\end{cor}

We have tried in vain to provide a direct proof of Theorem \ref{positive}, based only on the properties of the mapping $T$. 
The best we have so far is  Corollary \ref{conjco}, which relies 
on an unproven assumption. Fortunately, all is not lost, because in the next section 
the connection  between \eqref{dpi} and the space of initial conditions for Painlev\'e V will be made manifest. Consequently,  this will 
lead to identifying the initial value problem \eqref{inits} 
with certain classical solutions of  Painlev\'e V, and resulting not only in a proof that the positive solution is unique, but also in an explicit formula for this solution. 

Before switching gears and moving on to consider the geometry of \eqref{dpi}, 
we present one more technical result concerning positive solutions. 
\begin{propn}\label{asyvn}
For each $k\geq 0$, any positive solution of  \eqref{dpi} satisfies 
\beq\label{asyvbds} 
v_n = \sum_{i=0}^k (-1)^i s_{n,i}\,\eps^{i+1} + O(\eps^{k+2}) \qquad \mathrm{as} \,\,\eps\to 0, 
\eeq 
where the finite sum coincides with the first $k+1$ non-zero terms in the Taylor expansion of the rational function $\rb_n^{(k)}$ at $\eps=0$, that is 
\beq\label{btaylor} 
\rb_n^{(k)} (\eps) = \sum_{i=0}^k (-1)^i s_{n,i}\,\eps^{i+1} + O(\eps^{k+2}) , 
\eeq 
where the coefficients $s_{n,i}$ depend only on $n$. 
In particular, for all $n\geq 0$,  
$s_{n,0}=n+1$, 
$s_{n,1}=2(n+1)^2$, 
$s_{n,2}=8(n+1)^3+4(n+1)$, 
independent of $k$. 
\end{propn}
\begin{prf}
By Lemma \ref{delezero}, for all $k\geq 0$ and for each 
$n\geq 0$ the Taylor expansions of 
$\rho_n^{(k)}$ and $\rho_n^{(k+1)}$ at $\eps=0$ agree up to and including terms of order $\eps^k$, which implies that the 
corresponding Taylor expansions of 
$\rb_n^{(k)}$ and 
$\rb_n^{(k+1)}$ agree as far as order $\eps^{k+1}$, with their first 
$k+1$ non-zero terms depending only on $n$, as in \eqref{btaylor}. 
Then by Proposition \ref{existence}, together the bounds in Lemma \ref{abounds}, a positive solution must satisfy 
$$
 v_n -  \sum_{i=0}^k (-1)^i s_{n,i}\,\eps^{i+1}  = O(\eps^{k+2}) \qquad \mathrm{as} \,\,\eps \to 0, 
$$
for each $k$. From \eqref{rtaylor} it follows that 
$$ 
\rb_n^{(k)} = (n+1)\eps \, \big( 1- 2(n+1)\eps +O(\eps^2) \big) 
$$
for $k\geq 1$, 
giving the stated expressions 
for $s_{n,0}$ and $s_{n,1}$, while  
$s_{n,2}$ is obtained from expanding $\rb_n^{(k)}$ in \eqref{firstbs} 
up to $O(\eps^3)$; but for $k\geq 3$ these coefficients do not have a 
uniform expression in $n$.
\end{prf}

The expansions \eqref{asyvbds} extend to an asymptotic series 
\beq\label{vnseries}
v_n \sim \sum_{i=0}^\infty (-1)^i s_{n,i}\,\eps^{i+1} 
\eeq 
for each $n$. In particular, when  
$n=0$ the series is 
\beq\label{vasyo}
v_0 \sim \eps -2\eps^2 + 12\eps^3-112\eps^4+1392\eps^5-21472\eps^6+\ldots \qquad \mathrm{as}\quad \eps\to 0 . 
\eeq 
This should be compared with the Taylor expansion at $\eps=0$ of \eqref{vapprox} 
when $n=0$, 
that is 
\beq\label{vbar0} 
\hat{v}_0(\eps) =\tfrac{1}{4}\left(\sqrt{1+8\eps}-1\right) = \eps - 2\eps^2 +8\eps^3 - 40\eps^4+\ldots . 
\eeq 

The bounds $\rb_n^{(k)}$ provide a sequence of  rational  
approximations which alternate between upper/lower bounds for even/odd $k$. This is highly reminiscent of the situation for the convergents of Stieltjes-type continued fractions (S-fractions), which also provide successive upper/lower approximations based on a formal series. 
Indeed, such a  fraction can be associated  with  the series \eqref{vnseries}, which from the stated expressions for  
$s_{n,0}$, $s_{n,1}$ and 
$s_{n,2}$ must begin as follows: 
\beq\label{contf}
v_n 
= \cfrac{(n+1)\eps}{ 1 +\cfrac {2(n+1)\eps}{ 1+\cfrac{2\Big(n+1+(n+1)^{-1}\Big)\eps}{1+ \cdots} } }\,\,, 
\eeq 
In particular, setting $n=0$  gives the continued  fraction for $v_0$, of  the form 
\beq\label{v0frac}
v_0=\frac{\eps}{1+v_1} 
= \cfrac{\eps}{ 1 +\cfrac {2\eps}{ 1+\cfrac{4\eps}{ 1+\cfrac{5\eps}{ 1+\cfrac{7\eps}{1+ \cdots}}} } }\,\, , 
\eeq 
whose coefficients will be described explicitly in due course, 
towards the end of  section \ref{classical}. 

In fact,  
the continued fraction \eqref{v0frac} 
will turn out to provide us with the missing step in the 
proof of  uniqueness of the positive solution. Furthermore, 
for each $\eps>0$ it will 
precisely identify this solution as being the one specified by the initial 
condition 
$$
v_0 = \tfrac{1}{2}\left(\frac{K_{5/6}(\tfrac{1}{2}t)}{K_{1/6}(\tfrac{1}{2}t)}-1\right), 
\qquad \mathrm{where} \quad t=\frac{1}{3\eps}, 
$$  
The latter function is plotted against $\eps$ in the panel on the right-hand side of Fig.\ref{v0plots}, with a set of numerical computations using the iterated map $T$ appearing as dots on top of this curve. 
A key feature of our analysis will be to use the fact that 
this function  $v_0$ satisfies the Riccati equation 
\beq\label{ric} 
3\eps^2 \frac{\rd v_0}{\rd \eps} = \eps (1+2v_0)-v_0-v_0^2, 
 \eeq 
which gives a rapid way to generate the expansion \eqref{vasyo} recursively, and similarly  $v_1$ satisfies 
\beq\label{ric2} 
3\eps^2 \frac{\rd v_1}{\rd \eps} = \eps (2+v_1)-v_1-v_1^2. 
 \eeq 
 These Riccati equations arise as special reductions of the 
 Painlev\'e V equation, which we proceed to extract from the geometry 
 of the discrete equation \eqref{dpi} in the next section.

\section{Complex geometry of the dP$_{\rm I}$ equation}
\label{complex} 

\setcounter{equation}{0}

In this section, we obtain the space of initial conditions for the dP$_{\rm I}$ equation, and show 
that it corresponds to a particular case of the surface type $D_5^{(1)}$, with symmetry 
type $A_3^{(1)}$, coinciding with that for the \p\ V equation. 
%
Our goal is to describe the geometry of the equation \eqref{dpi}, 
and then use it to study some of its special solutions. 
This equation is a special case of the more general dP$_{\rm I}$ equation
\begin{equation}\label{eq:dP-I}
	x_{n+1} + x_{n-1} = \frac{\tal n + \tbe}{x_{n}} + \tgam
\end{equation}
for $\tal = \tbe=\eps$ 
and $\tgam = - 1$, so we do the general case and then specialise.
This equation has been studied in \cite{rg,grp,tgr,grsw,grw};
see also \cite{KajNouYam:2017:GAPE}.

We first construct the family of Sakai surfaces regularizing the dynamics \eqref{eq:dP-I}. 
For that, we begin by rewriting the recurrence as a first-order discrete dynamical system on 
$\mathbb{C}^{1} \times \mathbb{C}^{1}$ via $y_{n}:=x_{n+1}$. Then equation \eqref{eq:dP-I}
can be rewritten as $y_{n} + x_{n-1} = (\tal n + \tbe)/x_{n} + \tgam$ and we get the 
mapping
\begin{equation}
\ph_{n}(x_{n},y_{n}) = \left(y_{n},\frac{\tal(n+1) + \tbe}{y_{n}} + \tgam - x_{n}\right).
\end{equation}
Using the notation 
$\overline{x}_{n} = x_{n+1}$, $\underline{x}_{n} = x_{n-1}$, and same for $y_{n}$, we can omit the 
indices and rewrite our mapping and its inverse as
\begin{equation}\label{eq:maps}
	\ph(x,y) = \left(y,\frac{\tal(n+1) + \tbe}{y} + \tgam - x\right),\qquad 
	\ph^{-1}(x,y) = \left(\frac{\tal n + \tbe}{y} + \tgam - y,x\right).
\end{equation}
Next, we extend the mapping to $\mathbb{P}^{1} \times \mathbb{P}^{1}$ via the  introducion of  coordinates
at infinity, $X = 1/x$ and $Y = 1/y$, and then rewrite the mappings \eqref{eq:maps} from 
the affine chart $(x,y)$ to three other charts, $(X,y)$, $(x,Y)$, and $(X,Y)$. It is easy to see
that there are four base points where either rational mapping becomes undefined, which are
\begin{equation}
	\left(X=0,y=0\right),\quad (X=0,y=\tgam),\quad (x=0,Y=0),\quad (x=\tgam,Y=0).
\end{equation}
As usual, these singularities are resolved by using blowups, where blowing up
a point $q_{i}(x_{i},y_{i})$ amounts to introducing two new coordinate charts 
$(\ru_{i},\rv_{i})$ and $(U_{i},V_{i})$ via 
$x = x_{i} + \ru_{i} = x_{i} + U_{i} V_{i}$ and $y = y_{i} + \ru_{i}\rv_{i} = y_{i} + V_{i}$. 
We then extend the mapping to these new charts, check for new base points, resolve them in 
the same way, and continue this process until the mapping becomes defined everywhere. This
process, for discrete \p\ equations, terminates in a finite number of steps, and in our case, as usual, we get 
$8$ base points that come in four degeneration cascades,
\begin{equation}\label{eq:basepts-dP}
	\begin{aligned}
		q_{1}\left(X = \frac{1}{x}=0,y=\tgam\right) &\leftarrow q_{2}\left(\ru_{1} = X = \frac{1}{x} = 0,\rv_{1} = x(y - \tgam) = \tal\right),\\ 
			 q_{3}\left(X = \frac{1}{x} = 0,y=0\right)&\leftarrow q_{4}\left(\ru_{3} = X = \frac{1}{x} = 0, \rv_{3} = xy = (n+1)\tal + \tbe\right), \\
		q_{5}\left(x=0,Y = \frac{1}{y} = 0\right) &\leftarrow q_{6}\left(U_{5} = xy = n \tal + \tbe, V_{5} = Y = \frac{1}{y}= 0\right),\\ 
			 q_{7}\left(x=\tgam,Y = \frac{1}{y} = 0\right) &\leftarrow q_{8}\left(U_{7} = y(x-\tgam) = -\tal, V_{7} = Y = \frac{1}{y} = 0\right).
	\end{aligned}
\end{equation}
This point configuration is easily identified as a point configuration for the $D_{5}^{(1)}$ Sakai surface family. This family can also 
be thought of as the Okamoto space of initial conditions for the differential  equation $\Pain{V}$ in \eqref{pv}. 
From the point of view of geometry
it is convenient to consider the Hamiltonian version of this equation in the form given in \cite{KajNouYam:2017:GAPE}, namely 
\begin{equation}\label{eq:KNY-Ham5-sys}
	\left\{
	\begin{aligned}
		\deriv{q}{t} &= \frac{1}{t}\Big(q(q-1)(2p+t) - a_{1}(q-1) - a_{3}q\Big) = \frac{\partial H}{\partial p},\\
		\deriv{p}{t} &= \frac{1}{t}\Big(p(p+t)(1-2q) + (a_{1} + a_{3})p - a_{2}t\Big) = -\frac{\partial H}{\partial q},
	\end{aligned}
	\right.
\end{equation}
where the Hamiltonian function is 
\begin{equation}\label{eq:KNY-Ham5}
	H(q,p;t) = \frac{1}{t}\Big(q (q-1) p(p+t) - (a_{1} + a_{3})q p + a_{1}p + a_{2} t q\Big),	
\end{equation}
and the symplectic form is the usual one, $\omega = \rd p\wedge \rd q$. The parameters $a_{i}$ are the so-called \emph{root variables} that are related
to the standard \p\ parameters via
\begin{equation}\label{eq:root-pars-PV}
	{\al} = \tfrac{1}{2}{a_{1}^{2}},\quad {\be} = -\tfrac{1}{2}{a_{3}^{2}},\quad {\gam} = a_{0}-a_{2}, \quad {\delta} = -\tfrac{1}{2},\qquad
	\text{where }\quad a_{0} + a_{1} + a_{2} + a_{3} = 1.
\end{equation}
The system \eqref{eq:KNY-Ham5-sys} is equivalent to \eqref{pv} for the function $w(t) = 1 - {1}/{q(t)}$. 
The Okamoto space of initial conditions
for this system is obtained by removing the irreducible components of the anti-canonical divisor (the so-called vertical leaves) from 
the standard realization of the $D_{5}^{(1)}$-surface family. This family 
\begin{equation*}
\mathcal{X} = \mathcal{X}_{\mathbf{a}=(a_{0},a_{1},a_{2},a_{3})} = \operatorname{Bl}_{p_{1},\ldots,p_{n}}\left(\mathbb{P}^{1} \times \mathbb{P}^{1}\right),	
\end{equation*}
is obtained by blowing up $\mathbb{P}^{1} \times \mathbb{P}^{1}$
at the points
\begin{equation}\label{eq:basepts-std}
	p_{1}(\infty,-t)\leftarrow p_{2}(0,-a_{0}),\quad p_{3}(\infty,0)\leftarrow p_{4}(0,-a_{2}),\quad p_{5}(0,\infty)\leftarrow p_{6}(a_{1},0),\quad 
	p_{7}(1,\infty) \leftarrow p_{8}(a_{3},0).
\end{equation}
The latter point configuration and the resulting blowup surface 
are shown in Figure~\ref{fig:KNY-D5-surface}.
\begin{figure}[ht]
	\begin{tikzpicture}[>=stealth,basept/.style={circle, draw=red!100, fill=red!100, thick, inner sep=0pt,minimum size=1.2mm}]
	\begin{scope}[xshift=0cm,yshift=0cm]
	\draw [black, line width = 1pt] (-0.4,0) -- (2.9,0)	node [pos=0,left] {\small $H_{p}$} node [pos=1,right] {\small $p=0$};
	\draw [black, line width = 1pt] (-0.4,2.5) -- (2.9,2.5) node [pos=0,left] {\small $H_{p}$} node [pos=1,right] {\small $p=\infty$};
	\draw [black, line width = 1pt] (0,-0.4) -- (0,2.9) node [pos=0,below] {\small $H_{q}$} node [pos=1,above] {\small $q=0$};
	\draw [black, line width = 1pt] (2.5,-0.4) -- (2.5,2.9) node [pos=0,below] {\small $H_{q}$} node [pos=1,above] {\small $q=\infty$};
	\node (p3) at (2.5,0) [basept,label={[xshift = -8pt, yshift=-15pt] \small $p_{3}$}] {};
	\node (p4) at (3,0.5) [basept,label={[yshift=0pt] \small $p_{4}$}] {};
	\node (p1) at (2.5,1) [basept,label={[xshift = -8pt, yshift=-15pt] \small $p_{1}$}] {};
	\node (p2) at (3,1.5) [basept,label={[yshift=0pt] \small $p_{2}$}] {};
	\node (p5) at (0,2.5) [basept,label={[xshift = 8pt, yshift=0pt] \small $p_{5}$}] {};
	\node (p6) at (-.5,2) [basept,label={[yshift=-15pt] \small $p_{6}$}] {};
	\node (p7) at (1.5,2.5) [basept,label={[xshift = 8pt, yshift=0pt] \small $p_{7}$}] {};
	\node (p8) at (1,2) [basept,label={[yshift=-15pt] \small $p_{8}$}] {};
	\draw [red, line width = 0.8pt, ->] (p2) -- (p1);
	\draw [red, line width = 0.8pt, ->] (p4) -- (p3);
	\draw [red, line width = 0.8pt, ->] (p6) -- (p5);
	\draw [red, line width = 0.8pt, ->] (p8) -- (p7);	
	\end{scope}
	\draw [->] (6.5,1)--(4.5,1) node[pos=0.5, below] {$\operatorname{Bl}_{p_{1}\cdots p_{8}}$};
	\begin{scope}[xshift=9cm,yshift=0cm]
	\draw [red, line width = 1pt] (-0.4,0) -- (3.5,0)	node [pos=0, left] {\small $H_{p}-E_{3}$};
	\draw [red, line width = 1pt] (0,-0.4) -- (0,2.4) node [pos=0, below] {\small $H_{q}-E_{5}$};
	\draw [blue, line width = 1pt] (-0.2,1.8) -- (0.8,2.8) node [pos=0, left] {\small $E_{5}-E_{6}$};
	\draw [red, line width = 1pt] (-0.1,2.7) -- (0.4,2.2) node [pos=0, above] {\small $E_{6}$};
	\draw [blue, line width = 1pt] (1.2,1.8) -- (2.2,2.8) node [pos=0, xshift=-14pt, yshift=-5pt] {\small $E_{7}-E_{8}$};
	\draw [red, line width = 1pt] (1.6,2.4) -- (2.1,1.9) node [pos=1, below] {\small $E_{8}$};
	\draw [blue, line width = 1pt] (0.3,2.6) -- (4.2,2.6) node [pos=1,right] {\small $H_{p} - E_{5} - E_{7}$};
	\draw [blue, line width = 1pt] (3,-0.2) -- (4,0.8) node [pos=1,right] {\small $E_{3} - E_{4}$};
	\draw [red, line width = 1pt] (3.4,0.4) -- (3.9,-0.1) node [pos=1, below] {\small $E_{4}$};
	\draw [blue, line width = 1pt] (3.8,0.3) -- (3.8,3) node [pos=1, above] {\small $H_{q}-E_{1} - E_{3}$};	
	\draw [blue, line width = 1pt] (3,1) -- (4,2) node [pos=1,right] {\small $E_{1} - E_{2}$};
		\draw [red, line width = 1pt] (3.1,1.9) -- (3.6,1.4) node [pos=0, above] {\small $E_{2}$};
	\draw [red, line width = 1pt] (-0.4,1.2) -- (3.5,1.2)	node [pos=0, left] {\small $H_{p}-E_{1}$};
	\draw [red, line width = 1pt] (1.4,-0.4) -- (1.4,2.4) node [pos=0, below] {\small $H_{q}-E_{7}$};
	\end{scope}
	\end{tikzpicture}
	\caption{The standard $D_{5}^{(1)}$ Sakai surface family (the space of initial conditions for system \eqref{eq:KNY-Ham5-sys})}
	\label{fig:KNY-D5-surface}
\end{figure}
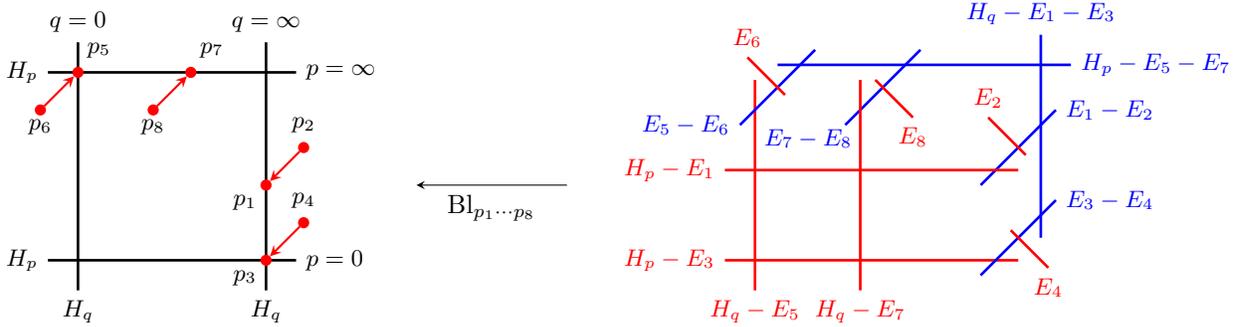

Note that configurations \eqref{eq:basepts-dP} and \eqref{eq:basepts-std} almost match --- geometrically they are the same up 
to the action of rescaling on the axes. Thus, let $q = \lambda x$ and $p= \mu y$. Then the symplectic form is 
$\omega = \rd p \wedge \rd q = \lambda \mu \rd y \wedge \rd x$ and matching coordinates of the base points results in
$$
	\mu \tgam = -t,\quad \lambda \mu \tal = -a_{0},\quad \lambda \mu (n \tal + \tbe) = a_{1},\quad 
	\lambda \mu ((n+1) \tal + \tbe) = -a_{2}, \lambda \tgam = 1,\quad \lambda \mu (-\tal) = a_{3}.
$$
Thus, $\lambda = \tgam^{-1}$ and the root variable normalization $a_{0} + a_{1} + a_{2} + a_{3} = - 3 \tal \lambda \mu = 1$ gives
$\mu = - \tgam (3 \tal)^{-1}$. The root variables describing configurations \eqref{eq:basepts-dP} then are
\begin{equation}
	a_{0} = \frac{1}{3},\quad a_{1} = - \frac{n}{3} - \frac{\tbe}{3 \tal},\quad a_{2} = \frac{n+1}{3} + \frac{\tbe}{3 \tal},\quad a_{3} = \frac{1}{3}.
\end{equation}
In particular, we see that we get an integer translation in the root variables only after three iterations. This can also be seen directly by computing
the push-forward $\ph_{*}$ and the pull-back $\ph^{*} = (\ph_{*})^{-1}$ actions
of the mapping \eqref{eq:maps} on the Picard lattice $\operatorname{Pic}(\mathcal{X})$. Specifying the basis as 
\begin{equation}\label{eq:root-vars-dP}
	\operatorname{Pic}(\mathcal{X}) = \operatorname{Span}_{\mathbb{Z}}\{\mathcal{H}_{x}, \mathcal{H}_{y},\mathcal{E}_{1},\ldots,\mathcal{E}_{8}\} , 
\end{equation}
the corresponding actions of $\ph_{*}$ and  $\ph^{*}$
are given by 
\begin{equation}\label{eq:pic-action}
	\begin{alignedat}{2}
	\mathcal{H}_{y} &\overset{\ph^{*}}\longleftarrow &\mathcal{H}_{x} &\overset{\ph_{*}}\longrightarrow \mathcal{H}_{x} + \mathcal{H}_{y} - \mathcal{E}_{5} - \mathcal{E}_{6},\\
	\mathcal{H}_{x} + \mathcal{H}_{y} - \mathcal{E}_{3} - \mathcal{E}_{4} &\longleftarrow &\mathcal{H}_{y} &\longrightarrow \mathcal{H}_{x},\\		
	\mathcal{E}_{5} &\longleftarrow &\mathcal{E}_{1} &\longrightarrow \mathcal{E}_{7},\\		
	\mathcal{E}_{6} &\longleftarrow &\mathcal{E}_{2} &\longrightarrow \mathcal{E}_{8},\\		
	\mathcal{E}_{7} &\longleftarrow &\mathcal{E}_{3} &\longrightarrow \mathcal{H}_{x} - \mathcal{E}_{6},\\		
	\mathcal{E}_{8} &\longleftarrow &\mathcal{E}_{4} &\longrightarrow \mathcal{H}_{x} - \mathcal{E}_{5},\\		
	\mathcal{H}_{y} - \mathcal{E}_{4} &\longleftarrow &\mathcal{E}_{5} &\longrightarrow \mathcal{E}_{1},\\		
	\mathcal{H}_{y} - \mathcal{E}_{3} &\longleftarrow &\mathcal{E}_{6} &\longrightarrow \mathcal{E}_{2},\\		
	\mathcal{E}_{1} &\longleftarrow &\mathcal{E}_{7} &\longrightarrow \mathcal{E}_{3},\\		
	\mathcal{E}_{2} &\longleftarrow &\mathcal{E}_{8} &\longrightarrow \mathcal{E}_{4}.		
	\end{alignedat}
\end{equation}

We initially use the same choice as in \cite{KajNouYam:2017:GAPE} for the surface root basis, shown in Figure~\ref{fig:d-roots-d5-KNY}, 
and the symmetry root bases, shown on Figure~\ref{fig:a-roots-a3-KNY}.	
\begin{figure}[ht]
\begin{equation}\label{eq:d-roots-d51}			
	\raisebox{-32.1pt}{\begin{tikzpicture}[
			elt/.style={circle,draw=black!100,thick, inner sep=0pt,minimum size=2mm}]
		\path 	(-1,1) 	node 	(d0) [elt, label={[xshift=-10pt, yshift = -10 pt] $\delta_{0}$} ] {}
		 (-1,-1) node 	(d1) [elt, label={[xshift=-10pt, yshift = -10 pt] $\delta_{1}$} ] {}
		 ( 0,0) 	node 	(d2) [elt, label={[xshift=-10pt, yshift = -10 pt] $\delta_{2}$} ] {}
		 ( 1,0) 	node 	(d3) [elt, label={[xshift=10pt, yshift = -10 pt] $\delta_{3}$} ] {}
		 ( 2,1) 	node 	(d4) [elt, label={[xshift=10pt, yshift = -10 pt] $\delta_{4}$} ] {}
		 ( 2,-1) node 	(d5) [elt, label={[xshift=10pt, yshift = -10 pt] $\delta_{5}$} ] {};
		\draw [black,line width=1pt ] (d0) -- (d2) -- (d1) (d2) -- (d3) (d4) -- (d3) -- (d5);
	\end{tikzpicture}} \qquad
			\begin{alignedat}{2}
			\delta_{0} &= \mathcal{E}_{1} - \mathcal{E}_{2}, &\qquad \delta_{3} &= \mathcal{H}_{p} - \mathcal{E}_{5} - \mathcal{E}_{7},\\
			\delta_{1} &= \mathcal{E}_{3} - \mathcal{E}_{4}, &\qquad \delta_{4} &= \mathcal{E}_{5} - \mathcal{E}_{6},\\
			\delta_{2} &= \mathcal{H}_{q} - \mathcal{E}_{1} - \mathcal{E}_{3}, &\qquad \delta_{5} &= \mathcal{E}_{7} - \mathcal{E}_{8}.
			\end{alignedat}
\end{equation}
	\caption{The surface root basis for the standard $D_{5}^{(1)}$ Sakai surface point configuration}
	\label{fig:d-roots-d5-KNY}	
\end{figure}
\begin{figure}[ht]
\begin{equation}\label{eq:a-roots-a31}			
	\raisebox{-32.1pt}{\begin{tikzpicture}[
			elt/.style={circle,draw=black!100,thick, inner sep=0pt,minimum size=2mm}]
		\path 	(-1,1) 	node 	(a0) [elt, label={[xshift=-10pt, yshift = -10 pt] $\al_{0}$} ] {}
		 (-1,-1) node 	(a1) [elt, label={[xshift=-10pt, yshift = -10 pt] $\al_{1}$} ] {}
		 ( 1,-1) node 	(a2) [elt, label={[xshift=10pt, yshift = -10 pt] $\al_{2}$} ] {}
		 ( 1,1) 	node 	(a3) [elt, label={[xshift=10pt, yshift = -10 pt] $\al_{3}$} ] {};
		\draw [black,line width=1pt ] (a0) -- (a1) -- (a2) -- (a3) -- (a0); 
		\draw [purple, dashed, line width = 0.5pt] (0,-1.5) -- (0,1.2)	node [pos=0,below] {\small $\sigma_{1}$};
		\draw [purple, dashed, line width = 0.5pt] (-1.5,-1.5) -- (1.2,1.2) node [pos=0,below left] {\small $\sigma_{2}$};
		\draw [purple, dashed, line width = 0.5pt] (1.5,-1.5) -- (-1.2,1.2) node [pos=0,below right] {\small $\sigma_{3}$};		
	\end{tikzpicture}} \qquad
			\begin{alignedat}{2}
			\al_{0} &= \mathcal{H}_{p} - \mathcal{E}_{1} - \mathcal{E}_{2}, &\qquad \al_{2} &= \mathcal{H}_{p} - \mathcal{E}_{3} - \mathcal{E}_{4},\\
			\al_{1} &= \mathcal{H}_{q} - \mathcal{E}_{5} - \mathcal{E}_{6}, &\qquad \al_{3} &= \mathcal{H}_{q} - \mathcal{E}_{7} - \mathcal{E}_{8}.
			\\[5pt]
			\hat{\delta} & = \mathrlap{\al_{0} + \al_{1} + \al_{2} + \al_{3}.} 
			\end{alignedat}
\end{equation}
	\caption{The symmetry root basis for the standard $A_{3}^{(1)}$ symmetry sub-lattice}
	\label{fig:a-roots-a3-KNY}	
\end{figure}
Then the action of the mapping $\ph_{*}$ on the symmetry roots is
\begin{equation}\label{eq:f-act-sym}
	\ph_{*}: \upalpha = \langle \al_{0},\al_{1},\al_{2},\al_{3} \rangle \mapsto 
		\langle \al_{3}, \al_{0} + \al_{1}, - \al_{1}, \al_{1} + \al_{2} \rangle,
\end{equation}
which is \emph{not} a translation on the symmetry sub-lattice. However, it is a \emph{quasi-translation}, since after \emph{three} iterations we get a translation, that is 
\begin{equation}\label{eq:f3-act-sym}
	\ph_{*}^{3}: \upalpha = \langle \al_{0},\al_{1},\al_{2},\al_{3} \rangle \mapsto 
		\upalpha + \langle 0,1,-1,0 \rangle \hat{\delta},\qquad 
        \hat{\delta} = - \mathcal{K}_{\mathcal{X}} = \al_{0}+\al_{1} + \al_{2} + \al_{3}.
\end{equation}
There are two standard and non-conjugate examples of discrete \p\ equations in the $D_{5}^{(1)}$-family: the equation given in  
\cite[(8.23)]{KajNouYam:2017:GAPE}, which acts on the symmetry roots as a translation
\begin{equation}\label{eq:d5-KNY-transl}
		\psi_{*}: \upalpha = \langle \al_{0},\al_{1},\al_{2},\al_{3} \rangle \mapsto \upalpha + \langle -1,1,-1,1 \rangle\, \hat{\delta};
\end{equation}
and the equation in \cite[(2.33-2.34)]{Sak:2007:PDPETLF} that acts on the symmetry roots as a translation
\begin{equation}\label{eq:d5-Sakai-transl}
		\phi_{*}: \upalpha = \langle \al_{0},\al_{1},\al_{2},\al_{3} \rangle \mapsto \upalpha + \langle -1,0,0,1 \rangle \,\hat{\delta}.
\end{equation}
We see that the cube of our mapping is conjugate, by the half-turn rotation of the Dynkin diagrams, 
to the dynamics considered by Sakai, also called the
dP$_{\rm IV}$ equation in the original Sakai paper \cite{Sak:2001:RSAWARSGPE}.

The dynamical system \eqref{eq:dP-I} is generated by the 
birational action of the fully extended affine Weyl group of type $A_{3}^{(1)}$, that is 
$$\widehat{W}\left(A_{3}^{(1)}\right) := W\left(A_{3}^{(1)}\right) \rtimes \operatorname{Aut}\left(A_{3}^{(1)}\right),$$ 
where
the above semi-direct product structure 
is given by the action of 
$\sigma \in \operatorname{Aut}\left(A_{3}^{(1)}\right)$ on 
$W\left(A_{3}^{(1)}\right)$ via $w_{\sigma(\alpha_{i})} = \sigma w_{\alpha_{i}} \sigma^{-1}$. The fully extended affine Weyl group  
acts on point configurations by elementary birational maps on $(q,p)$ and root variables $\boldsymbol{a}$.
This is known as a birational representation of 
this group, 
and its action  
by automorphisms of $\mathcal{X}$ is called the \emph{Cremona action} \cite{Sak:2001:RSAWARSGPE}. 
We describe this birational representation in the following lemma (\cite[Section A.3]{CheDzhHu:2020:PLUEDPE}).

\begin{lem}\label{thm:bir-weyl-d5} The birational representation of $\widehat{W}\left(A_{3}^{(1)}\right)$, written in the affine 
	$(q,p)$-chart and the root variables $a_{i}$, is as follows:  
	reflections $w_{i}$ on $\operatorname{Pic}(\mathcal{X})$ are induced by the elementary 
	birational mappings, also denoted by $w_{i}$, given by 
	\begin{alignat*}{2}
		w_{0}&: 
		\left(\begin{matrix} a_{0} & a_{1} \\ a_{2} & a_{3} \end{matrix}\ ;\ t\ ;
		\begin{matrix} q \\ p \end{matrix}\right) 
		&&\mapsto 
		\left(\begin{matrix} -a_{0} & a_{0} + a_{1} \\ a_{2} & a_{0} + a_{3} \end{matrix}\ ;\ t\ ;
		\begin{matrix} \displaystyle q + \frac{a_{0}}{p + t} \\ p \end{matrix}\right), \\
		w_{1}&: \left(\begin{matrix} a_{0} & a_{1} \\ a_{2} & a_{3} \end{matrix}\ ;\ t\ ;
		\begin{matrix} q \\ p \end{matrix}\right)
		&&\mapsto 
		\left(\begin{matrix} a_{0} + a_{1} & -a_{1} \\ a_{1} + a_{2} & a_{3} \end{matrix}\ ;\ t\ ;
		\begin{matrix}  q \\ \displaystyle p - \frac{a_{1}}{q} \end{matrix}\right), \\
		w_{2}&: 
		\left(\begin{matrix} a_{0} & a_{1} \\ a_{2} & a_{3} \end{matrix}\ ;\ t\ ;
		\begin{matrix} q \\ p \end{matrix}\right) 
		&&\mapsto 
		\left(\begin{matrix} a_{0} & a_{1} + a_{2} \\ -a_{2} & a_{2} + a_{3} \end{matrix}\ ;\ t\ ;
		\begin{matrix} \displaystyle q + \frac{a_{2}}{p}\\ p \end{matrix}\right), \\
		w_{3}&: 
		\left(\begin{matrix} a_{0} & a_{1} \\ a_{2} & a_{3} \end{matrix}\ ;\ t\ ;
		\begin{matrix} q \\ p \end{matrix}\right) 
		&&\mapsto 
		\left(\begin{matrix} a_{0}+a_{3} & a_{1} \\ a_{2}+a_{3} & -a_{3} \end{matrix}\ ;\ t\ ;
		\begin{matrix} q \\ \displaystyle  p - \frac{a_{3}}{q-1} \end{matrix}\right).
	\end{alignat*}	
    Note that the parameter $t$ can also change when we consider  Dynkin diagram automorphisms, so it is convenient to include it 
   	among the root variables.
	The actions of the generators $\sigma_1,\sigma_2$ of $\operatorname{Aut}\left(A_3^{(1)}\right)$, shown in Figure~\ref{fig:a-roots-a3-KNY}, as well as $\sigma_3=\sigma_1\sigma_2\sigma_1$, are given by the following birational mappings:
	\begin{alignat*}{2}
		\sigma_{1}&: 
		\left(\begin{matrix} a_{0} & a_{1} \\ a_{2} & a_{3} \end{matrix}\ ;\ t\ ;
		\begin{matrix} q \\ p \end{matrix}\right) 
		&&\mapsto 
		\left(\begin{matrix} a_{3} & a_{2} \\ a_{1} & a_{0}  \end{matrix}\ ;\ -t\ ;
		\begin{matrix} \displaystyle -\frac{p}{ t} \\ q t \end{matrix}\right), \\
		\sigma_{2}&: 
		\left(\begin{matrix} a_{0} & a_{1} \\ a_{2} & a_{3} \end{matrix}\ ;\ t\ ;
		\begin{matrix} q \\ p \end{matrix}\right) 
		&&\mapsto 
		\left(\begin{matrix} a_{2} & a_{1} \\ a_{0} & a_{3}  \end{matrix}\ ;\ -t\ ;
		\begin{matrix} q \\ p + t \end{matrix}\right), \\
		\sigma_{3}&: 
		\left(\begin{matrix} a_{0} & a_{1} \\ a_{2} & a_{3} \end{matrix}\ ;\ t\ ;
		\begin{matrix} q \\ p \end{matrix}\right) 
		&&\mapsto 
		\left(\begin{matrix} a_{0} & a_{3} \\ a_{2} & a_{1}  \end{matrix}\ ;\ -t\ ;
		\begin{matrix} 1-q \\ -p  \end{matrix}\right).
	\end{alignat*}

\end{lem}

Using standard techniques (see, e.g., \cite{DzhFilSto:2020:RCDOPWHWDPE}) the mapping \eqref{eq:dP-I} and its inverse decompose, in terms of generators, as
\begin{equation}\label{eq:dP-decomp}
	\begin{aligned}
	\varphi &= \sigma_{2}\sigma_{1}w_{2}: 	
		\left(\begin{matrix} a_{0} & a_{1} \\ a_{2} & a_{3} \end{matrix}\ ;\ t\ ;
		\begin{matrix} q \\ p \end{matrix}\right) \mapsto 
	\left(\begin{matrix} a_{1} + a_{2} & -a_{2} \\ a_{2} + a_{3} & a_{0} \end{matrix}\ ;\ t\ ;
	\begin{matrix} -\frac{p}{t} \\ t\left(q + \frac{a_{2}}{p} - 1\right) \end{matrix}\right), \\
	\varphi^{-1} &= w_{2}\sigma_{1}\sigma_{2}: 	
		\left(\begin{matrix} a_{0} & a_{1} \\ a_{2} & a_{3} \end{matrix}\ ;\ t\ ;
		\begin{matrix} q \\ p \end{matrix}\right) \mapsto 
	\left(\begin{matrix} a_{3} & a_{0} + a_{1} \\ -a_{1} & a_{1} + a_{2} \end{matrix}\ ;\ t\ ;
	\begin{matrix} 1 + \frac{p}{t} - \frac{a_{1}}{q t} \\ - qt \end{matrix}\right). 								
	\end{aligned}	
\end{equation}
Using $w(t) = 1 - 1/q(t)$ and \eqref{eq:KNY-Ham5-sys}, we can rewrite $\varphi$ and $\varphi^{-1}$ as B\"acklund transformations of a solution 
$w(t)$ of the standard \p\ V equation \eqref{pv}, namely 
\begin{equation}
	\begin{aligned}
		\varphi: w\mapsto w_{+} &= 1 - \frac{1}{\overline{q}} = 1 + \frac{t}{p} = 1 + \frac{2 t w}{t \frac{dw}{dt} - a_{1} w^{2} + (a_{1} - a_{3} -t)w + a_{3} }\\
		\varphi^{-1}: w\mapsto w_{-} & = 1 - \frac{1}{\underline{q}} = 1 - \frac{q t}{qt + qp - a_{1}}	
				= 1 - \frac{2 t w}{t \frac{dw}{dt} + a_{1} w^{2} - (a_{1} + a_{3} -t)w + a_{3} }.
	\end{aligned}
\end{equation}
In the next section, these results will be rederived using 
compositions of classical B\"acklund transformations for 
Painlev\'e V. 
Note that parameters ${\al},{\be},{\gam},{\delta}\in \mathbb{C}$ 
appearing as coefficients in the Painlev\'e V equation are 
\begin{equation}
	{\al} = \frac{(n \tal + \tbe)^{2}}{18 \tal^{2}},\quad {\be} = -\frac{1}{18},\quad {\gam} = - \frac{n \tal + \tbe}{3 \tal}, 
    \quad \delta = -\frac{1}{2}, 
\end{equation}
and further specializing to \eqref{dpi} by taking 
$\tal = \tbe=\eps$ 
and $\tgam = - 1$ gives $\la=-1$, $\mu=1/(3\eps)$ and 
\begin{equation}
\label{PVparams}
		a_{0} = a_{3} = \frac{1}{3},\quad a_{1} = - \frac{n+1}{3},\quad a_{2} = \frac{n+2}{3};\qquad 
		{\al} = \frac{(n+1)^{2}}{18},\quad {\be} = -\frac{1}{18},\quad {\gam} = - \frac{n+1}{3}, 
\end{equation}
where the solution $w$ of \eqref{pv} is related to the iterate $v_n = x_n = \lambda^{-1} q_{n} = \tilde{\gamma} q_{n}$ of  \eqref{dpi} via 
\beq\label{pvteps}
w(t) = 1+\frac{1}{v_n(\eps)}\qquad \mathrm{with}\quad t= - \mu \tilde{\gamma} = \frac{1}{3\eps}. 
\eeq 
An essential observation to make at this stage is that these 
values of the root variables, and the corresponding values of 
$\al,\be,\gam$, fall in a particular region of parameter space where 
the equation \eqref{pv} is known to admit special solutions 
in terms of classical special functions: see \cite[\S32.10(v)]{DLMF}, for instance. 

Geometrically, we can see this as follows. For $n=-1$ the root variable $a_{1}$ vanishes, which corresponds to the appearance of a nodal curve. 
Indeed, the base point $p_{6}$ in equation \eqref{eq:basepts-std} becomes $p_{6}\left(U_{5} = V_{5} = 0\right)$, which changes the 
point configuration and the corresponding blowup picture in Figure~\ref{fig:KNY-D5-surface} to the one in Figure~\ref{fig:KNY-D5-surface-nodal}.
The nodal curve is the $-2$-curve $H_{q}-E_{5} - E_{6}$ disjoint from the anti-canonical divisor. It is important to note that the existence of nodal 
curves is preserved by B\"acklund transformations, and such nodal curves define reductions to Riccati equations, see \cite{Sak:2001:RSAWARSGPE, SaiTer:2004:NCRSPE}.

\begin{figure}[ht]
	\begin{tikzpicture}[>=stealth,basept/.style={circle, draw=red!100, fill=red!100, thick, inner sep=0pt,minimum size=1.2mm}]
	\begin{scope}[xshift=0cm,yshift=0cm]
	\draw [black, line width = 1pt] (-0.4,0) -- (2.9,0)	node [pos=0,left] {\small $H_{p}$} node [pos=1,right] {\small $p=0$};
	\draw [black, line width = 1pt] (-0.4,2.5) -- (2.9,2.5) node [pos=0,left] {\small $H_{p}$} node [pos=1,right] {\small $p=\infty$};
	\draw [black, line width = 1pt] (0,-0.4) -- (0,2.9) node [pos=0,below] {\small $H_{q}$} node [pos=1,above] {\small $q=0$};
	\draw [black, line width = 1pt] (2.5,-0.4) -- (2.5,2.9) node [pos=0,below] {\small $H_{q}$} node [pos=1,above] {\small $q=\infty$};
	\node (p3) at (2.5,0) [basept,label={[xshift = -8pt, yshift=-15pt] \small $p_{3}$}] {};
	\node (p4) at (3,0.5) [basept,label={[yshift=0pt] \small $p_{4}$}] {};
	\node (p1) at (2.5,1) [basept,label={[xshift = -8pt, yshift=-15pt] \small $p_{1}$}] {};
	\node (p2) at (3,1.5) [basept,label={[yshift=0pt] \small $p_{2}$}] {};
	\node (p5) at (0,2.5) [basept,label={[xshift = 8pt, yshift=0pt] \small $p_{5}$}] {};
	\node (p6) at (-.3,1.8) [basept,label={[yshift=-15pt] \small $p_{6}$}] {};
	\node (p7) at (1.5,2.5) [basept,label={[xshift = 8pt, yshift=0pt] \small $p_{7}$}] {};
	\node (p8) at (1,2) [basept,label={[yshift=-15pt] \small $p_{8}$}] {};
	\draw [red, line width = 0.8pt, ->] (p2) -- (p1);
	\draw [red, line width = 0.8pt, ->] (p4) -- (p3);
	\draw [red, line width = 0.8pt, ->] (p6) --  (0,2.1) -- (p5);
	\draw [red, line width = 0.8pt, ->] (p8) -- (p7);	
	\end{scope}
	\draw [->] (6.5,1)--(4.5,1) node[pos=0.5, below] {$\operatorname{Bl}_{p_{1}\cdots p_{8}}$};
	\begin{scope}[xshift=9cm,yshift=0cm]
	\draw [red, line width = 1pt] (-0.4,0) -- (3.5,0)	node [pos=0, left] {\small $H_{p}-E_{3}$};
	\draw [blue, line width = 1pt] (0,-0.4) -- (0,2.4) node [pos=0, xshift = -12pt, yshift=-8pt] {\small $H_{q}-E_{5} - E_{6}$};
	\draw [blue, line width = 1pt] (0.8,2.0) -- (0.8,3) node [pos=1, above] {\small $E_{5}-E_{6}$};
	\draw [red, line width = 1pt] (-0.2,2.2) -- (1,2.2) node [pos=0, left] {\small $E_{6}$};
	\draw [blue, line width = 1pt] (1.2,1.8) -- (2.2,2.8) node [pos=0, xshift=-14pt, yshift=-5pt] {\small $E_{7}-E_{8}$};
	\draw [red, line width = 1pt] (1.6,2.4) -- (2.1,1.9) node [pos=1, below] {\small $E_{8}$};
	\draw [blue, line width = 1pt] (0.3,2.6) -- (4.2,2.6) node [pos=1,right] {\small $H_{p} - E_{5} - E_{7}$};
	\draw [blue, line width = 1pt] (3,-0.2) -- (4,0.8) node [pos=1,right] {\small $E_{3} - E_{4}$};
	\draw [red, line width = 1pt] (3.4,0.4) -- (3.9,-0.1) node [pos=1, below] {\small $E_{4}$};
	\draw [blue, line width = 1pt] (3.8,0.3) -- (3.8,3) node [pos=1, above] {\small $H_{q}-E_{1} - E_{3}$};	
	\draw [blue, line width = 1pt] (3,1) -- (4,2) node [pos=1,right] {\small $E_{1} - E_{2}$};
		\draw [red, line width = 1pt] (3.1,1.9) -- (3.6,1.4) node [pos=0, above] {\small $E_{2}$};
	\draw [red, line width = 1pt] (-0.4,1.2) -- (3.5,1.2)	node [pos=0, left] {\small $H_{p}-E_{1}$};
	\draw [red, line width = 1pt] (1.4,-0.4) -- (1.4,2.4) node [pos=0, below] {\small $H_{q}-E_{7}$};
	\end{scope}
	\end{tikzpicture}
	\caption{Special $D_{5}^{(1)}$ Sakai surface with a nodal curve corresponding to $a_{1}=0$}
	\label{fig:KNY-D5-surface-nodal}
\end{figure}
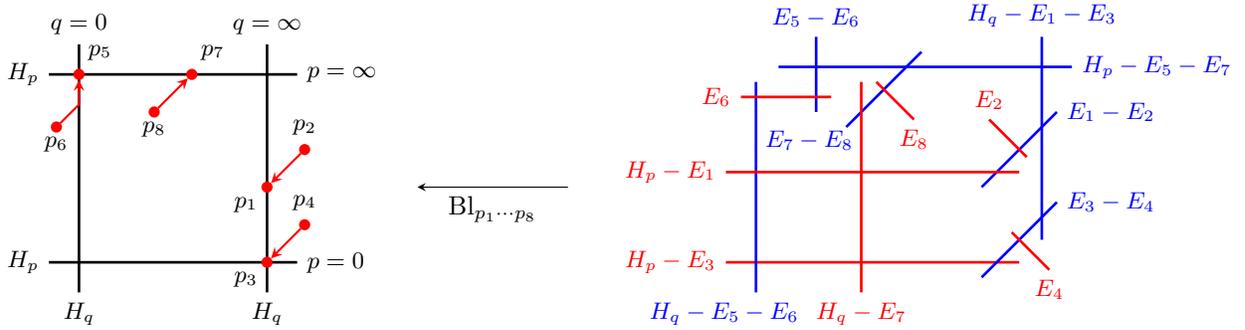
It is now easy to see that the solution we are interested in belongs to the Riccati class. Indeed, the initial condition $v_{-1}=0$ corresponds to $q_{-1}=0$ and 
$v_{0} = 3 \varepsilon p_{-1} = \frac{1}{t} p_{-1}$. The system \eqref{eq:KNY-Ham5-sys} for $n=-1$ and $a_{0} = a_{2} = a_{3} = 1/3$, $a_{1} = 0$ becomes
\begin{equation}\label{eq:KNY-Ham5-sys2}
	\begin{aligned}
		\deriv{q_{-1}}{t} &= \frac{q_{-1}}{t}\Big((q_{-1}-1)(2p_{-1}+t) - \frac{1}{3}\Big),\\
		\deriv{p_{-1}}{t} &= \frac{1}{t}\Big(p_{-1}(p_{-1}+t)(1-2q_{-1}) +\frac{1}{3}(p_{-1}-t)\Big).
	\end{aligned}
    .
\end{equation}
Thus, the flow preserves the nodal curve $q_{-1} = 0$ and is described by the Riccati equation which, when written in terms of $v_{0}$, becomes
\begin{equation}\label{eq:ricc-v0}
	\deriv{v_{0}}{t} = v_{0}^{2} + \left(1 - \frac{2}{3t}\right) v_{0} - \frac{1}{3t}.
\end{equation}
The corresponding reduction of equation \eqref{pv} follows from the substitution $w_{0} = 1 + \frac{1}{v_{0}}$,
\begin{equation}\label{eq:ricc-w0}
	3t \deriv{w_{0}}{t} = w_{0}^{2} - 3t w_{0} -1.
\end{equation}
We revisit this reduction in the next section. Solutions of these Riccati equations 
can be expressed in terms of classical special functions.  
This will lead us to the explicit 
formula for the initial condition, 
which we now consider. 


\section{Classical solutions of dP$_{\rm I}$ from BTs for \p\ V}\label{classical} 
\setcounter{equation}{0} 

In this section we present a large family of special solutions for 
the Painlev\'e V equation, connected to one another by 
B\"acklund transformations (BTs), and show how this includes 
a one-parameter family of solutions for (\ref{dpi}), corresponding 
to an explicit orbit of the dynamics whose geometry was revealed in the 
preceding section. In the first subsection below, we will construct explicit Wronskian determinant formulae for these solutions, thus providing special solutions of $\tau$-function relations that were  obtained in \cite{ahk}. Finally, in the second subsection, we will complete the proof of uniqueness of the positive solution of  (\ref{dpi}), as well as proving  
Theorem \ref{mainthm}. 

We consider the generic case of \PV\ \eqref{pv} when $\de\not=0$, and as usual we set $\de=-\tfrac12$, i.e. we take 
\beq\deriv[2]{w}{t}=\left(\frac{1}{2w} + \frac{1}{w-1}\right)\left(\deriv{w}{t}\right)^{\!\!2} -
\frac{1}{t}\deriv{w}{t} + \frac{(w-1)^2(\al w^2+\be)}{t^2w} + \frac{\gam w}{t} 
-\frac{w(w+1)}{2(w-1)}.\label{eq:pv}\eeq
Suppose that $w$ satisfies \PV\ \eqref{eq:pv} with parameters
$ \pms{\tfrac12a^2}{-\tfrac12b^2}{c}$;
then the \bt\ $\CMcal{T}_{\ep_1,\ep_2,\ep_3}$ is defined by
\begin{align} \CMcal{T}_{\ep_1,\ep_2,\ep_3}(w) &=1-\frac{2\ep_1tw}{\ds t\ds\deriv{w}{t}- \ep_2aw^2 +(\ep_2a-\ep_3b +\ep_1t)w+ \ep_3b},\end{align}
with $\ep_j=\pm1$, $j=1,2,3$, independently, and for the parameters
\beq \CMcal{T}_{\ep_1,\ep_2,\ep_3}\big(a,b,c\big)=\left(\tfrac12\big(c+\ep_1(1-\ep_3b-\ep_2a)\big),\tfrac12\big(c-\ep_1(1-\ep_3b-\ep_2a)\big),\ep_1(\ep_3b-\ep_2a)\right),\eeq
see, for example, \cite{refGF}, 
\cite[\S39]{refGLS} or \cite[\S32.7(v)]{DLMF}. Note that the \bt\ $\CMcal{T}_{\ep_1,\ep_2,\ep_3}$ is usually given for the parameters $\al$, $\be$ and $\gam$,  in terms of $a$, $b$ and $c$. However in order to apply a sequence of \bts, it is better to define the effect of the \bt\ on the parameters $a$, $b$ and $c$ to avoid any ambiguity in taking a square root. Indeed, the discrete symmetries of 
Painlev\'e V, corresponding to the extended affine Weyl group of 
type $A_3^{(1)}$, act naturally on the root variables, as in 
\eqref{eq:root-pars-PV}, and 
the parameters $a,b,c$ here correspond to 
$a_1,a_3,a_0-a_2$ in the notation of the previous section.  

To derive a discrete equation from \bts\ of a \peq, we use a \bt\ $\CMcal{R}$, which relates a solution $w$ to another solution $w_+$, and the inverse transformation $\CMcal{R}^{-1}$, which relates $w$ to a third solution $w_-$. Then eliminating the derivative between the two \bts\ gives an algebraic equation relating $w$, $w_+$ and $w_-$, which is the discrete equation, cf.~\cite{refCMW} for more details of this procedure, see also \cite{refFGR}.
Here we 
consider the \bt\ $\CMcal{R}=\CMcal{T}_{-1,1,1}$ which has 
inverse $\CMcal{R}^{-1}=\CMcal{T}_{-1,1,-1}\circ \CMcal{T}_{1,-1,-1}\circ\CMcal{T}_{1,-1,1}$, then
\begin{subequations}\label{PVbt1}\begin{align} 
w_{+}&=\CMcal{R}(w;a,b,c)=1+\frac{2tw}{\ds t\ds\deriv{w}{t}-aw^2+(a-b-t)w+b},\label{PVbt1a}\\
w_{-}&=\CMcal{R}^{-1}(w;a,b,c)=1-\frac{2tw}{\ds t\ds\deriv{w}{t}+aw^2-(a+b-t)w+b},\label{PVbt1b}
\end{align}\end{subequations}
and 
$ (a_+,b_+,c_+)=\CMcal{R}(a,b,c)=\left(\tfrac12\big(a+b+c-1\big),-\tfrac12\big(a+b-c-1\big), a-b\right)$,
$(a_-,b_-,c_-)=\CMcal{R}^{-1}(a,b,c)=\left(\tfrac12\big(a-b+c+1\big),\tfrac12\big(a-b-c+1\big), a+b\right)$.
Eliminating $\ds \ds\deriv{w}{t}$ in \eqref{PVbt1} gives the algebraic equation relating $w$, $w_+$ and $w_-$, 
and  comparing this with \eqref{pvteps} 
we set $w=1+1/v$ and $w_\pm=1+1/v_\pm$, which gives 
\[\frac{1}{w_{+}-1}+\frac{1}{w_{-}-1}=-\frac{a(w-1)}{t}-1 \implies 
v_{+}+v_{-}=-1-\frac{a}{tv}.\]

\begin{remark}{\rm
In terms of $w$, 
the transformation \eqref{PVbt1b} is the same as $\CMcal{T}_{1,-1,1}(w)$. 
However, the parameters $(a,b,c)$ map differently since
\[\CMcal{T}_{1,-1,1}(a,b,c)=\big(\tfrac12(a-b+c-1),\tfrac12(-a+b+c-1), a+b\big). \] Consequently it is straightforward to show that
$\CMcal{T}_{1,-1,1}\circ\CMcal{T}_{-1,1,1}(w)=w$, 
but $\CMcal{T}_{-1,1,1}\circ\CMcal{T}_{1,-1,1}(w)\not=w$, 
so $\CMcal{T}_{1,-1,1}$ is not the inverse of $\CMcal{T}_{-1,1,1}$.
}\end{remark}

Under successive iterations of 
$\CMcal{R}$ and its inverse, the parameters 
evolve as follows:
\begin{subequations}\begin{align} 
(a_{n+1},b_{n+1},c_{n+1})&=\,\,\CMcal{R}(a_n,b_n,c_n) 
=\left(\tfrac12\big(a_{n}+b_{n}+c_{n}-1\big),-\tfrac12\big(a_{n}+b_{n}-c_{n}-1\big), a_{n}-b_{n}\right),\\
(a_{n-1},b_{n-1},c_{n-1})&=\CMcal{R}^{-1}(a_n,b_n,c_n)
=\left(\tfrac12\big(a_{n}-b_{n}+c_{n}+1\big),\tfrac12\big(a_{n}-b_{n}-c_{n}+1\big), a_{n}+b_{n}\right).
\end{align}\end{subequations}
From these equations it can be shown that $a_n$ satisfies the third order difference equation
$$ 
a_{n+3}-a_n+1=0, \qquad \mathrm{with} \qquad  b_n=a_{n+1}-a_{n+2},\qquad c_n=a_{n+1}+b_{n+1},
$$
so we obtain the solutions
\begin{subequations}\label{abcn}\begin{align}  a_n&=\mu\cos \big(\tfrac23\pi n\big)+\ka\sin \big(\tfrac23\pi n\big)+\la-\tfrac13n,\\
b_n&=\sqrt{3}\,\ka\cos \big(\tfrac23\pi n\big)-\sqrt{3}\,\mu\sin \big(\tfrac23\pi n\big)+\tfrac13,\\
c_n&=-2\mu\cos \big(\tfrac23\pi n\big)-2\ka\sin \big(\tfrac23\pi n\big)+\la-\tfrac13n,
\end{align}\end{subequations}
with $\mu$, $\ka$ and $\la$ arbitrary constants.
\comment{Setting $\mu=\ka=0$ and $\la=-\tfrac13$  gives
\[ a_n=-\tfrac13(n+1),\qquad b_n=\tfrac13,\qquad c_n=-\tfrac13(n+1),\]
i.e.
\[ \pms{\tfrac{1}{18}(n+1)^2}{-\tfrac{1}{18}}{-\tfrac13(n+1)}.\]} 
(Note that $(a_n)_{n\in\Z}$ denotes a sequence of values of the root variable $a$ here, and should not be confused with the indices $0,1,2,3$ used in the previous section, as in \eqref{eq:root-pars-PV}, for instance.)  
The solution $w_n$ evolves according to 
\begin{subequations}\begin{align} 
w_{n+1}&=\CMcal{R}(w_n)=1+\frac{2tw_n}{\ds t\deriv{w_n}{t}-a_nw_n^2+(a_n-b_n-t)w_n+b_n},\\
w_{n-1}&=\CMcal{R}^{-1}(w_n)=1-\frac{2tw_n}{\ds t\deriv{w_n}{t}+a_nw_n^2-(a_n+b_n-t)w_n+b_n}.
\end{align}\end{subequations}
Eliminating $\ds \deriv{w_n}{t}$ gives the discrete equation
\beq \frac{1}{w_{n+1}-1}+\frac{1}{w_{n-1}-1}=-\frac{a_n(w_n-1)}{t}-1,\eeq
then setting $w_n=1+1/v_n$ 
gives the discrete equation 
\beq v_{n+1}+v_{n-1}+1+\frac{a_n}{tv_n}=0,\label{dPIgen}\eeq
which is equivalent to equation (3.23) in \cite{tgr}.
%
We remark that $v_n(t)$ also satisfies the second-order differential equation
$$
\deriv[2]{v_n}{t}=\frac12\left(\frac{1}{v_n}+\frac{1}{v_n+1}\right)\left(\deriv{v_n}{t}\right)^{\!2}-\frac1t\deriv{v_n}{t}
-\frac{a_n^2(v_n+1)^2+b_n^2}{2v_n(v_n+1)t^2} 
-\frac{c_nv_n(v_n+1)}{t}+\frac{v_n(v_n+1)(2v_n+1)}{2}, 
$$ 
and the differential-difference equations
\begin{subequations}\label{sys:vn}\begin{align}
v_{n+1}&+\frac{1}{2v_n(v_n+1)}\deriv{v_n}{t}+\frac{(a_n+b_n)v_n+a_n}{2tv_n(v_n+1)}+\frac12=0,\label{sys:vn1}\\
v_{n-1}&-\frac{1}{2v_n(v_n+1)}\deriv{v_n}{t}+\frac{(a_n-b_n)v_n+a_n}{2tv_n(v_n+1)}+\frac12=0.\label{sys:vn2}
\end{align}\end{subequations}
as well as the discrete equation \eqref{dPIgen}.

We are interested in the special case of \eqref{abcn} when $\mu=\ka=0$ and $\la=-\tfrac13$, i.e.
\beq 
\label{eqn:paramsn} a_n=-\tfrac13(n+1),\qquad b_n=\tfrac13,\qquad c_n=-\tfrac13(n+1),\eeq
and solutions of \PV\ \eqref{eq:pv} with
\beq 
\label{eqn:paramspv}
 \pms{\tfrac12a_n^2}{-\tfrac12b_n^2}{c_n}=\big(\tfrac{1}{18}(n+1)^2,-\tfrac{1}{18},-\tfrac13(n+1)\big).
\eeq
There are special function solutions of \PV\ \eqref{eq:pv} with $ \pms{\tfrac12a^2}{-\tfrac12b^2}{c}$ if 
there is some  $m\in\Z$ such that 
either 
$\ep_1 a+\ep_2b+\ep_3c=2m+1$, 
$a=m$, 
with $\ep_j=\pm1$, $j=1,2,3$, independently. 
For the parameters \eqref{eqn:paramsn}, \PV\ \eqref{eq:pv} has special function solutions since  
$a_{3n}-b_{3n}+c_{3n}=-2n-1$, $a_{3n+1}+b_{3n+1}+c_{3n+1}= -2n-1$, and $ a_{3n+2}=-n-1$.

\begin{lem}\label{RiccatiPf}
The only Riccati equation which is compatible with \PV\ \eqref{eq:pv} with parameters
$\pms{\tfrac12a_0^2}{-\tfrac12b_0^2}{c_0}=(\tfrac{1}{18},-\tfrac{1}{18},-\tfrac13)$, 
i.e.
\beq\deriv[2]{w_0}{t}=\left(\frac{1}{2w_0} + \frac{1}{w_0-1}\right)\left(\deriv{w_0}{t}\right)^{\!\!2} -
\frac{1}{t}\deriv{w_0}{t} + \frac{(w_0-1)^2(w_0^2-1)}{18t^2w_0} - \frac{w_0}{3t} 
-\frac{w_0(w_0+1)}{2(w_0-1)}.\label{eq:pv1}\eeq
is equation \eqref{eq:ricc-w0}, 
which has solution
\beq\label{sol:ric}
w_0(t)=-\frac{C_1\big\{ I_{1/6}(\tfrac12t)-I_{-5/6}(\tfrac12t)\big\}+C_2\big\{K_{1/6}(\tfrac12t)+K_{5/6}(\tfrac12t)\big\}}{C_1\big\{I_{1/6}(\tfrac12t)+I_{-5/6}(\tfrac12t)\big\}+C_2\big\{K_{1/6}(\tfrac12t)-K_{5/6}(\tfrac12t)\big\}},
\eeq
where $I_{\nu}(\tfrac12t)$ and $K_{\nu}(\tfrac12t)$ are modified Bessel functions,
with $C_1$ and $C_2$ arbitrary constants.
\end{lem}
\begin{proof}Using the Riccati equation
\[ \deriv{w_0}{t}=p_2(t)w_0^2+p_1(t)w_0+p_0(t),\]
where $p_2(t)$, $p_1(t)$ and $p_0(t)$ are functions to be determined,
to remove the derivatives in \eqref{eq:pv1} then equating coefficient of powers of $w_0$ shows that
$p_2(t) =\frac{1}{3t}$,
$p_1(t)=-1$, 
$p_0(t) =-\frac{1}{3t}$,
so we obtain \eqref{eq:ricc-w0}, 
as required. Letting  $w_0=-3t\frac{\rd}{\rd t} \log\ph_0$ 
in \eqref{eq:ricc-w0} gives
\beq t^2\deriv[2]{\ph_0}{t}+t(t+1)\deriv{\ph_0}{t}-\tfrac19\ph_0=0,\label{Ric:lin}\eeq
which has solution
\beq \ph_0(t)=\sqrt{t}\,\Big\{C_1\big[I_{1/6}(\tfrac12t)+I_{-5/6}(\tfrac12t)\big] + C_2\big[K_{1/6}(\tfrac12t)-K_{5/6}(\tfrac12t)\big] \Big\}\exp(-\tfrac12t),\label{phi0Bessel}\eeq
and so we obtain the solution \eqref{sol:ric}, as required. 
\end{proof}
\begin{remark}{\rm Special function solutions of \PV\ \eqref{eq:pv} are usually expressed in terms of the 
Whittaker functions $M_{\kappa,\mu}(t)$ and $W_{\kappa,\mu}(t)$, or equivalently the Kummer functions $M(a,b,t)$ and $U(a,b,t)$, cf.\ \cite[\S32.10(v)]{DLMF}. 
However if $b=2a+n$, with $n$ an integer, then the Kummer functions $M(a,b,t)$ and $U(a,b,t)$ can be respectively expressed in terms of the modified Bessel functions $I_{\nu}(\tfrac12t)$ and $K_{\nu}(\tfrac12t)$, for example 
$$M(\nu+\tfrac12,2\nu+1,t)=\Gamma(1+\nu) (\tfrac14t)^{-\nu}I_{\nu}(\tfrac12t)\exp(\tfrac12t),\quad 
U(\nu+\tfrac12,2\nu+1,t)=\pi^{-1/2}t^{-\nu}K_{\nu}(\tfrac12t)\exp(\tfrac12t); 
$$ 
see \cite[\S13.6(iii)]{DLMF}.
In terms of Kummer functions, the solution of \eqref{Ric:lin} is given by
\beq \ph_0(t)=t^{-1/3} \Big\{C_1M(\tfrac23,\tfrac13,t)+C_2U(\tfrac23,\tfrac13,t) \Big\}\,\re^{-t},\label{Ric:lin1}\eeq
with $C_1$ and $C_2$ arbitrary constants, which gives the solution of \eqref{eq:ricc-w0} 
\beq w_0(t)=-\frac{3t}{\ph_0(t)}\deriv{\ph_0}{t}
=3(t+1)-\frac{2C_1M(\tfrac53,\tfrac13,t) +\tfrac83C_2U(\tfrac53,\tfrac13,t)}{C_1M(\tfrac53,\tfrac13,t) +C_2U(\tfrac53,\tfrac13,t)}.\label{Ric:lin2}
\eeq
The Kummer functions $M(\tfrac23,\tfrac13,t)$,  $M(\tfrac53,\tfrac13,t)$, $U(\tfrac23,\tfrac13,t)$  and $U(\tfrac53,\tfrac13,t)$ can be expressed in terms of modified Bessel functions as follows 
\begin{align*}
M(\tfrac23,\tfrac13,t)&=(\tfrac14t)^{5/6}\Gamma(\tfrac16)\Big\{ I_{1/6}(\tfrac12t) + I_{-5/6}(\tfrac12t) \Big\}\exp(\tfrac12t),\\ 
M(\tfrac53,\tfrac13,t)&=\tfrac12(\tfrac14t)^{5/6}\Gamma(\tfrac16)\Big\{ (3t+4)I_{1/6}(\tfrac12t) + (3t+2)I_{-5/6}(\tfrac12t) \Big\}\exp(\tfrac12t),\\ 
U(\tfrac23,\tfrac13,t)&=\frac{3t^{5/6}}{2\sqrt{\pi}}\Big\{ K_{5/6}(\tfrac12t) - K_{1/6}(\tfrac12t)   \Big\}\exp(\tfrac12t),\\
U(\tfrac53,\tfrac13,t)&=\frac{9 t^{5/6}}{16\sqrt{\pi}}\Big\{ (3t+2)K_{5/6}(\tfrac12t) - (3t+4)K_{1/6}(\tfrac12t)   \Big\}\exp(\tfrac12t).
\end{align*}
The solutions \eqref{Ric:lin1} and \eqref{Ric:lin2}
can be expressed in terms of Whittaker functions since the relationship between the Whittaker functions $M_{\ka,\nu}(t)$, $W_{\ka,\nu}(t)$ and the Kummer functions $M(a,b,t)$, $U(a,b,t)$ is given by
\[ M_{\ka,\nu}(t)=t^{\mu+1/2}M(\mu-\ka+\tfrac12,1+2\mu,t)\exp(-\tfrac12t),\quad W_{\ka,\nu}(t)=t^{\mu+1/2}U(\mu-\ka+\tfrac12,1+2\mu,t)\exp(-\tfrac12t),\]
and conversely
\[M(a,b,t)=t^{-1/2}M_{b/2-a,b/2-1/2}(t)\exp(\tfrac12t),\quad U(a,b,t)=t^{-1/2}W_{b/2-a,b/2-1/2}(t)\exp(\tfrac12t),\]
see \cite[equations 13.14.2--13.14.5]{DLMF}
}\end{remark}

Hence, using \eqref{phi0Bessel}, we obtain 
\beq
w_0(t)=-\frac{3t}{\ph_0(t)}\deriv{\ph_0}{t}=-\frac{C_1\big\{ I_{1/6}(\tfrac12t)-I_{-5/6}(\tfrac12t)\big\}+C_2\big\{K_{1/6}(\tfrac12t)+K_{5/6}(\tfrac12t)\big\}}{C_1\big\{I_{1/6}(\tfrac12t)+I_{-5/6}(\tfrac12t)\big\}+C_2\big\{K_{1/6}(\tfrac12t)-K_{5/6}(\tfrac12t)\big\}},
\eeq
which satisfies \PV\ \eqref{eq:pv} with parameters 
$\pms{\tfrac1{18}}{-\tfrac1{18}}{-\tfrac13}$, 
and so $w_1=\CMcal{R}(w_0;-\tfrac13,\tfrac13,-\tfrac13)$  is given by 
\begin{align*}
w_{1}(t)
&=\frac{2\big\{C_1I_{1/6}(\tfrac12t)+C_2K_{1/6}(\tfrac12t)\big\}}{(3t+2)\big\{C_1I_{1/6}(\tfrac12t)+C_2K_{1/6}(\tfrac12t)\big\}+3t\big\{C_1I_{-5/6}(\tfrac12t)-C_2K_{5/6}(\tfrac12t)\big\}},
\end{align*}
which satisfies \PV\ \eqref{eq:pv} with parameters
$\pms{\tfrac2{9}}{-\tfrac1{18}}{-\tfrac23}$.
Therefore, since $v_n=1/(w_n-1)$, then
\begin{subequations}\label{vn}\begin{align} 
v_0(t)=\frac{1}{w_0(t)-1} &=-\frac12-\frac{C_1I_{-5/6}(\tfrac12t)-C_2K_{5/6}(\tfrac12t)}{2\big\{C_1I_{1/6}(\tfrac12t)+C_2K_{1/6}(\tfrac12t)\big\}},\\[2pt]
v_1(t)=\frac{1}{w_1(t)-1}
&=-1-\frac{2}{3t} -\frac{ 2\big\{C_1I_{-5/6}(\tfrac12t)-C_2K_{5/6}(\tfrac12t)\big\}}{3t\big\{C_1I_{1/6}(\tfrac12t)+C_2K_{1/6}(\tfrac12t)\big\}+C_1I_{-5/6}(\tfrac12t)-C_2K_{5/6}(\tfrac12t)}.
\end{align}\end{subequations}

Furthermore, if we set $a_n=-\tfrac13(n+1)$ in \eqref{dPIgen} then we obtain
\beq v_n(v_{n+1}+v_{n-1}+1)=\frac{n+1}{3t},\label{dPI}\eeq
which is \eqref{dpi} with $\eps=1/(3t)$, 
in agreement with \eqref{pvteps}.
Hence if we put $n=0$ in \eqref{dPI} and use \eqref{vn}, then we see that
 $$v_{-1}=-v_1-1+\frac{1}{3tv_0}=0.$$
%
Moreover, if we let $Z_{1/6}(t)=C_1I_{1/6}(t)+C_2K_{1/6}(t)$ and $Z_{-5/6}(t)=C_1I_{-5/6}(t)-C_2K_{5/6}(t)$, 
then 
the first three non-zero 
iterates of the dP$_{I}$ equation (\ref{dPI}) 
are given by 
\begin{align*} v_0(t)&=-\frac12-\frac{Z_{-5/6}(\tfrac12t)}{2Z_{1/6}(\tfrac12t)}, \quad 
v_1(t)
=-1-\frac{2}{3t}-\frac{2Z_{-5/6}(\tfrac12t)}{3t\big\{Z_{1/6}(\tfrac12t)+Z_{-5/6}(\tfrac12t)\big\}},\\
v_2(t)&=-\frac{3(t+2)}{2(3t+2)}+\frac{Z_{-5/6}(\tfrac12t)}{2Z_{1/6}(\tfrac12t)}
-\frac{4Z_{-5/6}(\tfrac12t)}{(3t+2)\big\{(3t+2)Z_{1/6}(\tfrac12t)+3tZ_{-5/6}(\tfrac12t)\big\}}.
\end{align*}

\begin{remark}{\rm
It was shown in Lemma \ref{RiccatiPf} that $w_0(t)$ satisfies the Riccati equation \eqref{eq:ricc-w0}. It follows that $w_1(t)$ 
also satisfies a Riccati equation, namely 
$
t\deriv{w_1}{t}=\tfrac23 w_1^2-(t+1)w_1+\tfrac13 
$, 
and hence 
$v_0(t)$ and $v_1(t)$  satisfy Riccati equations, namely 
\eqref{eq:ricc-v0} and 
$t\deriv{v_1}{t}=tv_1^2+(t-\tfrac13)v_1-\tfrac23$, 
respectively, which are equivalent to \eqref{ric} and \eqref{ric2}, after 
making the change of independent variable $\eps=1/(3t)$. 
The solutions $w_n(t)$ and $v_n(t)$ for $n\geq2$ do not satisfy Riccati equations with simple coefficients. However, it can be shown that $v_n$ for $n\geq 2$ does satisfy a Riccati equation with coefficients 
given by combinations of $v_{n-2}$ and lower $v_j$ \cite{hoppe}; for instance, $v_2$ satisfies a Riccati equation that includes $v_0$ among its coefficients (see below). 
}\end{remark}
%
\subsection{Determinantal representation of the solutions}
In this subsection we show that the solutions of dP$_{\mathrm{I}}$ \eqref{dPI} may be written in terms of determinants involving modified Bessel functions. These determinants can be regarded as 
particular examples of Painlev\'e V $\tau$-functions. Moreover, in \cite{ahk} it is shown that if a $\tau$-function $\upupsilon_n$ for \eqref{dpi} is introduced 
via 
\beq\label{taups} 
v_n =\frac{\upupsilon_n\upupsilon_{n-4}}{\upupsilon_{n-1}\upupsilon_{n-3}},
\eeq 
then it satisfies a trilinear (degree 3 homogeneous) equation of order 6. At the end of this subsection, we show that the determinants of modified Bessel functions 
provide special function solutions of this trilinear equation. 
We begin by defining some convenient notation for linear combinations of 
modified Bessel functions, and associated Wronskian determinants. 
\begin{definition}
\label{def:besseldet}
Let  $\BesselZ{\nu}{t}$ be defined by
\begin{equation}
\BesselZ{\nu}{t} 
=
\begin{cases}
d_1 \nBesselI{j}{t} + d_2 \sbr{-1}^j\nBesselK{j}{t}, \quad & \nu=j\in\Z,
\\
d_1 \nBesselI{\nu}{t} + d_2 \nBesselI{-\nu}{t}, \quad & \text{otherwise},
\end{cases}
\end{equation}
where $\nBesselI{\nu}{t}$ and $\nBesselK{\nu}{t}$ are modified Bessel functions and $d_1$ and $d_2$ are arbitrary constants.
For $m\in\Z$ and $n\in\N$, let $\mathcal{B}_{m,n,\nu}\!\sbr{t}$ be 
 the Wronskian determinant
\begin{equation}
\label{eqn:besseldetdef} 
\mathcal{B}_{m,n,\nu}\!\sbr{t} = \mathcal{W}\!\sbr{\cbr{f_{m-\ell,\nu+\ell}\!\sbr{t}}_{\ell=0}^{n-1}},
\end{equation}
and $\mathcal{B}_{m,0,\nu}\!\sbr{t}=1$, 
where
\begin{equation}
\label{eqn:fbesseldef}
f_{m,\nu}\!\sbr{t}
=
{\everymath={\displaystyle}
\begin{cases}
t^{-\nu}\Sum{j=0}{m}\binom{m}{j}\dfrac{\sbr{j+\nu}}{\poch{j+2\nu}{m+1}}\sbr{-1}^j\BesselZ{\nu+j}{\tfrac{1}{2}t},\quad & m\in\N, \\
t^{\vbr{m}-\nu}\Sum{j=0}{\vbr{m}}\binom{\vbr{m}}{j}\dfrac{\sbr{j-\nu}}{\poch{j-2\nu}{\vbr{m}+1}}\BesselZ{\nu-j}{\tfrac{1}{2}t},\quad & -m\in\N,  \\
\end{cases}
}
\end{equation}
whenever the denominators in \eqref{eqn:fbesseldef} are non-zero. 
\end{definition}
\begin{lem}
\label{lem:besselsolsPV}
Let $\mathcal{B}_{m,n,\nu}\!\sbr{t}$ be the determinant of modified Bessel functions given in Definition \ref{def:besseldet}. We also define the constants $C_{m,n,\nu}^{\sqbr{1}}$ and $C_{m,n,\nu}^{\sqbr{2}}$ as
\begin{equation}
C_{m,n,\nu}^{\sqbr{1}}
=
\begin{cases}
\sbr{\tfrac{1}{2}-m-\nu}, \quad & m\geq n+2, \\
1, \quad & m=n+1, \\
-\sbr{\nu+n+\tfrac{1}{2}}^{-1}, \quad & \mathrm{otherwise},
\end{cases}
\end{equation}
and
\begin{equation}
C_{m,n,\nu}^{\sqbr{2}}
=
\begin{cases}
-\sbr{\nu+m+\tfrac{1}{2}}, \quad & m\geq n, \\
1, \quad & \mathrm{otherwise},
\end{cases}
\end{equation} 
respectively. Then
\begin{equation}
w_{m,n,\nu}^{\sqbr{1}}\!\sbr{t} = C_{m,n,\nu}^{\sqbr{1}}\frac{\mathcal{B}_{m,n+1,\nu}\!\sbr{t}\mathcal{B}_{m-2,n,\nu+1}\!\sbr{t}}{\mathcal{B}_{m,n,\nu}\!\sbr{t}\mathcal{B}_{m-2,n+1,\nu+1}\!\sbr{t}},
\end{equation}
is a solution of $\PV$ for the parameters
\begin{equation}
\label{eqn:paramsbessel1}
\sbr{\al,\be,\gam} = \sbr{\tfrac{1}{8}\sbr{2\nu+2n+1}^2, -\tfrac{1}{8}\sbr{2\nu+2m+2n-1}^2, 2\nu+m-1},
\end{equation}
and
\begin{equation}
w_{m,n,\nu}^{\sqbr{2}}\!\sbr{t} = C_{m,n,\nu}^{\sqbr{2}}\frac{\mathcal{B}_{m,n,\nu+1}\!\sbr{t}\mathcal{B}_{m,n,\nu}\!\sbr{t}}{\mathcal{B}_{m,n-1,\nu+1}\!\sbr{t}\mathcal{B}_{m,n+1,\nu}\!\sbr{t}},
\end{equation}
is a solution of $\PV$ for the parameters
\begin{equation}
\label{eqn:paramsbessel2}
\sbr{\al,\be,\gam} = \sbr{\tfrac{1}{2}n^2,-\tfrac{1}{2}\sbr{2\nu+m+n}^2,m}.
\end{equation}
\end{lem}
\begin{proof}
   The special function solutions of \PV\ written in terms of the Kummer function $\nKummerM{a}{b}{t}$ were derived by Masuda \cite{masuda}; see also Forrester and Witte \cite{fw}. The solutions in Lemma \ref{lem:besselsolsPV} may be inferred from the work of Masuda by using \cite[\S 13.6(iii)]{DLMF}
\begin{subequations}
\label{eqn:KummerMtoBesselI}
\begin{align}
M\!\left(\nu+\tfrac{1}{2},2\nu+1+n,t\right)&=\Gamma\!\left(\nu\right)\EXP{\tfrac{1}{2}t}\left(\tfrac{1}{4}t\right)^{-\nu}\sum_{k=0}^{n}\sbr{-1}^k\binom{n}{k}\frac{{\left(2\nu\right)_{k}}(\nu+k)}{{\left(2\nu+1+n\right)_{k}}}I_{\nu+k}\!\left(\tfrac{1}{2}t\right),\label{eqn:MtoIpos} \\
M\!\left(\nu+\tfrac{1}{2},2\nu+1-n,t\right)&=\Gamma\!\left(\nu-n\right)\EXP{\tfrac{1}{2}t}\left(\tfrac{1}{4}t\right)^{n-\nu}\sum_{k=0}^{n}\binom{n}{k}\frac{{\left(2\nu-2n\right)_{k}}(\nu-n+k)}{{\left(2\nu+1-n\right)_{k}}}I_{\nu+k-n}\!\left(\tfrac{1}{2}t\right). \label{eqn:MtoIneg}
\end{align}
\end{subequations}
If the modified Bessel function $\nBesselK{\nu}{t}$ with $\nu\notin\Z$ is desired in the solution, we use \cite[eq. (10.27.4)]{DLMF}
\begin{equation}\label{Knu}
    \nBesselK{\nu}{t} = 
    \dfrac{\pi\big(\nBesselI{-\nu}{t}-\nBesselI{\nu}{t}\big)}{2\SIN{\nu \pi}}.
\end{equation}
For $j\in\Z$, the modified Bessel function $\nBesselK{j}{t}$ is given by \cite[eq. (10.27.5)]{DLMF}
\begin{equation}
    \nBesselK{j}{t} = \tfrac{1}{2}\sbr{-1}^{j-1}\sbr{\left. \pdiff{}{\nu}\nBesselI{\nu}{t}\right|_{\nu=j}+\left. \pdiff{}{\nu}\nBesselI{\nu}{t}\right|_{\nu=-j}}.
\end{equation}
\end{proof}
We use the following properties of $f_{m,\nu}\!\sbr{t}$ defined in \eqref{eqn:fbesseldef} to construct identities for $\mathcal{B}_{m,n,\nu}\!\sbr{t}$.
%
\begin{lem}
We have
\begin{subequations}
\begin{align}
&f_{m,\nu}\!\sbr{t}-f_{m-1,\nu+1}\!\sbr{t} = \sbr{m+\nu+\tfrac{1}{2}}f_{m+1,\nu}\!\sbr{t}, && m\geq 1 \label{eqn:besselfrecmgeq1},\\
&\sbr{\nu+\tfrac{1}{2}}\sbr{f_{1,\nu}\!\sbr{t}-f_{-1,\nu+1}\!\sbr{t}} = f_{0,\nu}\!\sbr{t}, && m=0 \label{eqn:besselfrecmeq0},\\
&f_{m,\nu}\!\sbr{t}+f_{m+1,\nu}\!\sbr{t} = -\sbr{\nu+\tfrac{1}{2}}f_{m-1,\nu+1}\!\sbr{t}, && m\leq -1\label{eqn:besselfrecmleqm1}.
\end{align}
\end{subequations}
The derivatives of $f_{m,\nu}\!\sbr{t}$ are given by
\begin{equation}
\label{eqn:besselfderivalt}
\ddiff{}{t}f_{m,\nu}\!\sbr{t}
=
\begin{cases}
\tfrac{1}{2} f_{m,\nu}-\sbr{m+\nu+\tfrac12}f_{m+1,\nu}, \quad & m \geq 0,\\[2pt]
\tfrac{1}{2}f_{m,\nu}+f_{m+1,\nu}, \quad & m \leq -1.
\end{cases}
\end{equation}
Furthermore, if $\nu\notin\Z$, the following symmetry holds:
\begin{equation}
\label{eqn:besselfswapconstants}
f_{m,\nu}\!\sbr{t;d_1,d_2} = 
\begin{cases}
t^{-2\nu-m}f_{m,-m-\nu}\!\sbr{t;d_2,d_1}, &\quad m\geq 0,
\\[2pt]
\sbr{-1}^m \,t^{-2\nu-m}f_{m,-m-\nu}\!\sbr{t;d_2,d_1}, &\quad m< 0.
\end{cases}
\end{equation}
\end{lem}
%
\begin{proof}
The properties of $f_{m,\nu}\!\sbr{t}$ are proved using the properties of the modified Bessel functions given in \cite[\S 10.29]{DLMF}. For example, \eqref{eqn:besselfrecmgeq1} is given by
\begin{multline}
\label{eqn:needbesselidfor2ndsum}
   f_{m,\nu}\!\sbr{t}-f_{m-1,\nu+1}\!\sbr{t} 
   = 
   t^{-\nu}\Sum{j=0}{m}\binom{m}{j}\dfrac{\sbr{j+\nu}}{\poch{j+2\nu}{m+1}}\sbr{-1}^j\BesselZ{\nu+j}{\tfrac{1}{2}t} 
   \\
   - t^{-\nu-1}\Sum{j=0}{m-1}\binom{m-1}{j}\dfrac{\sbr{j+\nu+1}}{\poch{j+2\nu+2}{m}}\sbr{-1}^j\BesselZ{\nu+1+j}{\tfrac{1}{2}t}.
\end{multline}
Using \cite[eqn. (10.29.1)]{DLMF}
\begin{equation}
\BesselZ{\nu-1}{\tfrac{1}{2}t}-\BesselZ{\nu+1}{\tfrac{1}{2}t} = \frac{4\nu}{t}\BesselZ{\nu}{\tfrac{1}{2}t},
\end{equation}
to rewrite the second sum in \eqref{eqn:needbesselidfor2ndsum}, we obtain
\begin{multline}
\label{eqn:needtocombineliketerms}
   f_{m,\nu}\!\sbr{t}-f_{m-1,\nu+1}\!\sbr{t}
   = 
   t^{-\nu}\Sum{j=0}{m}\binom{m}{j}\dfrac{\sbr{j+\nu}}{\poch{j+2\nu}{m+1}}\sbr{-1}^j\BesselZ{\nu+j}{\tfrac{1}{2}t} 
   \\
   + t^{-\nu}\Sum{j=0}{m-1}\binom{m-1}{j}\dfrac{1}{4\poch{j+2\nu+2}{m}}\sbr{-1}^j\cbr{\BesselZ{\nu+2+j}{\tfrac{1}{2}t}-\BesselZ{\nu+j}{\tfrac{1}{2}t}}.
\end{multline}
Combining like terms in \eqref{eqn:needtocombineliketerms} gives \eqref{eqn:besselfrecmgeq1}. The remaining identities are proved similarly.
\end{proof}
%
%
\begin{lem}
\label{lem:besseldetswapc1c2}
When $\nu\notin\Z$ the Bessel determinant $\mathcal{B}_{m,n,\nu}\!\sbr{t}$ has the following symmetry:
\begin{equation}
\label{eqn:besseldetswapc1c2}
\mathcal{B}_{m,n,\nu}\!\sbr{t;d_1,d_2} = 
\begin{cases}
\dfrac{r_{m,n,\nu}}{r_{m,n,1-m-n-\nu}}t^{n\sbr{1-m-n-2\nu}}\mathcal{B}_{m,n,1-m-n-\nu}\!\sbr{t;d_2,d_1}, \quad & m\geq n-1,
\\[1em]
\sbr{-1}^{mn}\dfrac{r_{m,n,\nu}}{r_{m,n,1-m-n-\nu}}t^{n\sbr{1-m-n-2\nu}}\mathcal{B}_{m,n,1-m-n-\nu}\!\sbr{t;d_2,d_1}, \quad & 1\leq m\leq n-2,
\\[1em]
\sbr{-1}^{mn}t^{n\sbr{1-m-n-2\nu}}\mathcal{B}_{m,n,1-m-n-\nu}\!\sbr{t;d_2,d_1}, \quad & m\leq 0.
\end{cases}
\end{equation}
where $r_{m,n,\nu}$ is the constant
\begin{equation}
r_{m,n,\nu} = 
\begin{cases}
\Prod{\ell=0}{n-1}\poch{\nu+m+\tfrac{1}{2}}{\ell}, \quad & m\geq n-1,
\\[1em]
\Prod{\ell=0}{m}\poch{\nu+n-\tfrac{1}{2}}{\ell}, \quad & \mathrm{otherwise}.
\end{cases}
\end{equation}
\end{lem}
%
\begin{proof}
We prove \eqref{eqn:besseldetswapc1c2} when $m\geq n-1$. Subtracting column $j+1$ from column $j$ in \eqref{eqn:besseldetdef}  for $j=1,2,\dots,k, k = n-1,n-2,\dots,1$, and using the recurrence relation \eqref{eqn:besselfrecmgeq1}, we obtain
\begin{equation}
\label{eqn:besseldetwronskian2j}
\mathcal{B}_{m,n,\nu}\!\sbr{t;d_1,d_2} = r_{m,n,\nu}\mathcal{W}\!\sbr{\cbr{f_{m+n-1-2j,\nu+j}\!\sbr{t;d_1,d_2}}_{j=0}^{n-1}}.
\end{equation}
By applying the symmetry \eqref{eqn:besselfswapconstants} to \eqref{eqn:besseldetwronskian2j}, we have
\begin{equation}
\label{eqn:besselswapc1c2needwronskiant}
\mathcal{B}_{m,n,\nu}\!\sbr{t;d_1,d_2}= r_{m,n,\nu}\mathcal{W}\!\sbr{\cbr{t^{1-m-n-2\nu}f_{m+n-1-2j,1+j-m-n-\nu}\!\sbr{t;d_2,d_1}}_{j=0}^{n-1}}.
\end{equation}
We then use the Wronskian identity \cite[Thm. 4.25]{determinantsveindale}
\begin{equation}
\label{eqn:wronskianidentityproduct}
\mathcal{W}\sbr{g\!\sbr{t} f_1\!\sbr{t}, g\!\sbr{t} f_2\!\sbr{t},\dots,g\!\sbr{t} f_n\!\sbr{t}} = g\!\sbr{t}^n \mathcal{W}\sbr{f_1\!\sbr{t}, f_2\!\sbr{t},\dots, f_n\!\sbr{t}},
\end{equation} 
in order 
to remove the powers of $t$ from the Wronskian in \eqref{eqn:besselswapc1c2needwronskiant}, thus:
\begin{equation}
\begin{aligned}
\mathcal{B}_{m,n,\nu}\!\sbr{t;d_1,d_2}&= r_{m,n,\nu}t^{n\sbr{1-m-n-2\nu}}\mathcal{W}\!\sbr{\cbr{f_{m+n-1-2j,1+j-m-n-\nu}\!\sbr{t;d_2,d_1}}_{j=0}^{n-1}}
\\
&= \frac{r_{m,n,\nu}}{r_{m,n,1-m-n-\nu}}t^{n\sbr{1-m-n-2\nu}}\mathcal{B}_{m,n,1-m-n-\nu}\!\sbr{t;d_2,d_1}.
\end{aligned}
\end{equation}
The proofs of the remaining cases are similar.
\end{proof}
%
\begin{lem}
\label{lem:bilin2}
Let $\mathcal{B}_{m,n,\nu}\!\sbr{t}$ be the determinant of modified Bessel functions defined in Definition \ref{def:besseldet}. When $m<0$, $\mathcal{B}_{m,n,\nu}\!\sbr{t}$ satisfies
\begin{equation}
\label{eqn:bilinearmleq0first}
\BesselDet{m}{n-1}{\nu+1}{t}\BesselDet{m}{n+1}{\nu}{t} 
+ \BesselDet{m-1}{n}{\nu+1}{t}\BesselDet{m+1}{n}{\nu}{t}
=
\BesselDet{m}{n}{\nu+1}{t}\BesselDet{m}{n}{\nu}{t} .
\end{equation}
\end{lem}
\begin{proof}
We prove Lemma \ref{lem:bilin2} using the Jacobi identity \cite{refDod}, sometimes known as the \textit{Lewis Carroll formula}, for determinants.
Let $\mathcal{D}$ be an arbitrary determinant, and $\mathcal{D}\bn{i}{j}$ be the determinant with the $i$th row and $j$th column removed from $\mathcal{D}$. Then we have the Jacobi identity:
\begin{equation}
\label{eqn:jacobideterminant}
\mathcal{D}
\,
\mathcal{D}\bn{i,j}{k,\ell}
= 
\mathcal{D}\bn{i}{k}\mathcal{D}\bn{j}{\ell}
-
\mathcal{D}\bn{i}{\ell}\mathcal{D}\bn{j}{k}
.
\end{equation}
Using the derivative of $f_{m,\nu}\!\sbr{t}$ given in \eqref{eqn:besselfderivalt} we rewrite the Wronskian determinant $\BesselDet{m}{n}{\nu}{t}$ when $m<0$ as
\begin{equation}
\label{eqn:besselforjacobi1}
\mathcal{B}_{m,n,\nu}\!\sbr{t} = \mathrm{det}\vbr{\Sum{j=0}{k}\binom{k}{j} 2^{j-k}f_{m+j-\ell,\nu+\ell}}_{k,\ell=0}^{n-1}. 
\end{equation}
Since we can add a multiple of any row to any other row without changing the determinant in \eqref{eqn:besselforjacobi1}, we keep the last term in each sum:
\begin{equation}
\label{eqn:besselforjacobi}
\mathcal{B}_{m,n,\nu}\!\sbr{t} = \mathrm{det}\bbr{f_{m+k-\ell,\nu+\ell}}_{k,\ell=0}^{n-1}. 
\end{equation}
We apply the Jacobi identity \eqref{eqn:jacobideterminant} to the determinant in \eqref{eqn:besselforjacobi}, choosing $i=1$, $j=n$ for the rows and $k=1$, $\ell=n$ for the columns. The relevant minor determinants are
\begin{subequations}
\label{eqn:minordetsforbessel}
\begin{align}
\mathcal{B}_{m,n,\nu}\bn{1,n}{1,n} &= \mathrm{det}\bbr{f_{m+k-\ell,\nu+1}}_{i,j=0}^{n-3} = \mathcal{B}_{m,n-2,\nu+1},
\\
\mathcal{B}_{m,n,\nu}\bn{1}{1} &= \mathrm{det}\bbr{f_{m+k-\ell,\nu+1}}_{i,j=0}^{n-2} = \mathcal{B}_{m,n-1,\nu+1},
\\
\mathcal{B}_{m,n,\nu}\bn{n}{n} &= \mathrm{det}\bbr{f_{m+k-\ell,\nu+\ell}}_{i,j=0}^{n-2} = \mathcal{B}_{m,n-1,\nu},
\\
\mathcal{B}_{m,n,\nu}\bn{1}{n} &= \mathrm{det}\bbr{f_{m+1+k-\ell,\nu+\ell}}_{i,j=0}^{n-2} = \mathcal{B}_{m+1,n-1,\nu},
\\
\mathcal{B}_{m,n,\nu}\bn{n}{1} &= \mathrm{det}\bbr{f_{m-1+k-\ell,\nu+1+\ell}}_{i,j=0}^{n-1} = \mathcal{B}_{m-1,n-1,\nu+1}.
\end{align}
\end{subequations}
Substituting \eqref{eqn:minordetsforbessel} into \eqref{eqn:jacobideterminant} gives \eqref{eqn:bilinearmleq0first}.
\end{proof}
\begin{lem}
\label{lem:bilin1}
Let $\mathcal{B}_{m,n,\nu}\!\sbr{t}$ be the determinant of modified Bessel functions defined in Definition \ref{def:besseldet}. When $m\leq 0$, we have
\begin{equation}
\label{eqn:bilin1mleqm1}
\BesselDet{m}{n+1}{\nu}{t}\BesselDet{m-2}{n}{\nu+1}{t} +\sbr{\nu+n+\tfrac{1}{2}} \BesselDet{m-2}{n+1}{\nu+1}{t}\BesselDet{m}{n}{\nu}{t} 
= -\BesselDet{m-1}{n}{\nu+1}{t}\BesselDet{m-1}{n+1}{\nu}{t}.
\end{equation}
\end{lem}
%
\begin{proof}
In the Wronskian determinant \eqref{eqn:besseldetdef}, subtracting $\sbr{\nu+j-1/2}$ times column $j+1$ from column $j$ for $j=1,2,\dots,k$, where $k$ decreases from $n-1$ to $1$, and using the recurrence relation \eqref{eqn:besselfrecmleqm1}, we obtain
\begin{equation}
\label{eqn:bilineartep1}
\mathcal{B}_{m,n,\nu}\!\sbr{t} = \sbr{-1}^{n\sbr{n-1}/2}\mathcal{W}\!\sbr{\cbr{f_{m+n-1-2j,\nu+j}}_{j=0}^{n-1}}.
\end{equation}
By adding $\sbr{\nu+j-1/2}$ times column $j+1$ to column $j$ for $j=1,\dots,n-1$ in \eqref{eqn:bilineartep1},  we have 
\begin{equation}
\label{eqn:bilinearstep2}
\mathcal{B}_{m+1,n,\nu}\!\sbr{t} = \sbr{-1}^{1+n\sbr{n+1}/2}\mathcal{W}\!\sbr{\cbr{f_{m+n-1-2j,\nu+j}}_{j=0}^{n-2},f_{m+2-n,\nu+n-1}}.
\end{equation}
Adding $1/\sbr{\nu+n-3/2}$ times column $n-1$ to column $n$ in \eqref{eqn:bilinearstep2} gives
\begin{equation}
\mathcal{B}_{m+1,n,\nu}\!\sbr{t} = \frac{1}{\nu+n-\tfrac{3}{2}}\sbr{-1}^{n\sbr{n+1}/2}\mathcal{W}\!\sbr{\cbr{f_{m+n-1-2j,\nu+j}}_{j=0}^{n-2},f_{m+4-n,\nu+n-2}}.
\end{equation}
Vein and Dale prove three variants of the Jacobi identity 
\eqref{eqn:jacobideterminant} in \cite[Thm. 3.6]{determinantsveindale}. To prove the bilinear relation \eqref{eqn:bilin1mleqm1}, we use
\begin{equation}
\label{eqn:JacobiB}
      \mathcal{A}_n
      \bn{i}{p}
      \mathcal{A}_{n+1}
      \bn{n+1}{q}
      -
      \mathcal{A}_n
      \bn{i}{q}
      \mathcal{A}_{n+1}
      \bn{n+1}{p}
      =
      \mathcal{A}_n
      \mathcal{A}_{n+1}
      \bn{i,n+1}{p,q}
      ,
\end{equation}
which is identity (B) in \cite[Thm. 3.6]{determinantsveindale}. Let
\begin{equation}
A_{n+1}=\mathcal{W}\!\sbr{\cbr{f_{m+n-2j,\nu+j}}_{j=0}^{n-1},f_{m+3-n,\nu+n-1}}, \qquad A_{n}=\mathcal{W}\!\sbr{\cbr{f_{m+n-2j,\nu+j}}_{j=0}^{n-1}}.
\end{equation}
Setting $i=n$, $p=1$, and $q=n$, we have
\begin{subequations}
\label{eqn:AforJacobiBessel}
\begin{align}
   \mathcal{A}_n\bn{n}{1} &= \mathcal{W}\!\sbr{\cbr{f_{m+n-2-2j,\nu+1+j}}_{j=0}^{n-2}} = \sbr{-1}^{\sbr{n-1}\sbr{n-2}/2}\mathcal{B}_{m,n-1,\nu+1}\!\sbr{t},
   \\
   \mathcal{A}_{n+1}\bn{n+1}{n} &= \mathcal{W}\!\sbr{\cbr{f_{m+n-2j,\nu+j}}_{j=0}^{n-2},f_{m+3-n,\nu+n-1}} = \sbr{-1}^{1+n\sbr{n+1}/2}\mathcal{B}_{m,n,\nu}\!\sbr{t},
   \\
   \mathcal{A}_n\bn{n}{n} &= \mathcal{W}\!\sbr{\cbr{f_{m+n-2j,\nu+j}}_{j=0}^{n-2}} =  \sbr{-1}^{\sbr{n-1}\sbr{n-2}/2}\mathcal{B}_{m,n,\nu}\!\sbr{t},
   \\
   \mathcal{A}_{n+1}\bn{n+1}{1} &= \mathcal{W}\!\sbr{\cbr{f_{m+n-2-2j,\nu+1+j}}_{j=0}^{n-2},f_{m+3-n,\nu+n-1}} =  \sbr{\nu+n-\tfrac{1}{2}}\sbr{-1}^{n\sbr{n+1}/2}\mathcal{B}_{m,n,\nu}\!\sbr{t},
   \\
   \mathcal{A}_n &=  \sbr{-1}^{n\sbr{n-1}/2}\mathcal{B}_{m,n,\nu}\!\sbr{t},
   \\
   \mathcal{A}_{n+1}\bn{n,n+1}{1,n} &= \mathcal{W}\!\sbr{\cbr{f_{m+n-2-2j,\nu+1+j}}_{j=0}^{n-3},f_{m+3-n,\nu+n-1}} =  \sbr{-1}^{1+n\sbr{n-1}/2}\mathcal{B}_{m,n,\nu}\!\sbr{t}.
\end{align}
\end{subequations}
Substituting \eqref{eqn:AforJacobiBessel} into the Jacobi identity \eqref{eqn:JacobiB} gives
\eqref{eqn:bilin1mleqm1}.
\end{proof}
\begin{thm}\label{wrdpi}
The solutions of  dP$_{\rm I}$ \eqref{dPI} are a special case of the modified Bessel function solutions of $\PV$ given in Lemma \ref{lem:besselsolsPV}, namely $v_n\!\sbr{t} = 1/\sbr{w_n\!\sbr{t}-1}$ where $w_n\!\sbr{t}$ satisfies
\begin{equation}
w_n\!\sbr{t} = 
\begin{cases}
w_{1-k,k,-{1}/{6}}^{\sqbr{1}}\!\sbr{t;d_1,d_2}, \quad & n = 3k,\enskip k \in\N, \\[4pt]
w_{-k,k,{1}/{6}}^{\sqbr{1}}\!\sbr{t;d_2,d_1}, \quad & n = 3k+1,\enskip k \in\N, \\[4pt]
w_{-k-1,k+1,-{1}/{6}}^{\sqbr{2}}\!\sbr{t;d_1,d_2}, \quad & n = 3k+2,\enskip k \in\N.
\end{cases}
\end{equation}
\end{thm}
%
\begin{proof}
When $n=3k$, the parameters of \PV\ \eqref{eqn:paramspv} become
\begin{equation}
    \sbr{\alpha,\beta,\gamma} = \sbr{\tfrac{1}{18}\sbr{3k+1}^2,-\tfrac{1}{18},-k-\tfrac{1}{3}},
\end{equation}
which is \eqref{eqn:paramsbessel1} with $m=1-k$, $n=k$ and $\nu=-\tfrac16$. The cases when $n=3k+1$ and $n=3k+2$ are obtained similarly.
\end{proof}
%
Using the recurrence relations in Lemmas \ref{lem:bilin2} and \ref{lem:bilin1}, the solutions $v_n\!\sbr{t}$ of dP$_{\rm I}$ \eqref{dPI} may be written as
\begin{equation}
\label{eqn:qnrecrelapplied}
v_n\!\sbr{t} = 
{\everymath={\displaystyle}
\begin{cases}
0, \quad & n = -1, \\[1em]
\sbr{k+\tfrac{1}{3}}\frac{\BesselDet{1-k}{k}{-1/6}{t;d_1,d_2}\BesselDet{-1-k}{k+1}{5/6}{t;d_1,d_2}}{\BesselDet{-k}{k}{5/6}{t;d_1,d_2}\BesselDet{-k}{k+1}{-1/6}{t;d_1,d_2}}, \quad & n = 3k,\enskip k \in\N, \\[1em]
\sbr{k+\tfrac{2}{3}}\frac{\BesselDet{-k}{k}{1/6}{t;d_2,d_1}\BesselDet{-2-k}{k+1}{7/6}{t;d_2,d_1}}{\BesselDet{-k-1}{k}{7/6}{t;d_2,d_1}\BesselDet{-k-1}{k+1}{1/6}{t;d_2,d_1}}, \quad & n = 3k+1,\enskip k \in\N, \\[1em]
\frac{\BesselDet{-1-k}{k}{5/6}{t;d_1,d_2}\BesselDet{-1-k}{k+2}{-1/6}{t;d_1,d_2}}{\BesselDet{-k-2}{k+1}{5/6}{t;d_1,d_2}\BesselDet{-k}{k+1}{-1/6}{t;d_1,d_2}}, \quad & n = 3k+2,\enskip k \in\N.
\end{cases}
}
\end{equation}
Furthermore, using the symmetry \eqref{eqn:besseldetswapc1c2}, we may rewrite $v_{3k+1}\!\sbr{t}$ as
\begin{equation}
v_{3k+1}\!\sbr{t} = -\frac{3k+2}{3t}\frac{\BesselDet{-k}{k}{5/6}{t;d_1,d_2}\BesselDet{-k-2}{k+1}{5/6}{t;d_1,d_2}}{\BesselDet{-k-1}{k}{5/6}{t;d_1,d_2}\BesselDet{-k-1}{k+1}{5/6}{t;d_1,d_2}}.
\end{equation}
Substituting \eqref{eqn:qnrecrelapplied} into \eqref{dPI} gives the trilinear equations
\begin{align}
&\mathcal{B}_{-k-1,k+1,5/6}\cbr{\mathcal{B}_{-k-1,k,5/6}\mathcal{B}_{-k,k+1,-1/6}+\sbr{k+\tfrac{1}{3}}\mathcal{B}_{1-k,k,-1/6}\mathcal{B}_{-k-2,k+1,5/6}}
\nonumber\\
&\qquad= -\mathcal{B}_{-k,k,5/6}\cbr{\mathcal{B}_{-k-1,k,5/6}\mathcal{B}_{-k-1,k+2,-1/6}+\mathcal{B}_{-k,k+1,-1/6}\mathcal{B}_{-k-2,k+1,5/6}},\\
&t\mathcal{B}_{-k-1,k+1,5/6}\cbr{\mathcal{B}_{-k-1,k,5/6}\mathcal{B}_{1-k,k+1,-1/6}+\mathcal{B}_{-k,k-1,5/6}\mathcal{B}_{-k,k+1,-1/6}}
\nonumber\\
&\qquad= \mathcal{B}_{-k,k,5/6}\cbr{\mathcal{B}_{-k-1,k,5/6}\mathcal{B}_{-k,k+1,-1/6}+\sbr{k+\tfrac{2}{3}}\mathcal{B}_{1-k,k,-1/6}\mathcal{B}_{-k-2,k+1,5/6}},\\
&t\mathcal{B}_{-k-1,k,5/6}\cbr{\mathcal{B}_{-k-1,k+1,5/6}\mathcal{B}_{-k-1,k+2,-1/6}+\sbr{k+\tfrac{4}{3}}\mathcal{B}_{-k,k+1,-1/6}\mathcal{B}_{-k-2,k+2,5/6}}
\nonumber\\
&\qquad= \mathcal{B}_{-k-2,k+1,5/6}\cbr{\sbr{k+\tfrac{2}{3}}\mathcal{B}_{-k-1,k+2,-1/6}\mathcal{B}_{-k,k,5/6}+\sbr{k+1}\mathcal{B}_{-k,k+1,-1/6}\mathcal{B}_{-k-1,k+1,5/6}}.
\end{align}
After a gauge transformation (which depends on $n\bmod 3$), to match up the $\tau$-function $\upupsilon_n$ in \eqref{taups} with an appropriate Wronskian, each of the latter equations is equivalent to the trilinear equation 
in \cite{ahk}.

\subsection{Unique positive solution: finale}

The preceding results on repeated application of BTs for Painlev\'e V show that these generate a  
solution 
of dP$_{\rm I}$ in the case that one initial value $v_{-1}=0$, while $v_0$ is arbitrary. Indeed, for any choice of $v_0$, there is a value of the ratio $\la=C_1/C_2$ in (\ref{vn}) which provides a complete solution 
of the difference equation (\ref{dPI}) in terms of ratios of modified Bessel functions, and this is equivalent to \eqref{dpi} 
with $t=\tfrac{1}{3\eps}$. Note that if we rearrange the formula for $v_0$ in (\ref{vn}) 
as 
\beq\label{vlambda} 
2v_0+1 = 
\frac{K_{5/6}(\tfrac{1}{2}t)-\la I_{-5/6}(\tfrac{1}{2}t)}
{K_{1/6}(\tfrac{1}{2}t)+\la I_{1/6}(\tfrac{1}{2}t)}, \qquad 
t=\frac{1}{3\eps}, 
\eeq 
then for any choice of $v_0$ we can invert the M\"obius transformation above to find $\lambda$ in terms of $v_0$ and $\eps$. 
So for each fixed $\eps$ there is a one-to-one correspondence between 
the choice of $v_0$ and the choice of parameter $\la$. 
However, we can characterize one particular solution by its 
distinct asymptotic behaviour. 

\begin{propn}\label{uniqueasy}
The function 
\beq\label{v0fn}
v_0 (\eps) = \tfrac{1}{2}\left(\frac{K_{5/6}(\tfrac{1}{6\eps})}{K_{1/6}(\tfrac{1}{6\eps})}-1\right)
\eeq 
is the unique initial condition for \eqref{dpi} that has the asymptotic behaviour 
\eqref{vasyo} as $\eps\to 0$. 
\end{propn}
\begin{prf} 
From the leading order asymptotics of the modified Bessel functions, that is 
$$  
K_\nu (\tfrac{1}{2}t)\sim \sqrt{\frac{\pi}{t}}\, \exp(-\tfrac{1}{2}t), \qquad
I_\nu (\tfrac{1}{2}t)\sim {\frac{1}{\sqrt{\pi t}}}\, \exp(\tfrac{1}{2}t)
\qquad \mathrm{as} \quad t\to\infty 
$$
we see that the right-hand side of \eqref{vlambda} tends to 1 as $t\to\infty$ when $\la=0$, but otherwise it tends to $-1$, Equivalently,  if  $\la=0$ then $v_0\to 0 $ as $\eps\to 0$, but 
for all $\la\neq 0$ this ratio of modified Bessel functions gives 
$v_0\to -1 $ as $\eps\to 0$. Hence the function 
\eqref{v0fn} is the only member of this one-parameter family that 
is compatible with the asymptotic behaviour \eqref{vasyo} as $\eps\to 0$. Since all of the functions $v_0$ given by (\ref{vn}) satisfy 
the Riccati equation \eqref{ric}, the latter series can be obtained by 
substituting in $v_0\sim\sum_{i=0}^\infty (-1)^i s_{0,i}\eps^{i+1}$. 
This immediately yields the recursion 
$$
s_{0,i+1}=(3i+1)s_{0,i}+\sum_{j=0}^i s_{0,i-j}s_{0,j} \quad\mathrm{for}\,\,i\geq 0, \quad \mathrm{with} \,\, s_{0,0}=1, 
$$
producing the sequence $1,2,12,112,1392,$ etc.
\end{prf}

The computation of quotients of modified Bessel functions is an important problem in numerical analysis \cite{onoe}, and continued fraction methods provide effective tools for doing this \cite{dhs, amos}. For the function  \eqref{v0fn}, the continued fraction 
expansion \eqref{v0frac} can be calculated directly from the Riccati equation \eqref{ric}, which is one among a family that includes many examples 
first considered in the pioneering works of Euler and Lagrange (see Chapter II in \cite{khov}). 

If we set 
\beq\label{1strics}
\eta_0=v_0, \qquad \eta_1=v_1 , \qquad \eta_2=v_0+v_2, 
\eeq 
then we see that the iteration of (\ref{dpi}) with $v_{-1}=0$ is 
consistent with the recursion 
\beq\label{ricrec}
\eta_n =\frac{\xi_n\eps}{1+\eta_{n+1}}, \qquad n\geq 0, 
\eeq 
and this generates a continued fraction representation for $v_0$ in the form 
\beq\label{thefrac}
v_0
= \cfrac{\xi_0\eps}{ 1 +\cfrac {\xi_1\eps}{ 1+\cfrac{\xi_2\eps}{1+ \cdots} } }.  
\eeq 
At the same time, given that $v_0$ is a solution of (\ref{ric}), 
it follows by induction that each $\eta_n$ satisfies a Riccati 
equation, namely 
\beq\label{etaric}
3\eps^2 \frac{\rd \eta_n}{\rd \eps}+\eta_n^2 
+(1-\ze_n\eps)\eta_n-\xi_n\eps =0, \qquad n\geq 0, 
\eeq 
provided that 
$\zeta_{n+1}  =3-\ze_n$, 
$\xi_{n+1}=\xi_n + \ze_{n+1}$. 
Then we require $\xi_0=1$ and $\ze_0=2$ from (\ref{ric}), which 
implies $\xi_1=2$ and $\ze_1=1$, in agreement with (\ref{ric2}), 
and hence the continued fraction \eqref{thefrac}
and its associated sequence of Riccati equations \eqref{etaric} 
are completely specified by
\beq\label{fracdata}
\xi_{2m} = 3m+1, \quad \ze_{2m}=2 \qquad \mathrm{and} \qquad
\xi_{2m+1} = 3m+2, \quad \ze_{2m+1}=1  \qquad \mathrm{for}\,\, m\geq 0, 
\eeq 
which reveals the pattern in \eqref{v0frac}. 

\begin{remark}
Upon taking the difference of the 
equations (\ref{etaric}) for $n=0$ and $n=2$, and using $v_2=\eta_2-\eta_0$, we find that 
$v_2$ also satisfies a  
Riccati equation, that is 
$$
3\eps^2 \frac{\rd v_2}{\rd \eps}+v_2^2 
+(1-2\eps+2v_0)v_2-3\eps =0,
$$
which has $v_0$ appearing among its coefficients. Similarly, it 
is possible to use (\ref{dpi}) to show by induction that all $v_n$ for $n\geq 2$ satisfy Riccati equations 
with 
$v_j$ for $j\leq n-2$ 
included in their coefficients (cf.\ \cite{hoppe}). 
\end{remark} 

The continued fraction \eqref{thefrac} for $\eta_0=v_0$ 
thus obtained has a sequence of convergents 
$ \bar{\eta}_0^{(k)} =P^{(k)}/Q^{(k)}$, $k\geq 0$, 
which correctly approximate the first $k+1$ non-zero terms in the series 
expansion, 
so 
$
\bar{\eta}_0^{(k)} =\sum_{i=0}^k (-1)^i s_{0,i}\eps^{i+1} 
+ O(\eps^{k+2}) 
$,  
and the numerators and denominators  are polynomials in $\eps$ generated by the same three-term relation 
%
$$
P^{(k+1)} =   P^{(k)}+ \xi_{k+1}\eps\, P^{(k-1)} , \qquad Q^{(k+1)} = Q^{(k)}+ \xi_{k+1}\eps\, Q^{(k-1)}, 
$$
with initial conditions $P^{(-2)}=1$, $P^{(-1)}=0$, $Q^{(-2)}=0$, $Q^{(-1)}=1$. 
Standard theory \cite{khov} then implies that with the coefficients 
$\xi_n$ as above, the continued fraction is convergent for all $\eps>0$, 
being equal to the alternating sum 
$$ 
\eta_0 =  \frac{P^{(0)}}{Q^{(0)}} +\sum_{k=1}^\infty 
(-1)^k\frac{\xi_0\xi_1\cdots \xi_k \,\eps^{k+1}}{Q^{(k-1)}Q^{(k)}}
=\lim_{k\to\infty}\bar{\eta}_0^{(k)} .
$$
However, the continued fraction is formally divergent at $\eps=0$, which corresponds to fact that the series \eqref{vasyo} is divergent. 
Thus we see that the continued fraction \eqref{v0frac} represents the 
function $\eta_0=v_0$ in \eqref{v0fn} for all $\eps\in (0,\infty)$, and hence this function is positive on the whole positive semi-axis. This provides a 
much stronger characterization of this function than the asymptotic 
one in Proposition \ref{uniqueasy}, namely the 
\begin{cor}\label{posric}
The function \eqref{v0fn} is the unique solution of the Riccati equation \eqref{ric} that is positive for all $\eps>0$. 
\end{cor}

We finally return  to the fixed point method considered in section 2, and the upper/lower bounds on the iterates of the mapping $T$. 
It turns out that the bound $\rb_0^{(k)}$ and the convergent 
$ \bar{\eta}_0^{(k)}$ both approximate the asymptotic series \eqref{vasyo} correctly to the same order $\eps^{k+1}$, but the former 
is a better approximant than the latter in the sense that its coefficient at order $\eps^{k+2}$ is closer to the correct value. 
This leads to a more precise statement in terms of inequalities, as follows. 
\begin{propn}\label{interlace} 
The convergents of the continued fraction for the function 
\eqref{v0fn} interlace with the upper/lower bounds obtained for the mapping $T$ in \eqref{tdef},  according to 
\beq\label{lace} 
{\rb}_0^{(2j-1)}< \bar{\eta}_0^{(2j+1)} \leq \rb_0^{(2j+1)} 
<\rb_0^{(2j+2)} \leq \bar{\eta}_0^{(2j+2)} < \rb_0^{(2j)} \qquad \mathrm{for}\,\,\mathrm{for}\,\mathrm{all} \quad j\geq 0.
\eeq 
\end{propn}
\begin{prf}
The middle inequality in (\ref{lace})     was already shown as 
part of Lemma \ref{abounds}, so the main new content of the statement 
above can be concisely paraphrased as 
\beq\label{klace}
(-1)^k{\rb}_0^{(k)}\leq (-1)^k \bar{\eta}_0^{(k)} < (-1)^k \rb_0^{(k-2)} \qquad \mathrm{for}\,\, k\geq 1. 
\eeq 
This is proved by induction on $k$, via a comparison of two different expressions for $v_0$: the first is the standard 
continued fraction \eqref{v0frac}, which generates  the sequence of  convergents $\bar{\eta}_0^{(k)}$; while the second is the structure of iteration of 
\eqref{dpi}, and the action of the mapping $T$, which generates another sequence of rational approximants 
${\rb}_0^{(k)}$, obtained from $v_0$ given as a kind of branched continued fraction: 
\beq\label{branched}
v_0 = \cfrac{\eps}{1+
\cfrac{2\eps}{1+
\cfrac{\eps}{1+
\cfrac{2\eps}{1+\cdots}
} 
+ \cfrac{3\eps}{1+
\cfrac{2\eps}{1+\cdots} +\cfrac{4\eps}{1+\cdots}
}
}
}\,\, . 
\eeq 
For the induction, observe that truncation at level $k=0$ in each fraction gives the same approximant 
$ \bar{\eta}_0^{(0)}= {\rb}_0^{(k)}=\eps$, 
and also at levels $k=1$ and $k=2$ we have 
 $\bar{\eta}_0^{(1)}= {\rb}_0^{(1)}=\frac{\eps}{1+2\eps}$, 
 $\bar{\eta}_0^{(2)}= {\rb}_0^{(2)}=\frac{\eps(1+4\eps)}{(1+6\eps)}$, 
so by \eqref{bkbds} with $n=0$ it follows that \eqref{klace} holds 
for the base cases $k=1,2$; but for $k\geq 3$ all of the inequalities 
in \eqref{klace} become strict. For the induction, we can consider 
the sequence of convergents $\bar{\eta}_1^{(k)}$ of the standard S-fraction for $v_1$, that is 
$v_1=2\eps/(1+4\eps/(1+5\eps/(1+...)))$
so that we have 
$\bar{\eta}_0^{(k+1)}=\eps/(1+\bar{\eta}_1^{(k)})$. 
while from the action of $T$ we have 
${\rb}_0^{(k+1)}=
T({\rb}_0^{(k)})= \eps/(1+{\rb}_1^{(k)})$.
Hence it follows that \eqref{klace} holds by induction, provided that at the next level we have the analogous inequalities 
\beq\label{klace2}
(-1)^k{\rb}_1^{(k)}\leq (-1)^k \bar{\eta}_1^{(k)} < (-1)^k \rb_1^{(k-2)} \qquad \mathrm{for}\,\, k\geq 1. 
\eeq 
For instance, if \eqref{klace2} holds for some even $k=2j$, 
then 
$$
\frac{\eps}{1+\rb_1^{(2j-2)}}<\frac{\eps}{1+\bar{\eta}_1^{(2j)}}
\leq \frac{\eps}{1+\rb_1^{(2j)}}, 
$$
which is precisely  \eqref{klace} for $k=2j+1$; and the reasoning 
is the same starting from \eqref{klace2} with odd $k$, but with 
the inequalities reversed. 

Of course, this begs the question of the validity of \eqref{klace2}, which must be verified by going  
down one more level  and considering 
$$
v_1 = \frac{2\eps}{1+\eta_2} = \frac{2\eps}{1+v_0+v_2},
$$ 
for which the leading order truncation gives 
$\bar{\eta}_1^{(0)}=\rb_1^{(0)}=2\eps$, 
while subsequent truncations 
require  comparison of 
\beq\label{stage2}
\eta_2 = \cfrac{4\eps}{1+\cfrac{5\eps}{1+\cdots}} \quad 
\mathrm{and} \quad v_0 + v_2 = \cfrac{\eps}{1+\cfrac{2\eps}{1+\cdots}} + 
\cfrac{3\eps}{1+\cfrac{2\eps}{1+\cdots}+\cfrac{4\eps}{1+\cdots}}
\eeq 
at the next stage. It is clear that the leading order truncation 
in \eqref{stage2} has $4\eps = \eps +3\eps$, which in turn shows 
that $\bar{\eta}_1^{(1)}=\rb_1^{(1)}$, confirming \eqref{klace2} for $k=1$,  but for the next order 
comparison it is required that 
\beq\label{aconvex}
\frac{4\eps}{1+5\eps}<
\frac{\eps}{1+2\eps}+
\frac{3\eps}{1+6\eps}. 
\eeq 
The latter is just a particular case of the general inequality
$$
\frac{A}{1+B}<
\frac{A_1}{1+B_1}+
\frac{A_2}{1+B_2},\qquad \mathrm{for}\quad A=A_1+A_2, \quad AB=A_1B_1+A_2B_2, \quad A_1,A_2,B_1,B_2>0,
$$   
which holds 
due to the convexity of the function $1/(1+x)$. Then 
\eqref{aconvex} implies that $2\eps=\rb_1^{(0)}>\bar{\eta}_1^{(2)} > \rb_1^{(2)}$, establishing \eqref{klace2} for $k=2$. Subsequent upper/lower bounds follow similarly, by repeatedly applying the same convexity argument to compare each lower stage of the two 
continued fractions.
\end{prf}

We can now present the final steps of the proof of uniqueness of the 
positive solution, and conclude with the proof of the main theorem. 

\begin{proof}[Proof of Theorem \ref{positive}] 
Taking  limits of 
the middle three inequalities in \eqref{klace} 
gives 
$$ 
\lim_{j\to\infty}\bar{\eta}_0^{(2j+1)} \leq \lim_{j\to\infty} \rb_0^{(2j+1)} 
\leq \lim_{j\to\infty}\rb_0^{(2j+2)} \leq \lim_{j\to\infty}\bar{\eta}_0^{(2j+2)}. 
$$
Then the convergence of the continued fraction \eqref{v0frac} gives the 
equality of limits of the upper and lower sequences of convergents, that is 
\beq\label{v0limit}
\lim_{j\to\infty}\bar{\eta}_0^{(2j+1)} = v_0 = \lim_{j\to\infty}\bar{\eta}_0^{(2j+2)}
\eeq 
which in turn gives 
$$\lim_{j\to\infty} \rb_0^{(2j+1)} 
= \lim_{j\to\infty}\rb_0^{(2j+2)}\implies 
\lim_{j\to\infty} \rho_0^{(2j+1)} 
= \lim_{j\to\infty}\rho_0^{(2j)} \implies 
\lim_{j\to\infty}\Delta_{0}^{(2j)}=0 
=\lim_{j\to\infty}\Delta_{0}^{(2j+1)}.$$
Now taking the limit $k\to\infty$ in the $n=0$ case of \eqref{delid} 
produces 
$$
\lim_{k\to\infty}\Delta_{0}^{(k+1)} =2\eps 
\Big(\lim_{k\to\infty}\rho_0^{(k)}
\Big)^2\lim_{k\to\infty}\Delta_{1}^{(k)} \implies \lim_{k\to\infty}\Delta_{1}^{(k)}=0, 
$$
and hence, by induction on $n$, repeated application of  
\eqref{delid} yields 
$\lim_{k\to\infty}\Delta_{n}^{(k)}=0$ for all $n\geq 0$.
From the result of Proposition \ref{existence} we see that, for each $n$ the upper 
and lower bounds in \eqref{vbbds} coincide, giving the unique positive solution $\bv=(v_n)_{n\geq 0}$, as required. 
\end{proof} 

\begin{proof}[Proof of Theorem \ref{mainthm}]
It remains to point out that for each $\eps>0$, the limit \eqref{v0limit} obtained from the continued fraction \eqref{v0frac} is precisely the function $v_0(\eps)$ given by \eqref{v0fn}. 
For the other $v_n$, note from \eqref{eqn:KummerMtoBesselI} that setting $\la=C_1/C_2=0$ in \eqref{vlambda} is equivalent to taking 
$d_1/d_2=-1$ in \eqref{eqn:qnrecrelapplied}. Thus, without loss of generality, by fixing $d_1=1=-d_2$ in Theorem \ref{wrdpi},  for each $n$ we obtain the  explicit expression for 
$v_n(\eps)>0$ in terms of ratios of Wronskian determinants. 
\end{proof}

\section{Conclusions}
\setcounter{equation}{0} 

We have shown that the quantum minimal surface obtained from a pair of operators satisfying the equation for a parabola, $Z_2=Z_1^2$, admits an exact solution in terms of modified Bessel functions, where the positive solution of the associated discrete Painlev\'e I equation corresponds to a particular sequence of classical solutions of the continuous Painlev\'e V equation with specific parameter values. The key to finding this exact solution was to 
use the complex geometry of the discrete Painlev\'e equation, constructing the associated Sakai surface, which identified 
the dP$_{\rm I}$ equation with the action of a quasi-translation on the space of initial conditions for  Painlev\'e V. 
Once the appropriate parameters for the 
 Painlev\'e V equation had been found, this enabled us to compare with known results on classical solutions, and match these up the initial conditions for the dP$_{\rm I}$ equation, 
 which identified the unique positive solution. While previous results in the literature have expressed these classical solutions in terms of Whittaker functions (or equivalently, Kummer functions), 
 some current work in progress (by two of us in collaboration with Dunning) has allowed the unique positive solution to be expressed with modified Bessel functions, which are a special case of Kummer functions and Whittaker functions, cf.~\cite[\S13.6, \S13.18]{DLMF}.

 It is interesting to note that other instances of classical solutions of Painlev\'e equations have appeared in the recent literature, providing the unique solutions of discrete Painlev\'e equations that satisfy positivity or other special initial boundary value problems  \cite{cla, latimer, vanassche}. These particular unique solutions seem to arise in specific application areas, such as random matrices and orthogonal polynomials, but it would be worthwhile to see if they can be characterized in some other way (geometrically, for instance). In fact, 
the asymptotics of oscillatory solutions of certain dP$_{\rm I}$ equations, including \eqref{dpi},  were recently considered in \cite{apt}. Such solutions are known to   arise 
 from a growth problem defined by a normal random matrix ensemble \cite{teod}, so it is natural to wonder whether 
 the unique positive solution of  \eqref{dpi} has an 
 interpretation in that context.  

 A recent preprint by Hoppe  \cite{hoppe} includes some comments on Whittaker function expressions for $v_0$ and $v_2$, which are equivalent to the ones we have found. The latter work raises the question of whether similar results should apply for other quantum minimal surfaces from rational curve equations of the form $Z_2^r = Z_1^s$ for positive integers $r<s$ (with $\gcd(r,s)=1$), 
following a remark made at the end of \cite{ahk}, where it is suggested that these curves should also give rise to discrete integrable systems. Indeed, the condition 
\eqref{zzdag} 
 gives rise to a difference equation for $v_n$, which (after integrating) becomes an equation of order $2(s-1)$: this should be a discrete Painlev\'e equation of higher order. Some more details of the example $(r,s)=(1,3)$ are considered in \cite{hoppe}, where the difference equation in question is 
 \beq\label{4thorderP}
 v_n ( v_{n+1}v_{n+2} + v_{n-1}v_{n+1}+v_{n-2}v_{n-1}+1) = \eps (n+1).  
 \eeq 
 Preliminary investigations show that this equation admits a positive solution with analogous properties to the case of the quantum parabola considered here: since the initial acceptance of this paper, the preprint \cite{FH} has appeared,  which 
 relates the family of quantum $(1,s)$ curves to orthogonal polynomials with complex densities,  recovers our expression for $v_0$ when $s=2$, and yields formulae for all $v_n$ as a (manifestly positive) ratio of integrals.  
 This and the other $(r,s)$ curves are an interesting subject for further study: while various higher order analogues of discrete Painlev\'e equations have been considered, there is currently no version of the Sakai correspondence in dimension greater than two. 
 
\small
\bigskip\noindent{\bf Acknowledgments:} 
PC and BM acknowledge simulating discussions with Clare Dunning. 
AD is grateful to Alexander Stokes for many stimulating discussions, and to BIMSA for the financial support through the startup grant, and for creating an ideal environment for doing research. 
AH  is grateful to Jens Hoppe, Bas Lemmens and Ana Loureiro for insightful conversations, and thanks the following 
institutions: 
University of Northern Colorado for supporting his visit 
in October 2023; BIMSA for supporting his attendance at the RTISART meeting in July 2024; and the Graduate School of Mathematical Sciences, University of Tokyo, for the invitation to speak at the FoPM 
Symposium in February 2025, where some 
of these results were first presented.
BM is grateful to EPSRC for the support of a postgraduate studentship. The authors also thank an anonymous referee for some helpful comments and suggestions for additional references that were missing from the original submission, which improved the final version.

\noindent
Conflict of Interest: The authors declare that they have no
conflicts of interest.

\def\fit#1{\textit{\frenchspacing#1}}


\begin{thebibliography}{99}

\bibitem{ach} J. Arnlind, J. Choe and J. Hoppe, 
Noncommutative minimal surfaces, \fit{Lett. Math. Phys.}
\textbf{106} 
(2016) 1109--1129. 

\bibitem{ikkt} N. Ishibashi, H. Kawai, Y. Kitazawa and A. Tsuchiya,
A large-$N$ reduced model as superstring,
\fit{Nuclear Phys. B} \textbf{498} (1997) 467--491. 

\bibitem{ct} L. Cornalba and W. Taylor IV, 
Holomorphic curves from matrices, 
\fit{Nuclear Phys. B} \textbf{536} (1999) 513--552. 

\bibitem{eisenhart} 
L.P. Eisenhart, 
Minimal surfaces in Euclidean four-space, 
\fit{Amer. J. Math.} \textbf{34} 
(1912) 215--236.

\bibitem{ahk} J. Arnlind, J. Hoppe and M. Kontsevich, 
Quantum minimal surfaces,
{\tt arXiv:1903.10792v1}.

\bibitem{rg} A. Ramani and B. Grammaticos, 
Discrete \p\ equations: coalescence, limits and degeneracies, 
\fit{Physica A} \textbf{228} (1996) 160--171. 

\bibitem{grp}B. Grammaticos, A. Ramani and V. Papageorgiou, 
Discrete dressing transformations and \p\ equations, 
\fit{Phys. Lett. A} \textbf{235} (1997) 475--479.

\bibitem{tgr} 
T. Tokihiro, B. Grammaticos and A. Ramani, 
From the continuous P$_{\rm V}$ to discrete \p\ equations, 
\fit{J. Phys. A: Math. Gen.} \textbf{35} (2002) 5943--5950.

\bibitem{refFIK}
A.S. Fokas, A.R. Its and A.V. Kitaev, 
Discrete \peqs\ and their appearance in quantum gravity, 
\fit{Commun. Math. Phys.} \textbf{142} (1991) 313--344. 

\bibitem{magnus}
A.P. Magnus, Painlev\'{e}-type differential equations for the recurrence coefficients of semi-classical orthogonal polynomials, 
\fit{J. Comput. Appl. Math.} \textbf{57} (1995) 215--237.

\bibitem{refFGR}
A.S. Fokas, B. Grammaticos and A. Ramani, 
From continuous to discrete \peqs,
\fit{J. Math. Anal. Appl.} \textbf{180} (1993) 342--360.

\bibitem{refCMW}
P.A. Clarkson, E.L. Mansfield and H.N. Webster, On the relation between the continuous and discrete \peqs,
\fit{Theo.\ Math.\ Phys.} \textbf{122} (2000) 1--16. 

\bibitem{refAS} V.E. Adler and A.B. Shabat, Some exact solutions of the Volterra lattice,
\fit{Theoret. Math. Phys.}, \textbf{201} (2019) 1442--1456.

\bibitem{Sak:2001:RSAWARSGPE}
H. Sakai, 
Rational surfaces associated with affine root systems and
 geometry of the \p\ equations, 
\fit{Comm. Math. Phys.} \textbf{220}
 (2001) 
 165--229. 

\bibitem{Oka:1979:FAESOPCFPP} 
K. Okamoto, {Sur les feuilletages associ\'{e}s aux \'{e}quations du
  second ordre \`a points critiques fixes de {P}. {P}ainlev\'{e}}, \fit{Japan. J.
  Math. (N.S.)} \textbf{5} (1979) 
  1--79. 

\bibitem{KajNakTsu:2011:PRDPSTA}
K. Kajiwara, N. Nakazono and  T. Tsuda,  Projective
  reduction of the discrete {P}ainlev\'{e} system of type {$(A_2+A_1)^{(1)}$}.
  \emph{Int. Math. Res. Not. IMRN} \textbf{2011}, \textit{4}, 930--966. 
 
\bibitem{fw} 
P.J. Forrester and N.S. Witte, 
Application of the $\tau$-function theory
of \p\ equations to random matrices:
P$_{\rm V}$, P$_{\rm III}$, the LUE, JUE, and CUE, 
\fit{Comm. Pure Appl. Math.} \textbf{55} (2002) 679--727.

\bibitem{masuda} T. Masuda, 
Classical transcendental solutions of the \p\ equations and their degeneration, 
\fit{Tohoku Math. J.} \textbf{56} (2004) 467--490. 

\bibitem{grsw}B. Grammaticos, A. Ramani, J. Satsuma and R. Willox, 
Discretising the \p\ equations \`a la Hirota-Mickens,
\fit{J. Math. Phys.} \textbf{53} (2012) 023506.

\bibitem{grw} 
B. Grammaticos, A. Ramani and R. Willox,
Restoring discrete \p\ equations from an ${\rm E}_8^{(1)}$ associated one,
\fit{J. Math. Phys.} \textbf{60} (2019) 063502.

\bibitem{KajNouYam:2017:GAPE} K. Kajiwara, M. Noumi and Y. Yamada, Geometric aspects of \p\ equations, \fit{J. Phys. A: Math. Theor.} \textbf{50} (2017) 073001.

\bibitem{Sak:2007:PDPETLF} H. Sakai, 
Problem: discrete \p\ equations and their Lax forms,
\textit{Algebraic, analytic and geometric aspects of complex differential equations
and their deformations. \p\ hierarchies}, RIMS 
K$\^{\mathrm{o}}$ky$\^{\mathrm{u}}$roku 
Bessatsu, B2, Res. Inst. Math. Sci. (RIMS), Kyoto, 2007, 195--208. 

\bibitem{CheDzhHu:2020:PLUEDPE}
Y. Chen, A. Dzhamay and J. Hu, 
Gap {P}robabilities in the
  {L}aguerre {U}nitary {E}nsemble and discrete {P}ainlev\'e equations, 
  \fit{J. Phys. A} \textbf{53} (2020) 354003. 

\bibitem{DzhFilSto:2020:RCDOPWHWDPE}
A. Dzhamay, G. Filipuk and A. Stokes,  Recurrence
coefficients for discrete orthogonal polynomials with hypergeometric weight
and discrete {P}ainlev\'{e} equations, \fit{J. Phys. A} \textbf{53} (2020)
 495201. 

\bibitem{DLMF}
F.W.J. Olver, A.B. {Olde Daalhuis}, D.W. Lozier, B.I. Schneider, R.F. Boisvert, C.W. Clark, B.R. Miller, B.V. Saunders, H.S. Cohl and M.A. McClain (Editors),
NIST Digital Library of Mathematical Functions, 
\url{https://dlmf.nist.gov/}, Release 1.2.4 (15 March 2025).

\bibitem{SaiTer:2004:NCRSPE}
M.-H. Saito 
and H. Terajima, 
Nodal curves and {R}iccati solutions of {P}ainlev\'{e} equations,
\textit{J. Math. Kyoto Univ.}, \textbf{44} (2004) 
529--568. 

\bibitem{refGF}
V.I. Gromak and G. Filipuk,
The \bk\ transformations of the fifth \p\ equation and their applications,
\fit{Math. Model. Anal.}   \textbf{6} (2001) 221--230.

\bibitem{refGLS}
V.I. Gromak, I. Laine and S. Shimomura,
\textit{\p\ Differential Equations in the Complex Plane}, \emph{Studies in Math.}, vol.\ \textbf{28}, de Gruyter, Berlin, New York (2002).

\bibitem{hoppe} J. Hoppe, Quantum states from minimal surfaces, 
{\tt arXiv:2502.18422} 

\bibitem{determinantsveindale} R. Vein and P. Dale,
\textit{Determinants and Their Applications in Mathematical Physics},
Springer, New York (1999).

\bibitem{refDod}
C.L. Dodgson, Condensation of determinants, being a new and brief method for computing their arithmetical values,
\textit{Proc.\ R.\ Soc. Lond.} \textbf{15} (1866) 150--155. 

\bibitem{onoe} M. Onoe, Modified quotients of cylinder functions, \fit{Math. Tables Aids Comput.} \textbf{10} (1956) 27-–28. 

\bibitem{dhs} C.D. Ahlbrandt and A.C. Peterson, 
\textit{Discrete Hamiltonian Systems: Difference Equations, Continued Fractions, and Riccati Equations}, \fit{Kluwer Texts Math. Sci.}, vol.\ \textbf{16}, Springer, New York (1996). 

\bibitem{amos}  D.E. Amos, 
Computation of modified Bessel functions and their ratios, 
\fit{Math. Comp.} \textbf{28}  
(1974) 239--251. 

\bibitem{khov} A.N. Khovanskii, 
\textit{The Applications of Continued Fractions and their Generalizations 
to Problems in Approximation Theory} (transl. P. Wynn), 
P. Noordhoff N. V. , Groningen (1963).

\bibitem{cla} 
P.A. Clarkson, A.F. Loureiro and W. Van Assche, Unique positive solution for an alternative discrete \p\ I equation, \fit{J. Difference Equ. Appl.} \textbf{22} (2016) 656--675. 

\bibitem{latimer} 
T. Lasic Latimer, 
Unique positive solutions to $q$-discrete equations associated with orthogonal polynomials, \fit{J. Difference Equ. Appl.} \textbf{27} (2021) 763--775. 

\bibitem{vanassche} W. Van Assche, 
Unique special solution for discrete \p\ II, 
\fit{J. Difference Equ. Appl.} \textbf{30} (2023) 465--474.

\bibitem{apt} A.I. Aptekarev and V.Y. Novokshenov,  Asymptotics of solutions of the discrete Painlev\'e I equation, \fit{Math Notes} \textbf{116} (2024) 1170--1182. 

\bibitem{teod} 
R. Teodorescu, E. Bettelheim, O. Agam, A. Zabrodin and P. Wiegmann, 
Normal random matrix ensemble as a growth problem,
\fit{Nuclear Phys. B}, \textbf{704} (2005) 407--444.

\bibitem{FH} 
G. Felder and J. Hoppe, Orthogonal polynomials with complex densities and quantum minimal surfaces, 
{\tt arXiv:2504.06197}  

\end{thebibliography}
\end{document}